\documentclass[11pt]{article}
\usepackage{amsmath}
\usepackage{amssymb, amsthm, amsfonts,bm,mathtools}
\usepackage[inline]{enumitem}
\usepackage{csquotes}
\usepackage{hhline}
\usepackage{cite}
\usepackage{framed}
\usepackage[framemethod=tikz]{mdframed}
\usepackage{graphicx}
\usepackage{color}
\usepackage{multirow}
\usepackage{algorithm, subcaption}
\usepackage[noend]{algpseudocode}
\usepackage{amssymb}
\usepackage{float}
\setitemize{noitemsep,topsep=3pt,parsep=3pt,partopsep=3pt}
\usepackage{xcolor}

\usepackage[margin = 1in]{geometry}
\usepackage{adjustbox}
\usepackage{xspace}
\usepackage{thmtools}
\usepackage{thm-restate}

\definecolor{darkgreen}{rgb}{0,0.5,0}
\definecolor{darkblue}{rgb}{0,0,0.8}
\usepackage{hyperref}
\hypersetup{
    unicode=false,          % non-Latin characters in Acrobat’s bookmarks
    colorlinks=true,        % false: boxed links; true: colored links
    linkcolor=darkblue,          % color of internal links (change box color with linkbordercolor)
    citecolor=darkgreen,        % color of links to bibliography
    filecolor=magenta,      % color of file links
    urlcolor= black           % color of external links
}
\RequirePackage[]{silence}
\usepackage[nameinlink,capitalize]{cleveref}

\newtheorem{theorem}{Theorem}[section]
\newtheorem{lemma}{Lemma}[section]

\newtheorem{definition}{Definition}[section]
\newtheorem{corollary}{Corollary}[section]

\newtheorem{observation}{Observation}[section]
\newtheorem{property}{Property}[section]

\newcommand{\vol}{\operatorname{\text{{\rm vol}}}}

\newcommand{\ignore}[1]{}

\algnewcommand\algorithmicswitch{\textbf{switch}}
\algnewcommand\algorithmiccase{\textbf{case}}

\algdef{SE}[SWITCH]{Switch}{EndSwitch}[1]{\algorithmicswitch\ #1\ \algorithmicdo}{\algorithmicend\ \algorithmicswitch}%
\algdef{SE}[CASE]{Case}{EndCase}[1]{\algorithmiccase\ #1}{\algorithmicend\ \algorithmiccase}%
\algtext*{EndSwitch}%
\algtext*{EndCase}%
\newcommand{\MCM}{{\sc mcm}}

\newcommand{\MWM}{{\sc mwm}}

\newcommand{\MIS}{\operatorname{MIS}}
\newcommand{\MaxIS}{{\sc MaxIS}}
\newcommand{\eps}{\varepsilon}

\newcommand{\congest}{\ensuremath{\mathsf{CONGEST}~}}
\newcommand{\local}{$\mathsf{LOCAL}$\xspace}

\newcommand{\poly}{\operatorname{\text{{\rm poly}}}}
\newcommand{\score}{\operatorname{\text{{\rm score}}}}

\newcommand{\polylog}{\operatorname{\text{{\rm polylog}}}}

\newcommand{\floor}[1]{\lfloor #1 \rfloor}

\newcommand{\ID}{\operatorname{ID}}
\newcommand{\dist}{\operatorname{dist}}

\newcommand{\mix}{\ensuremath{\tau_{\operatorname{mix}}}}
\newcommand{\Er}{E^\mathsf{r}}

\renewcommand{\paragraph}[1]{\vspace{0.15cm}\noindent {\bf #1}:}

\newcommand{\GG}{\mathcal{G}}
\newcommand{\HH}{\mathcal{H}}
\newcommand{\tw}{\operatorname{tw}}
\newcommand{\diam}{\operatorname{diam}}

\title{Narrowing the \local--\congest Gaps in Sparse Networks\\ via Expander Decompositions\thanks{A preliminary version~\cite{10.1145/3519270.3538423} of this paper was presented at the 2022 ACM Symposium on Principles of Distributed Computing (PODC), July 25--29, 2022, Salerno, Italy.}} 
\author{
Yi-Jun Chang
\\
\normalsize National University of Singapore \\
\normalsize \texttt{cyijun@nus.edu.sg} \\
\and
Hsin-Hao Su\\
\normalsize Boston College\\
\normalsize \texttt{hsinhao.su@bc.edu} \\
}	

\begin{document}
\date{}
\maketitle
\thispagestyle{empty}   
\setcounter{page}{0}

\begin{abstract}
Many combinatorial optimization problems, including maximum weighted matching and maximum independent set, can be approximated within $(1 \pm \epsilon)$ factors in  $\poly(\log n, 1/\epsilon)$ rounds in the \local model via network decompositions [Ghaffari, Kuhn, and Maus, STOC 2018]. These approaches, however, require sending messages of unlimited size, so they do not extend to the more realistic \congest model, which restricts the message size to be $O(\log n)$ bits.  For example, despite the long line of research devoted to the distributed matching problem,  it still remains a major open problem whether an $(1-\epsilon)$-approximate maximum weighted matching can be computed in $\poly(\log n, 1/\epsilon)$ rounds in the \congest model.

In this paper, we develop a generic framework for obtaining $\poly(\log n, 1/\epsilon)$-round $(1\pm \epsilon)$-approximation algorithms for many combinatorial optimization problems, including maximum weighted matching, maximum independent set, and correlation clustering, in graphs excluding a fixed minor in the \congest model. This class of graphs covers many  sparse network classes  that have been studied in the literature, including planar graphs, bounded-genus graphs, and bounded-treewidth graphs. 

Furthermore, we show that our framework can be applied to 
give 
an efficient distributed property testing algorithm for an arbitrary minor-closed graph property that is closed under taking disjoint union, significantly generalizing the previous distributed property testing algorithm for planarity in [Levi,  Medina, and Ron, PODC 2018 \& Distributed Computing 2021].

Our framework uses distributed expander decomposition algorithms [Chang and Saranurak, FOCS 2020] to decompose the graph into clusters of high conductance. We show that any graph excluding a fixed minor admits small edge separators. Using this result, we show the existence of a high-degree vertex in each cluster in an expander decomposition, which allows the entire graph topology of the cluster to be routed to a vertex. Similar to the use of network decompositions in the \local model, the vertex will be able to perform any local computation on the subgraph induced by the cluster and broadcast the result over the cluster.
%We show that this framework can be embedded into a variant of the sequential $(1-\epsilon)$-approximate maximum weighted matching algorithm of [Duan and Pettie, JACM 2014] to give an efficient distributed algorithm in the \congest model.
\end{abstract}
\maketitle
%%%%%%%%%%%%%%%%%%%%%%%%%%%%%%%%%%%%%%%%%%%%%%%%

\sloppy
\newpage
\thispagestyle{empty}   
\setcounter{page}{0}
\tableofcontents
\newpage

\section{Introduction}
The \local and  \congest models are two prominent vertex-centric models for studying distributed graph algorithms. In these models, vertices host processors and operate in synchronized rounds. In each round, each vertex sends a message to each of its neighbors, receives messages from its neighbors, and performs local computations. The time complexity of an algorithm is defined to be the number of rounds used. The main difference between the two models is the restriction on the message size. In the \local model, we allow messages of unlimited size to be sent across each link; while in the \congest model, an upper bound of $O(\log n)$ bits is imposed on the message size, where $n$ is the number of nodes.  Algorithms designed for the vertex-centric models can be optimized by Pregel-like systems \cite{MABDHLC10} such as GraphX \cite{GXDCFS14} and Gigraph \cite{Giraph} to process massive graph data, see \cite{MWM15} for a comprehensive survey. Since algorithms designed for the \congest model use smaller messages, it is likely they will be converted to more efficient processes than their counterparts in the \local model. 

Combinatorial optimization problems, such as matching and independent set, are central in the area of distributed graph algorithms.
Many combinatorial optimization problems are known to be efficiently solvable in the \local model. Ghaffari, Kuhn, and Maus \cite{GhaffariKM17} gave a general framework for approximating packing and covering integer linear programming problems within $(1 \pm \epsilon)$ of the optimality in $\poly(\log n, 1/\epsilon)$ rounds. The framework covers, for example, the maximum weighted matching problem and the maximum independent set problem.  With the recent breakthrough of Rozho\v{n} and Ghaffari~\cite{RozhonG20} on deterministic network decompositions, their approach can even be implemented deterministically. The approach of \cite{GhaffariKM17}, however, requires sending messages of unlimited size, so the complexities of many of these problems remain to be tackled in the \congest model. For example, in contrary to the \local model, it is still unclear whether a $\poly(\log n, 1/\epsilon)$-round  $(1-\epsilon)$-approximate algorithm for maximum weighted matching exists in the \congest model.  Moreover, it is known that some problems cannot be computed efficiently in the \congest model in general \cite{BachrachCDELP19,EfronGK20}. For example, there is a constant $\epsilon > 0$ such that finding an $(1-\epsilon)$-approximate maximum independent set requires $\tilde{\Omega}(n^2)$ rounds.

%\footnote{We note that although variants of \cite{RozhonG20} work in the \congest model, the main hurdle lies in \cite{GhaffariKM17}, which requires large messages.} 
%However, in the \local model, many algorithms with $\poly(1/\epsilon, \log n)$ running times are known \cite{Nieberg08,GhaffariKM17, GHK18,GhaffariKMU18, Harris19}. 

%An epitome of such problems is the maximum weighted matching (\MWM) problem. Despite a long line of research  devoted to this problem,  it still remains a major open problem whether a $(1-\epsilon)$-approximate algorithm with $\poly(\log n, 1/\epsilon)$ rounds exists, even for the easier unweighted maximum cardinality matching (\MCM) problem (see \cref{table:matching}). In contrast, in the \local model, many algorithms with $\poly(1/\epsilon, \log n)$ running times are known \cite{Nieberg08,GhaffariKM17, GHK18,GhaffariKMU18, Harris19}. 

\paragraph{Our Contribution}
In this paper, we develop a new tool set for solving combinatorial optimization problems in the \congest model on a wide range of sparse network classes that have been studied in the literature. Our framework applies to any graph classes that are minor closed, covering many natural graph classes such as planar graphs, bounded-genus graphs, and bounded-treewidth graphs.
 
 Our approach is as follows.
We use an expander decomposition to decompose the graph into components of high conductance. The existence of small edge separators guarantees the existence of a high-degree vertex in each component, which allows the entire graph topology of the component to be routed to a vertex. Similar to the use of the network decompositions in the \local model, the vertex will then be able to solve the problem locally and broadcast the result over the component. 

We show that our framework can be applied to give efficient algorithms to solve various combinatorial optimization problems, property testing problems, and graph decomposition problems in the \congest model, narrowing the gaps of these problems between the \congest model and the \local model in $H$-minor-free networks. 

%\footnote{It is necessary that we restrict our attention to a special graph class. For general graphs, an $(1\pm\epsilon)$-approximate solution for many problems cannot be computed efficiently in \congest \cite{BachrachCDELP19,EfronGK20}. For example, there is a constant $\epsilon > 0$ such that finding an $(1-\epsilon)$-approximate maximum independent set requires $\tilde{\Omega}(n^2)$ rounds.}

\paragraph{Notation} Throughout this paper, $n = |V|$ denotes the number of the vertices and $\Delta = \max_{v \in V} \deg(v)$ denotes the maximum degree of the graph $G=(V,E)$ under consideration. We say that an algorithm succeeds with high probability (w.h.p.)~if it succeeds with probability  $1  - 1/\poly(n)$. We write $\tilde{O}(\cdot)$, $\tilde{\Omega}(\cdot)$, and $\tilde{\Theta}(\cdot)$ to compress a $\log^{\pm O(1)} n$ factor.

\subsection{Our Results}

\paragraph{Matching} A {\it matching} is a set of edges that do not share endpoints. Given a weighted graph $G = (V,E,w)$, the maximum weight matching (\MWM{}) problem is to compute a matching $M$ with the maximum weight, where the weight of $M$ is defined as $\sum_{e \in M} w(e)$. Given an unweighted graph $G = (V,E)$, the maximum cardinality matching (\MCM{}) problem is to compute a matching $M$ such that $|M|$ is maximized. Clearly, the \MCM{} problem is a special case of the \MWM{} problem.  For \MWM, we assume that all the edge weights $w(e)$ are positive integers, and we write $W$ to denote the maximum weight $\max_{e \in E} w(e)$.

\begin{table}[htbp]\centering
\caption{Previous results on \MCM{} and \MWM{} in the \congest model and the \local model.  }\label{table:matching}
\begin{adjustbox}{width=\textwidth,center}
\begin{tabular}{llllll}
Citation                                          & Problem                                                                           & Ratio                            & Running Time                                                                                             & Type              & Model                        \\ \hline
\multicolumn{1}{|l|}{\cite{II86}}            & \multicolumn{1}{l|}{\MCM}                                                          & \multicolumn{1}{l|}{$\frac{1}{2}$}              & \multicolumn{1}{l|}{$O(\log n)$}                                                                         & \multicolumn{1}{l|}{Rand.} & \multicolumn{1}{l|}{\congest} \\ \hline
\multicolumn{1}{|l|}{\cite{ABI86} }                 & \multicolumn{1}{l|}{\MCM}                                                          & \multicolumn{1}{l|}{$\frac{1}{2}$}              & \multicolumn{1}{l|}{$O(\log n)$}                                                                         & \multicolumn{1}{l|}{Rand.} & \multicolumn{1}{l|}{\congest} \\ \hline
\multicolumn{1}{|l|}{\cite{Luby86}}                        & \multicolumn{1}{l|}{\MCM}                                                          & \multicolumn{1}{l|}{$\frac{1}{2}$}              & \multicolumn{1}{l|}{$O(\log n)$}                                                                         & \multicolumn{1}{l|}{Rand.} & \multicolumn{1}{l|}{\congest} \\ \hline
\multicolumn{1}{|l|}{\cite{HKP01}}           & \multicolumn{1}{l|}{\MCM}                                                          & \multicolumn{1}{l|}{$\frac{1}{2}$}              & \multicolumn{1}{l|}{$O(\log^4  n)$}                                                                         & \multicolumn{1}{l|}{Det.}  & \multicolumn{1}{l|}{\congest} \\ \hline
\multicolumn{1}{|l|}{\cite{WW04}} & \multicolumn{1}{l|}{\MWM}                                                          & \multicolumn{1}{l|}{$\frac{1}{5}$}              & \multicolumn{1}{l|}{$O(\log^2 n)$}                                                                       & \multicolumn{1}{l|}{Rand.} & \multicolumn{1}{l|}{\congest} \\ \hline
\multicolumn{1}{|l|}{\cite{LPR09}}               & \multicolumn{1}{l|}{\MWM}                                                          & \multicolumn{1}{l|}{$\frac{1}{4} - \epsilon$} & \multicolumn{1}{l|}{$O(\epsilon^{-1}\log \epsilon^{-1} \log n)$}                                         & \multicolumn{1}{l|}{Rand.}      & \multicolumn{1}{l|}{\congest} \\ \hline
\multicolumn{1}{|l|}{\multirow{3}{*}{\begin{tabular}[c]{@{}l@{}} \cite{LPP15}\end{tabular}}}                & \multicolumn{1}{l|}{\begin{tabular}[c]{@{}l@{}}\MCM\\  (bipartite)\end{tabular}} & \multicolumn{1}{l|}{$1-\epsilon$}     & \multicolumn{1}{l|}{$O(\log n/\epsilon^3)$}                                                              & \multicolumn{1}{l|}{Rand.} & \multicolumn{1}{l|}{\congest} \\ \cline{2-6}
\multicolumn{1}{|l|}{}                            & \multicolumn{1}{l|}{\MCM}                                                          & \multicolumn{1}{l|}{$1-\epsilon$}     & \multicolumn{1}{l|}{$2^{O(1/\epsilon)} \cdot O(   \epsilon^{-4} \log \epsilon^{-1} \cdot \log n)$} & \multicolumn{1}{l|}{Rand.}      & \multicolumn{1}{l|}{\congest} \\ \cline{2-6}
\multicolumn{1}{|l|}{}                            & \multicolumn{1}{l|}{\MWM}                                                          & \multicolumn{1}{l|}{$\frac{1}{2} - \epsilon$} & \multicolumn{1}{l|}{$O(\log(1/\epsilon) \cdot \log n)$}                                                  & \multicolumn{1}{l|}{Rand.} & \multicolumn{1}{l|}{\congest} \\ \hline
\multicolumn{1}{|l|}{\multirow{4}{*}{\begin{tabular}[c]{@{}l@{}} \cite{BCGS17} \end{tabular}}}          & \multicolumn{1}{l|}{\MWM}                                                          & \multicolumn{1}{l|}{$\frac{1}{2}$}              & \multicolumn{1}{l|}{$O(\log n \cdot \log W)$}                                                            & \multicolumn{1}{l|}{Rand.} & \multicolumn{1}{l|}{\congest} \\ \cline{2-6}
\multicolumn{1}{|l|}{}                            & \multicolumn{1}{l|}{\MWM}                                                          & \multicolumn{1}{l|}{$\frac{1}{2}$}              & \multicolumn{1}{l|}{$O(\Delta + \log n)$}                                                                & \multicolumn{1}{l|}{Det.}  & \multicolumn{1}{l|}{\congest} \\ \cline{2-6}
\multicolumn{1}{|l|}{}                            & \multicolumn{1}{l|}{\MWM}                                                          & \multicolumn{1}{l|}{$\frac{1}{2}-\epsilon$}   & \multicolumn{1}{l|}{$O(\log \Delta/ \log \log \Delta)$}                                                  & \multicolumn{1}{l|}{Rand.} & \multicolumn{1}{l|}{\congest} \\ \cline{2-6}
\multicolumn{1}{|l|}{}                            & \multicolumn{1}{l|}{\MCM}                                                          & \multicolumn{1}{l|}{$1-\epsilon$}     & \multicolumn{1}{l|}{$2^{O(1/\epsilon)}\cdot O(\log \Delta/ \log \log \Delta)$}                                                  & \multicolumn{1}{l|}{Rand.} & \multicolumn{1}{l|}{\congest} \\ \hline
\multicolumn{1}{|l|}{\multirow{2}{*}{\begin{tabular}[c]{@{}l@{}} \cite{Fischer17}\end{tabular}}}                         & \multicolumn{1}{l|}{\MCM}                                                          & \multicolumn{1}{l|}{$\frac{1}{2}$}              & \multicolumn{1}{l|}{$O(\log^2 \Delta \cdot \log n)$}                                                     & \multicolumn{1}{l|}{Det.}  & \multicolumn{1}{l|}{\congest} \\ \cline{2-6}
\multicolumn{1}{|l|}{}                            & \multicolumn{1}{l|}{\MWM}                                                          & \multicolumn{1}{l|}{$\frac{1}{2}-\epsilon$}   & \multicolumn{1}{l|}{$O(\log^2 \Delta \cdot \log \epsilon^{-1} + \log^{*} n)$}                              & \multicolumn{1}{l|}{Det.}  & \multicolumn{1}{l|}{\congest} \\ \hline
\multicolumn{1}{|l|}{\multirow{2}{*}{\begin{tabular}[c]{@{}l@{}} \cite{AKO18}\end{tabular}}}              & \multicolumn{1}{l|}{\begin{tabular}[c]{@{}l@{}}\MWM \\  (bipartite)\end{tabular}} & \multicolumn{1}{l|}{$1-\epsilon$}     & \multicolumn{1}{l|}{$O(\frac{\log(\Delta W )}{\epsilon^2} + \frac{\log^2 \Delta + \log^{*} n} {\epsilon})$}           & \multicolumn{1}{l|}{Det.}  & \multicolumn{1}{l|}{\congest} \\ \cline{2-6}
\multicolumn{1}{|l|}{}                            & \multicolumn{1}{l|}{\MWM}                                                          & \multicolumn{1}{l|}{$\frac{2}{3} - \epsilon$} & \multicolumn{1}{l|}{$O(\frac{\log(\Delta W )}{\epsilon^2} + \frac{\log^2 \Delta + \log^{*} n} {\epsilon})$}           & \multicolumn{1}{l|}{Det.}  & \multicolumn{1}{l|}{\congest} \\ \hline
\multicolumn{1}{|l|}{\cite{FFK21}}            & \multicolumn{1}{l|}{\MWM}                                                          & \multicolumn{1}{l|}{$1-\epsilon$}              & \multicolumn{1}{l|}{$2^{O(1/\epsilon)}\cdot \polylog(n)$}                                                                         & \multicolumn{1}{l|}{Det.} & \multicolumn{1}{l|}{\congest} \\ \hline
\multicolumn{1}{|l|}{\cite{FMU22}}            & \multicolumn{1}{l|}{\MCM}                                                          & \multicolumn{1}{l|}{$1-\epsilon$}              & \multicolumn{1}{l|}{$\poly(\log n, 1/\epsilon)$}                                                                         & \multicolumn{1}{l|}{Det.} & \multicolumn{1}{l|}{\congest} \\ \hline
\multicolumn{1}{|l|}{\cite{czygrinow2008fast}}                      & \multicolumn{1}{l|}{\begin{tabular}[c]{@{}l@{}}\MCM \\  (planar)\end{tabular}}                                                          & \multicolumn{1}{l|}{$1-\epsilon$}     & \multicolumn{1}{l|}{$O(\log(1/\epsilon)\cdot \log^{*} n)$}                                                                                    & \multicolumn{1}{l|}{Det.}  & \multicolumn{1}{l|}{\local}   \\ \hline
\multicolumn{1}{|l|}{\cite{CzygrinowHS09}}                      & \multicolumn{1}{l|}{\begin{tabular}[c]{@{}l@{}}\MCM \\  (bounded arb.)\end{tabular}}                                                          & \multicolumn{1}{l|}{$1-\epsilon$}     & \multicolumn{1}{l|}{$(1/\epsilon)^{O(1/\epsilon)}+O(\log^{*} n)$}                                                                                    & \multicolumn{1}{l|}{Det.}  & \multicolumn{1}{l|}{\local}   \\ \hline
\multicolumn{1}{|l|}{\cite{Nieberg08}}                     & \multicolumn{1}{l|}{\MWM}                                                          & \multicolumn{1}{l|}{$1-\epsilon$}     & \multicolumn{1}{l|}{$O(\epsilon^{-2} \log n \cdot T_{\MIS}(n^{O(1/\epsilon)}))$}                          & \multicolumn{1}{l|}{}      & \multicolumn{1}{l|}{\local}   \\ \hline
\multicolumn{1}{|l|}{\cite{BEPS16}}            & \multicolumn{1}{l|}{\MCM}                                                          & \multicolumn{1}{l|}{$\frac{1}{2}$}              & \multicolumn{1}{l|}{$O(\log \Delta + \log^{4} \log n)$}                                                  & \multicolumn{1}{l|}{Rand.} & \multicolumn{1}{l|}{\local}   \\ \hline
\multicolumn{1}{|l|}{\multirow{2}{*}{\begin{tabular}[c]{@{}l@{}} \cite{EvenMR15}\end{tabular}}}                         & \multicolumn{1}{l|}{\MCM}                                                          & \multicolumn{1}{l|}{$1-\epsilon$}              & \multicolumn{1}{l|}{$\Delta^{O(1/\epsilon)}+O(\log^{*} n / \epsilon^2)$}                                                     & \multicolumn{1}{l|}{Det.}  & \multicolumn{1}{l|}{\local} \\ \cline{2-6}
\multicolumn{1}{|l|}{}                            & \multicolumn{1}{l|}{\MWM}                                                          & \multicolumn{1}{l|}{$1-\epsilon$}   & \multicolumn{1}{l|}{
$O(\epsilon^{-2}\log \epsilon^{-1}) \cdot \log^\ast n + \Delta^{O(1/\epsilon)} \cdot O(\log \Delta)$
%$\log^{O(1/\epsilon)}W \cdot(\Delta^{O(1/\epsilon)}+O(\log^{*} n ))$
}                              & \multicolumn{1}{l|}{Det.}  & \multicolumn{1}{l|}{\local} \\ \hline
\multicolumn{1}{|l|}{\cite{FGK17}}              & \multicolumn{1}{l|}{\MCM}                                                          & \multicolumn{1}{l|}{$1-\epsilon$}     & \multicolumn{1}{l|}{$O(\Delta^{1/\epsilon} + \poly(\frac{1}{\epsilon}) \log^{*} n)$}                                                                                    & \multicolumn{1}{l|}{Det.}  & \multicolumn{1}{l|}{\local}   \\ \hline
\multicolumn{1}{|l|}{\multirow{2}{*}{\begin{tabular}[c]{@{}l@{}}\cite{GhaffariKM17} $+$\\ ~\cite{RozhonG20,GhaffariGR21} \end{tabular}}}             & \multicolumn{1}{l|}{\MWM}                                                          & \multicolumn{1}{l|}{$1-\epsilon$}     & \multicolumn{1}{l|}{$O(\epsilon^{-1} \log^3 n)$}                                                                                    & \multicolumn{1}{l|}{Rand.}  & \multicolumn{1}{l|}{\local}  \\ \cline{2-6}
\multicolumn{1}{|l|}{}                            & \multicolumn{1}{l|}{\MWM}                                                          & \multicolumn{1}{l|}{$1 - \epsilon$} & \multicolumn{1}{l|}{$O(\epsilon^{-1} \log^7 n)$}           & \multicolumn{1}{l|}{Det.}  & \multicolumn{1}{l|}{\local}  \\ \hline
\multicolumn{1}{|l|}{\cite{GHK18}}             & \multicolumn{1}{l|}{\MCM}                                                          & \multicolumn{1}{l|}{$1-\epsilon$}     & \multicolumn{1}{l|}{$O(\epsilon^{-9} \log^{5} \Delta \log^2 n)$}                                                                                    & \multicolumn{1}{l|}{Det.}  & \multicolumn{1}{l|}{\local}   \\ \hline
\multicolumn{1}{|l|}{\begin{tabular}[c]{@{}l@{}}\cite{GHK18} $+$\\ \cite{GhaffariKMU18} \end{tabular}}             & \multicolumn{1}{l|}{\MWM}                                                          & \multicolumn{1}{l|}{$1-\epsilon$}     & \multicolumn{1}{l|}{$O(\epsilon^{-7} \log^{4} \Delta \log^3 n)$}                                                                                    & \multicolumn{1}{l|}{Det.}  & \multicolumn{1}{l|}{\local}   \\ \hline
\multicolumn{1}{|l|}{\multirow{2}{*}{\begin{tabular}[c]{@{}l@{}} \\ \cite{Harris19}\end{tabular}}}                      & \multicolumn{1}{l|}{\MWM}                                                          & \multicolumn{1}{l|}{$1-\epsilon$}     & \multicolumn{1}{l|}{$O(\epsilon^{-4} \log^2 \Delta + \epsilon^{-1} \log^{*} n)$}                                                                                    & \multicolumn{1}{l|}{Det.}  & \multicolumn{1}{l|}{\local} \\ \cline{2-6}
\multicolumn{1}{|l|}{}                            & \multicolumn{1}{l|}{\MWM}                                                          & \multicolumn{1}{l|}{$1 - \epsilon$} & \multicolumn{1}{l|}{\begin{tabular}[c]{@{}l@{}}$O(\epsilon^{-3}\log (\Delta + \log \log n)$ \\$+\epsilon^{-2}\cdot (\log \log n)^2 )$\end{tabular}}           & \multicolumn{1}{l|}{Rand.}  & \multicolumn{1}{l|}{\local}  \\ \hline
\end{tabular}
\end{adjustbox}
\end{table}

 For a comprehensive survey on matching in the \congest and the \local models, see  \cref{table:matching}.
In the \congest model, \cite{LPP15, BCGS17} showed that a $(1-\epsilon)$ approximate \MCM{} can be computed in rounds with exponential dependencies on $(1/\epsilon)$. Very recently, and independently from our work, \cite{FMU22} showed that a $(1 - \epsilon)$-approximate \MCM{} can be computed in $\poly(\log n, 1/\epsilon)$ rounds. However, for the \MWM{} problem in general graphs, currently the best approximation ratio one can get in $\poly(\log n, 1/\epsilon)$ rounds is $(2/3 - \epsilon)$ by the rounding approach of \cite{AKO18}. Using exponential in $(1/\epsilon)$ rounds, recently \cite{FFK21} showed that a $(1 - \epsilon)$-approximate \MWM{} can be computed in general graphs. Also in bipartite graphs, a $(1-\epsilon)$-approximate \MWM{} is known to be obtainable in $\poly(\log n, 1/\epsilon)$ rounds \cite{LPP15, Fischer17}. 

On the other hand, as shown in \cref{table:matching}, in the \local model, many fast $\poly(\log n, 1/\epsilon)$-round algorithms for computing a $(1-\epsilon)$-approximate \MWM{} in general graphs have been developed.  Using our framework, we obtain the first $\poly(\log n, 1/\epsilon)$-round algorithms for computing $(1-\epsilon)$-approximate \MWM{} in non-trivial graph classes outside bipartite graphs in the \congest model.

%, where the $2/3$ threshold stems from by the integrality gap of the linear program in non-bipartite graphs \footnote{The linear program (without the blossom constraints) has an integrality gap of $3/2$ in non-bipartite graphs and $1$ in bipartite graphs. Incorporating the blossom constraints into the linear program seems to be a challenging task in the \congest model due to the bandwidth.}. 

\begin{restatable}{theorem}{thmmatching}
\label{thm:matching}
A $(1-\eps)$-approximate maximum weighted matching of an $H$-minor-free network $G$ can be computed in  $\eps^{-O(1)} \log^{O(1)} n$ rounds with high probability in the \congest model.
\end{restatable}
%\begin{remark} Note that the constants in the running times of our algorithms depend on $H$, the minor for which the corresponding graph is free of it.\end{remark}

Throughout the paper, although the  hidden leading constant in the round complexity in our algorithms for $H$-minor-free networks depend on $H$, we emphasize that the constants $O(1)$ in the exponents of the round complexity   $\eps^{-O(1)} \log^{O(1)} n$ are independent of $H$.

\paragraph{Maximum Independent Set}
An independent set is a set of non-adjacent vertices. The maximum independent set (\MaxIS) problem is to find an independent set whose cardinality is maximum over all possible independent sets.  Note that a maximal independent set is a $(1/\Delta)$-approximation to the \MaxIS{} problem. Therefore, in the \congest model, a $(1/\Delta)$-approximate solution can be computed in $\MIS(n,\Delta)$ time, where $\MIS(n,\Delta)$ is the number rounds needed to compute a maximal independent set in the \congest model. The weighted version of the problem  was considered in~\cite{BCGS17}, and they gave an algorithm that finds a $(1/\Delta)$-approximate weighted \MaxIS{} in $O(\MIS(n,\Delta)\cdot \log W) $ rounds, where $W$ is the maximum weight. Later, it was shown in \cite{KawarabayashiKS20}  that a $((1-\epsilon)/\Delta)$-approximate weighted \MaxIS{} can be computed in $\poly(\log \log n)\cdot O(1/\epsilon)$ rounds with high probability. Moreover, they also showed that a $((1-\epsilon)/8\alpha)$-approximate weighted \MaxIS{} in graphs of arboricity $\alpha$ can be obtained in $\tilde{O}(\log n / \epsilon)$ rounds with high probability. For the unweighted version, \cite{KawarabayashiKS20} also showed that a $((1-\epsilon)/\Delta)$-approximate \MaxIS{} can be computed in $O(1/\epsilon)$ rounds with high probability when $\Delta \leq n/\log n$.  

In the \local model, Ghaffari, Kuhn, Maus \cite{GhaffariKM17} showed that an $(1-\epsilon)$-approximation to the \MaxIS{} problem can be computed in $\poly(\log n, 1/\epsilon)$ rounds. No analogous $(1-\epsilon)$-approximation algorithms are known in the \congest model as there are lower bounds showing algorithms with constant approximation ratios require $n^{\Theta(1)}$ rounds \cite{BachrachCDELP19,EfronGK20}. Using our framework, we show:

\begin{restatable}{theorem}{thmindependentset}
\label{thm:independentset}
A $(1-\eps)$-approximate maximum independent set of an $H$-minor-free network $G$ can be computed in  $\eps^{-O(1)} \log^{O(1)} n$ rounds with high probability and $\eps^{-O(1)} 2^{O(\sqrt{\log n \log \log n})}$  rounds deterministically in the \congest model.
\end{restatable}

\paragraph{Correlation Clustering} The correlation clustering problem introduced by Bansal, Blum, and Chawla~\cite{BBC04} is known to have various applications in spam detection, gene clustering, chat disentanglement, and co-reference resolution \cite{CDK14, BGL14, ES09, ARS09, EIV07}. In this problem, each edge is labeled with a positive label or a negative label that denotes whether the two endpoints of the edge are positively correlated or negatively correlated.  

The goal is to partition the vertices $V$ into clusters $V_1, V_2, \cdots, V_k$ such that they are as {\it consistent} with the labels as possible. Let $E^{+}$ denote the positively-labeled edges and $E^{-}$ denote the negatively-labeled edges.  There are two versions of the problem: In the agreement maximization version, the goal is to maximize $\sum_{i=1}^k |E^+ \cap (V_i \times V_i)| + \sum_{1 \leq i < j \leq k} |E^- \cap (V_i \times V_j)|$. In the disagreement minimization version, the goal is to minimize $\sum_{i=1}^k |E^- \cap (V_i \times V_i)| + \sum_{1 \leq i < j \leq k} |E^+ \cap (V_i \times V_j)|$. Note that two versions of the problem are equivalent if one is looking for the exact solution. 

We focus on approximate solutions for the agreement maximization version of the problem. In the centralized setting, the problem is shown to be APX-Hard in general graphs \cite{EF03, CGW05}. In particular, Charikar, Guruswami, and Wirth \cite{CGW05} showed that it is NP-hard to approximate the problem within a factor of $115/116+\epsilon$ for any $\epsilon > 0$. On the positive side, they gave a 0.7664-approximation algorithm for the problem. Later, Swamy \cite{Swamy04} gave a 0.7666-approxmation algorithm for the problem. In the distributed setting, while there are $O(1)$-approximation parallel algorithms on complete graphs \cite{CDK14, PPORRJ15, CambusCMU21} for the disagreement minimization problem, to our knowledge, no efficient algorithms for the \congest model or the \local model have been proposed outside of complete graphs for both versions of the problem.\footnote{It is, however, not hard to see that a $\poly(\log n, 1/\epsilon)$-round $(1-\epsilon)$-approximate algorithm for the agreement maximization problem in general graphs can be obtained via low diameter decompositions in the \local model.} Using our framework, we show:

\begin{restatable}{theorem}{thmclustering}
\label{thm:clustering}
A $(1-\eps)$-approximate agreement maximization correlation clustering of an $H$-minor-free network $G$ can be computed in  $\eps^{-O(1)} \log^{O(1)} n$ rounds with high probability and $\eps^{-O(1)} 2^{O(\sqrt{\log n \log \log n})}$  rounds deterministically in the \congest model.
\end{restatable}

% By embedding our framework into a variant of the sequential $(1-\epsilon)$-\MWM{} algorithm of Duan and Pettie~\cite{DuanP-approxMWM}, we give the first distributed algorithm that runs in $\poly(1/\epsilon, \log n)$ rounds for computing a $(1-\epsilon)$-\MWM{} in the \congest model, for all $H$-minor-free networks.

    %In particular, an $(1\pm \epsilon)$-approximate solution of many unweighted optimization problems can be computed in $\poly(1/\epsilon, \log n)$ in a straightforward manner. 

%Roughly speaking, our approach is to decompose the graph into components of high conductance by removing a small fraction of edges, whose removal will only degrade the optimal solution slightly. As each component has a high conductance, we can route the entire graph toplogy of the component to a single vertex efficiently. After gathering the topology, the vertex solves the problem in the component locally and then broadcast the solution back. 

In addition to approximation algorithms for combinatorial optimization problems, we demonstrate  applications of our framework to the realm of property testing and graph decompositions.

\paragraph{Property Testing} 
A {graph property} $\mathcal{P}$ is a set of graphs. We say that a graph $G$ has property $\mathcal{P}$ if $G \in \mathcal{P}$. We say that an $n$-vertex graph $G=(V,E)$ is \emph{$\eps$-far} from having property $\mathcal{P}$ if removing and adding at most $\epsilon|E|$ edges cannot turn $G$ into a graph in $\mathcal{P}$.
%if the hamming distance between $G$ and any $n$-vertex graph $G' \in \mathcal{P}$ is at least $\eps \binom{n}{2}$, where an $n$-vertex graph is interpreted as a binary string of length $\binom{n}{2}$.
The study of  property testing in the distributed setting was initiated by Censor-Hillel, Fischer, Schwartzman, and Vasudev~\cite{censor2019fast}.
We say that a distributed property testing algorithm $\mathcal{A}$ for a property $\mathcal{P}$ with proximity parameter $\eps$ is correct if it satisfies the following.
\begin{itemize}
    \item If $G$ has property $\mathcal{P}$, then all vertices output {\sf Accept}.
    \item If $G$ is $\eps$-far from having property $\mathcal{P}$, then at least one vertex outputs {\sf Reject}.
\end{itemize}

 Levi, Medina, and Ron~\cite{levi2021property} showed an $\eps^{-O(1)} \cdot O(\log n)$-round distributed algorithm for property testing of planarity in the \congest model with \emph{one-sided error}.
 If $G$ has property $\mathcal{P}$, then all vertices output {\sf Accept}.  If $G$ is $\eps$-far from having property $\mathcal{P}$, then at least one vertex outputs {\sf Reject} with high probability.
 Their algorithm uses the distributed planarity testing algorithm of Ghaffari and Haeupler~\cite{ghaffari2016planar} as a subroutine. 
 %The algorithm of~\cite{ghaffari2016planar} decides whether the underlying network $G$ is planar in $O(D \cdot \min\{D, \log n\})$ rounds deterministically, where $D$ is the diameter of $G$.
 
Using our framework, %in \cref{sec:property_testing} 
we give a simple proof that distributed property testing of planarity can be solved in $\poly(1/\eps, \log n)$ rounds in the randomized setting and in $n^{o(1)} \cdot \poly(1/\eps)$ rounds in the deterministic setting. 
More generally, our algorithm can be generalized to testing an \emph{arbitrary} minor-closed graph property that is closed under taking disjoint union.
%, where 
% a graph property $\mathcal{P}$ is said to be minor-closed if $G \in \mathcal{P}$ implies $H \in \mathcal{P}$ for all $H \preceq G$.

\begin{restatable}{theorem}{thmtesting}
\label{thm:testing}
Distributed property testing for any minor-closed graph property $\mathcal{P}$ that is closed under taking disjoint union can be solved in $\eps^{-O(1)} \log^{O(1)} n$ rounds with high probability and $\eps^{-O(1)} 2^{O(\sqrt{\log n \log \log n})}$  rounds deterministically in the \congest model.
\end{restatable}

\paragraph{Graph Decompositions}
An $(\eps, D)$ \emph{low-diameter decomposition} of a graph $G=(V,E)$ is a partition of the vertex set $V = V_1 \cup V_2 \cup \cdots \cup V_k$ such that the number of inter-cluster edges is at most $\eps |E|$ and the diameter of the induced subgraph $G[V_i]$ is at most $D$ for each $1 \leq i \leq k$.

It is well-known~\cite{klein1993excluded,Fakcharoenphol2003improved,Ittai2019padded} that for any $H$-minor-free graph, a low-diameter decomposition  with $D = O(\eps^{-1})$ exists, where the hidden constant in $O(\cdot)$ depends only on $H$. It is straightforward to see that the inverse linear dependence $D = O(\eps^{-1})$  on $\eps$ is the best possible by considering cycle graphs.

In the distributed setting, Czygrinow, Ha{\'n}{\'c}kowiak, and Wawrzyniak~\cite{czygrinow2008fast} designed a distributed algorithm that computes a low-diameter decomposition with $D = \eps^{-O(1)}$ in $\eps^{-O(1)} \cdot O(\log^\ast n)$ rounds for planar networks in the \local model. Their algorithm also applies to the edge-weighted setting where the guarantee of the algorithm is that the summation of the weights of inter-cluster edges is at most $\eps$-fraction of the sum of the weights of all edges. 
%Although they presented their algorithm in the \local model, the algorithm also works in the \congest model. %Czygrinow, Ha{\'n}{\'c}kowiak, and Wawrzyniak~\cite{czygrinow2008fast} used this decomposition algorithm to design efficient approximate algorithms for maximum matching, maximum independent set, and minimum dominating set for planar networks. 
Levi, Medina, and Ron~\cite{levi2021property} also designed a distributed algorithm that computes a low-diameter decomposition with $D = \eps^{-O(1)}$ in $\eps^{-O(1)} \cdot O(\log n)$ rounds for $H$-minor-free networks in the \congest model, which is used in their distributed algorithm for property testing of planarity.

%Using \cref{thm:routing-main}, 
Using our framework, %in  \cref{sec:low_diam_decomposition}
we improve the inverse polynomial dependence $D = \eps^{-O(1)}$ on $\eps$ to the optimal $D = O(\eps^{-1})$. 
We present a simple proof that a low-diameter decomposition with $D = O(\eps^{-1})$ can be computed in $\eps^{-O(1)} \log^{O(1)} n$ rounds with high probability and $\eps^{-O(1)} 2^{O(\sqrt{\log n \log \log n})}$  rounds deterministically in the \congest model.

\begin{restatable}{theorem}{thmdecomp}
\label{thm:decomp}
Given an $H$-minor free network, a low-diameter decomposition with $D = O(\eps^{-1})$ can be computed in $\eps^{-O(1)} \log^{O(1)} n$ rounds with high probability and $\eps^{-O(1)} 2^{O(\sqrt{\log n \log \log n})}$  rounds deterministically in the \congest model.
\end{restatable}

\subsection{Our Framework}

Our framework of algorithm design is based on the recently developed distributed constructions of expander decompositions in the \congest model~\cite{ChangS20,Chang2021triangleJACM}. 
 We say that a graph is an \emph{$\phi$-expander} if its {conductance} is at least $\phi$. 
An $(\eps, \phi)$ \emph{expander decomposition} of a graph is a removal of at most $\eps$ fraction of the edges such that each remaining connected component is an $\phi$-expander.  Intuitively, the conductance of a graph measures how well-connected it is. In particular, any random walk converges quickly to its stationary distribution in a high-conductance graph. Expander decompositions have a wide range of applications in theoretical computer science, including linear system solvers~\cite{spielman2004nearly},
 unique games~\cite{AroraBS2015,RaghavendraS2010},
 minimum cut~\cite{KawarabayashiT15}, property testing~\cite{KumarSS2018,GoldreichR1999}, and dynamic algorithms~\cite{NanongkaiSW17,chuzhoy2020deterministic}.
 % intro to expander decompositions + its application in various areas of theoretical computer science.

%\paragraph{Our Approach}  In this paper, we present a  framework of algorithm design based on expander decompositions that have applications for a wide range of combinatorial optimization, property testing, and graph decomposition problems in the \congest model.

% Essentially all applications are limited to the subgraph finding problem, can we obtain applications in other types of problems?

We say that $H$ is a \emph{minor} of $G$ if $H$ can be obtained from $G$ by iteratively  removing vertices and edges and contracting edges. We write $H \preceq G$ if $H$ is a minor of $G$. We say that $G$ is \emph{$H$-minor-free} if $H \npreceq G$.
A class of graphs $\GG$ is \emph{minor-closed} if $G \in \GG$ implies $H \in \GG$ for any $H \preceq G$. 
Many natural graph classes, such as planar graphs, bounded-genus graphs, and bounded-treewidth graphs, are minor-closed.
The graph minor theorem of Robertson and Seymour~\cite{ROBERTSON2004325} implies that for any minor-closed family of graphs $\GG$, there exists a \emph{finite} set of \emph{forbidden minors} $\HH$ such that $G \notin \GG$ if and only if $H \preceq G$ for some $H \in \HH$. For example, if $\GG$ is the set of all planar graphs, then $\HH = \{K_5, K_{3,3}\}$. That is, $G$ is planar if and only if $G$ is $K_{3,3}$-minor-free and $K_5$-minor-free. Note that the graph minor theorem also implies that a minor-closed family of graphs must be a subset of the family of $H$-minor free graphs for some fixed graph $H$.

In this paper, we focus on the class of $H$-minor-free networks for any fixed $H$. The idea of our framework is that we want to use expander decompositions in the \congest model in a way similar to the use of low-diameter decompositions~\cite{LinialS93,RozhonG20} in the \local model. That is, for each low-diameter cluster $V_i$, we want to gather the graph topology $G[V_i]$ to a vertex $v_i^\ast \in V_i$ so that $v_i^\ast$ can run any sequential algorithm on $G[V_i]$ locally and broadcast the result to all other vertices in $V_i$. This approach clearly requires sending messages of unlimited size in the general case.

An \emph{edge separator} of a graph is a cut $\{S, V\setminus S\}$ such that $\min\{|S|, |V\setminus S|\} \geq |V|/3$. The \emph{size} of an edge separator $\{S, V\setminus S\}$ is the number of edges crossing $S$ and $V \setminus S$. If $G[V_i]$ is an $\phi$-expander and  admits a small edge separator, then there must exist a high-degree vertex $v_i^\ast \in V_i$, so the connectivity property of a $\phi$-expander allows us to design an efficient routing algorithm to let $v_i^\ast$  gather the entire graph topology of  $G[V_i]$. These properties are shown in \cref{sec:graph_p}.

It is known~\cite{diks1993separator,Miller86} that planar graphs admit an edge separator of size $O(\sqrt{\Delta |V|})$.  More generally, any graph that can be embedded on a surface of genus $g$  has an edge separator of size $O(\sqrt{g\Delta |V|})$~\cite{SYKORA1993419}.
In this paper, we generalize these results to show that all  $H$-minor-free graphs admit an edge separator of size $O(\sqrt{\Delta |V|})$, so the approach discussed above is applicable to all $H$-minor-free graphs.

\begin{restatable}{theorem}{thmedgeseparator}\label{thm:edge-separator}
For any $H$-minor-free graph $G=(V,E)$, there is a cut $S$ such that $\min\{|S|, |V\setminus S|\} \geq n/3$ and $|\partial(S)|=O(\sqrt{\Delta n})$, where the hidden constant in $O(\cdot)$ depends only on $H$.
\end{restatable}

\cref{thm:edge-separator} is of independent interest. For example, it was asked~\cite{lason2021modularity} in whether such an edge separator exists for any $H$-minor-free graph.
The proof of \cref{thm:edge-separator} is in \Cref{sect:edge-separator}. The proof is based on  the graph structure theorem of Robertson and Seymour~\cite{ROBERTSON2003nonplanar}, which states that any $H$-minor-free graph can be obtained by gluing graphs that are almost embeddable on some fixed surfaces in a tree structure.

\paragraph{Remark} There is a series~\cite{chuzhoy2019large,krivelevich2021complete,kleinberg1996short,kawarabayashi2010separator} of research investigating the relation between maximum degree and minor containment for expanders. Specifically, define $f(n,\alpha,d)$ to be the maximum number such that every graph $H$ with at most  $f(n,\alpha,d)$ vertices and edges is a minor of every $n$-vertex graph $G=(V,E)$ with maximum degree at most $d$ such that  for every partition $(S, V \setminus S)$ of its vertices into two non-empty parts, the number of edges connecting $S$ to $V \setminus S$ is at least $\alpha \cdot \min \{|A|, |B|\}$.
It was shown in~\cite{chuzhoy2019large} that there exists a universal constant $c$ such that $f(n,\alpha,d) = \Omega\left( \frac{n}{\log n} \cdot \left(\frac{\alpha}{d}\right)^c \right)$, and this bound holds for $c = 2$~\cite{krivelevich2021complete}, see~\cite{chuzhoy2019large} for more details.  
This bound implies that for any fixed $H$, the maximum degree of an $n$-vertex $\phi$-expander is  $\Delta = \Omega\left( \phi \cdot \left(\frac{n}{\log n}\right)^{\frac{1}{c}} \right)$. This lower bound is not suitable for our purpose. To aim for a round complexity of $\poly(1/\phi, \log n)$ for graph topology gathering, it is necessary to obtain a lower bound of the form  $\Delta = \Omega\left( \frac{n}{\poly(1/\phi, \log n)}\right)$, which we show in \cref{lem:separator}.

\subsection{Applying Our Framework %and Other Combinatorial Optimization Problems
}\label{sec:matching_intro}

We show that by using our framework, many unweighted optimization problems can be approximated within $(1 \pm \epsilon)$ factors in a straightforward manner. Then, we use the \MWM{} problem as an example to demonstrate that our framework can be applied to solve weighted problems as well. 

\paragraph{Unweighted Problems}
%, Maximum Independent Set and Correlation Clustering}
%To understand how our framework works, we first describe how it can be applied for obtaining $(1-\epsilon)$-approximate solutions for some unweighted combinatorial problems, namely, the maximum independent set problem and the agreement maximization correlation clustering problem on $H$-minor-free graphs and the \MCM{} problem on planar graphs. 
%As a warm up, we show that by using the framework, we immediately
As a warm up, to illustrate how our framework can be used, we first describe how to use our framework to obtain simple $\poly(1/\epsilon, \log n)$-round $(1-\epsilon)$-approximate algorithms for \MCM{} in planar graphs, as well as other unweighted problems such as the \MaxIS{} problem and the correlation clustering problem in $H$-minor-free graphs. %, and \MCM{} in planar graphs.

The idea behind these algorithms is very simple. If the size of an optimal solution is linear in the number of vertices, then we can simply let each cluster of an expander decomposition to compute its local optimal solution by letting a high-degree vertex in the cluster learn the graph topology of the cluster. We just need to show that ignoring the $\epsilon |E|$ inter-cluster edges only worsens the quality solution by a factor of at most $(1 - \eps)$. It is conceivable that this  holds for many unweighted problems. Indeed, this is true for the \MaxIS{} problem in $O(1)$-arboricity graphs and the agreement maximization correlation clustering problem in general graphs, so our framework immediately gives efficient $(1-\epsilon)$-approximate algorithms for these problems in $H$-minor-free graphs,  see \cref{sec:independent_set,sec:correlation_clustering} for details.

For the case of the \MCM{} problem, the size of an optimal solution is not linear in the number of vertices in general, but it is possible to preprocess the graph so that the size of an optimal solution is linear in the number of the vertices by using the preprocessing procedure of \cite{czygrinow2006} for planar graphs. By doing so, we obtain a $(1-\eps)$-approximate algorithm for \MCM{} on planar graphs. We describe such an approach in \Cref{sec:unweighted_matching}.

\paragraph{%Technical Summary for 
Weighted Matching} Extending the framework to weighted problems is significantly more challenging, because when applying the expander decomposition in \cref{thm:routing-main} we do not have control over which edges we will remove. For unweighted problems, the $\epsilon |E|$ edges that we remove usually can only cause a small degrade on the optimal solution. However, in the weighted problem, the small fraction of edges could have very high weights. As a result, the optimal solution could become much worse after removing those edges.

To overcome this obstacle, instead of applying the decomposition only once in the beginning, we embed our method into Duan and Pettie's sequential scaling algorithm \cite{DuanP-approxMWM} for approximating \MWM{}. Roughly speaking, their scaling algorithm is a primal-dual algorithm that consists of multiple iterations. It processes the subgraphs from the ones induced by higher weight edges to the ones induced by lower weight edges over the iterations. Each iteration consists of non-trivial steps such as the augmentation step as well as the blossom shrinking step that are not easily implementable in the \congest model, but implementable in linear time in the centralized setting. For example, the augmentation step involves finding a maximal set of augmenting paths in the working subgraph. Since the length of an augmenting path can as large as $\Theta(n)$, it would not be possible to find it in $\poly(\log n , 1/\epsilon)$ rounds in the \congest model.

%For example, the augmentation step involves finding a maximal set of augmenting paths in non-bipartite graphs, which is highly non-trivial in the \congest model, although very recently \cite{FMU22} developed an efficient algorithm to find an approximate maximal set of bounded-length augmenting paths. The blossom shrinking step is to contract some sets of vertices. In \cite{DuanP-approxMWM}, these sets can be large and nested so it is unclear if they can be done efficiently in the \congest model.

We apply our expander decomposition framework to the working subgraph before some of the non-trivial steps. Instead of physically removing the inter-component edges from the graph, we add or subtract a small weight to the edges so they are no longer the ``tight''  edges (i.e.~the edges in the working subgraph) in the primal-dual algorithm. We show that adding or subtracting the small weights would only degrade the optimal solution slightly. Moreover, we show this allows us to process each component independently (e.g.~the long augmenting paths mentioned in the previous paragraph would be broken). Each component can then route the topology to a vertex and let the vertex perform the non-trivial steps locally and broadcast the result back.  

This summarizes the high-level idea. However, there are several technical challenges such as that the expander decomposition may cut through some intermediate structures (i.e.~the active blossoms). In addition, similar to the aforementioned unweighted case, we also need to preprocess the working subgraph before running the expander decomposition to ensure the number of inter-component edges is small relative to the solution. The planar graph preprocessing procedure of \cite{czygrinow2006} does not work for $H$-minor free graphs in general. We discuss how we resolve these issues in \Cref{sec:alg}. 

\subsection{Related Work}

  Chang, Pettie, Saranurak, and Zhang~\cite{Chang2021triangleJACM} gave the first application of expander decompositions to the \congest model of distributed computing. They designed a distributed algorithm for constructing an expander decomposition and applied it to give a near-optimal distributed algorithm for the \emph{triangle listing} problem, based on the following framework.
First construct an $(\eps, \phi)$-expander decomposition to partition the vertex set $V= V_1 \cup V_2 \cup \cdots \cup V_k$ into $\phi$-expanders.
Using existing routing algorithms~\cite{GhaffariKS17,GhaffariL2018} for $\phi$-expanders, existing distributed triangle listing algorithms that make use of non-local communication can be simulated in $\phi$-expanders with small overhead. Based on this approach, all triangles containing at least one edge in $G[V_i]$ can be listed efficiently, for all high-conductance clusters $G[V_i]$ in parallel. Finally, the remaining $\eps$-fraction of the inter-cluster edges are handled using recursive calls. 
% Distributed expander decomposition: discuss the algorithm design framework behind the triangle listing algorithm, and discuss expander routing.

Subsequent to the work of~\cite{Chang2021triangleJACM}, expander decomposition has been applied to numerous other problems in the \congest model via this framework of algorithm design~\cite{CensorLL20,censor2021tight,EFFKO19,izumiLM20,LeGall2021}.
So far, all applications of distributed expander decomposition  have been confined to the distributed subgraph finding problems~\cite{censorhillel:LIPIcs.ICALP.2021.3}, except the work of Daga, Henzinger, Nanongkai, and Saranurak~\cite{Daga2019distributed}, where they designed a sublinear-round exact min-cut algorithm in the \congest model by incorporating distributed expander decomposition into the sequential min-cut algorithm of Kawarabayashi and Thorup~\cite{KawarabayashiT15}.
% Subsequent applications of expander decomposition in CONGEST.

\paragraph{Distributed Algorithms on Minor-closed Networks}
 Many real-world networks have sparse structures.  Over the past few years, much of the research effort has been devoted to designing efficient distributed algorithms in \local and \congest utilizing structural properties of sparse networks, and many natural graph classes studied in the literature, such as planar graphs, bounded-genus graphs, and bounded-treewidth graphs, are minor-closed, so they can be characterized by a finite list of excluded minors. 
 %We briefly review the prior work on distributed algorithms for graphs with an excluded minor.
 %,  fitting into the realm of $H$-minor-free graphs. 
 %In this section we  review the prior work on this topic.
% \begin{description}

\paragraph{Distributed Approximation}
 There is a long line of research studying distributed approximation on graphs with an excluded minor~\cite{akhoondian2016local,amiri2019distributed,bonamy2021tight,Czygrinow2007cocoon,CZYGRINOW2006JDA,czygrinow2006,Czygrinow2006ESA,czygrinow2008fast,czygrinow2014distributed,czygrinow2020distributed,lenzen2013distributed,wawrzyniak2014strengthened}.   Czygrinow, Ha{\'n}{\'c}kowiak, and Wawrzyniak~\cite{czygrinow2008fast} showed that an $(1 \pm \eps)$-approximation of maximum matching, maximum independent set, and minimum dominating set of a planar graph can be constructed in $O(\log^\ast n)$ rounds deterministically in the \local model, for any constant $\eps > 0$.
 The algorithm for minimum dominating set was later extended to $k$-dominating set on  bounded-genus graphs~\cite{amiri2019distributed,czygrinow2008fast,czygrinow2020distributed}.
 These algorithms are based on a generic approach~\cite{amiri2019distributed,czygrinow2020distributed,CZYGRINOW2006JDA, Czygrinow2006ESA,Czygrinow2007cocoon,czygrinow2008fast} using low-diameter decompositions. A common ingredient shared by all these algorithms is a computation of an $(1 \pm \eps)$-approximate solution of each low-diameter cluster via a brute-force information gathering, requiring sending unbounded-size messages and confining all these algorithms to the \local model. Our framework which is based on expander decompositions provides an opportunity to extend this line of research %~\cite{amiri2019distributed,czygrinow2020distributed,CZYGRINOW2006JDA, Czygrinow2006ESA,Czygrinow2007cocoon,czygrinow2008fast} 
  to the \congest model.

 %alipour2020local - O(1)-factor approx dominating set on certain planar graphs
 %akhoondian2016local - O(1)-factor approx dominating set on bounded-genus graphs
 %amiri2019distributed - 1+eps approx dominating set for bounded-genus graphs (using the approaxh of CHW)
 %bonamy2021tight - dominating set, outerplanar
 %czygrinow2020distributed - 1+eps approx k-dominating set for certain H-minor-free graphs (using the approaxh of CHW)

 %Decomposition method: old: CZYGRINOW2006JDA, Czygrinow2006ESA, Czygrinow2007cocoon  main: czygrinow2008fast,  new: amiri2019distributed, czygrinow2020distributed

\paragraph{Low-congestion Shortcuts and its Applications} There is a line of work designing efficient algorithms on networks with an excluded minor via low-congestion shortcuts~\cite{ghaffari2021low,ghaffari2016distributed,ghaffari2017near,haeupler2016low,haeupler2016near,haeupler2018minor,haeuplerLi2018disc,haeupler2018round}. Given a partition of the vertex set $V$ of a graph $G=(V,E)$ into connected clusters $V = V_1 \cup V_2 \cup \cdots \cup V_k$, a low-congestion shortcut with congestion $c$ and dilation $d$ is a set of subgraphs $H_1, H_2, \ldots, H_k$ such that the diameter of $G[V_i] + H_i$ is at most $d$ and each edge belongs to at most $c$ subgraphs $H_i$. Here  $G[V_i] + H_i$ denotes the subgraph of $G$ induced by the union of the edges in $G[V_i]$ and $H_i$.
For any clustering of an $H$-minor-free graph, there is an $\tilde{O}(D)$-round \congest algorithm computing a low-congestion shortcut with $c = O(D \log n)$ and $d = O(D)$, where $D$ is the diameter of the graph~\cite{ghaffari2021low}. As a result, many graph problems, including minimum spanning tree, minimum cut,
and shortest-path approximations, can be solved in near-optimal $\tilde{O}(D)$ rounds in \congest on any $H$-minor-free graph~\cite{ghaffari2016distributed,haeupler2018round}.

The type of problems efficiently solvable via low-congestion shortcuts is fundamentally very different from the type of problems efficiently solvable via our framework. Low-congestion shortcut is useful in designing near-optimal $\tilde{O}(D)$-round algorithms for \emph{global} problems that already require $\Omega(D)$ rounds to solve. Our framework is useful in designing algorithms that take $\log^{O(1)} n$ or $n^{o(1)}$ rounds for \emph{local} problems that do not have the  $\Omega(D)$ lower bound.

\paragraph{Other Topics} For planar networks in the \congest model, efficient algorithms  for diameter computation~\cite{li2019planar}, reachability~\cite{parter2020distributed}, and depth-first search~\cite{ghaffari2017near} have been designed. An efficient algorithm for distributed planarity testing was given by Ghaffari and Haeupler~\cite{ghaffari2016planar}, which was later used in the property testing algorithm of Levi, Medina, and Ron~\cite{levi2021property}.
  Following a work of Naor, Parter, and Yogev~\cite{Naor2020soda} which demonstrated the existence of a
distributed interactive proof for planarity, there is a series of research studying {local certification} for various minor-closed graph classes~\cite{feuilloley2021compact,DBLP:journals/corr/abs-2108-00059,DBLP:journals/corr/abs-2007-08084,esperet2021local}.
% \end{description}
%\subsection{Technical Summary}
%     Czygrinow, Ha{\'n}{\'c}kowiak, and Wawrzyniak~\cite{czygrinow2008fast} showed that an $1 \pm \eps$ approximation of maximum matching, maximum independent set, and minimum dominating set of a planar graph can be constructed in $O(\log^\ast n)$ rounds deterministically in the \local model, for any constant $\eps > 0$. Their approach has been adapted to solving many distributed problems in planar graphs and other graph classes. However, all algorithms resulting from this line of research requires large messages. 
% In this paper, we show that using expander decompositions, it is possible to turn all these algorithms into message-efficient ones. In particular,  an $1 \pm \eps$ approximation of maximum matching, maximum independent set, and minimum dominating set of a planar graph can be constructed in $\poly(\eps^{-1}, \log n)$ rounds with high probability or $n^{o(1)} \poly(\eps^{-1})$ rounds deterministically, using $O(\log n)$-bit messages.

\subsection{Subsequent Developments}

Since the initial publication of this work~\cite{10.1145/3519270.3538423}, several of our results have been improved. Building on the techniques developed here, a subsequent work~\cite{huang20231} showed that a $(1 - \eps)$-approximate \MWM{} in general networks can be computed deterministically in $\poly(\log n, 1/\eps)$ rounds in the \congest model.

Another follow-up work~\cite{chang2023efficient} developed improved decomposition and routing algorithms for $H$-minor-free networks. Combined with our framework, these results yield faster algorithms for many problems in combinatorial optimization and property testing. For example, a $(1 - \eps)$-approximate maximum independent set in an $H$-minor-free network can now be computed deterministically in $O(\eps^{-1} \log^\ast n) + \eps^{-O(1)}$ rounds~\cite{chang2023efficient}. Additionally, property testing for any minor-closed graph property that is closed under disjoint union can be performed deterministically in $O(\log n)$ rounds when $\eps$ is constant, or in $O(\eps^{-1} \log n) + \eps^{-O(1)}$ rounds when the maximum degree $\Delta$ is constant.

\subsection{Organization}\label{sec:organization}

Our framework of algorithm design based on expander decompositions is presented in 
\cref{sec:graph_p}. Using this framework, in \cref{sec:applications}, we give $\poly(1/\eps, \log n)$-round randomized algorithms and $n^{o(1)} \cdot \poly(1/\eps)$-round deterministic algorithms for various optimization, property testing, and graph decomposition problems on planar or $H$-minor-free networks, proving \cref{thm:independentset,thm:clustering,thm:decomp,thm:testing}. In \cref{sec:matching}, we embed our framework into the sequential scaling algorithm of~\cite{DuanP-approxMWM} for approximating \MWM{} to prove \cref{thm:matching}. In \cref{sect:edge-separator}, we show that a small edge separator exists for any $H$-minor-free graph, proving \cref{thm:edge-separator}, which is needed in our framework.

\section{Graph Partitioning}\label{sec:graph_p}
Let $G = (V,E)$ be a graph. Consider the following graph terminology regarding a subset $S \subseteq V$.
\begin{align*}
   \vol(S) & =\sum_{v \in S} \deg(v),\\
   \partial(S) &= E(S, V \setminus {S}) = \{ e = \{u,v\} \in E \ | \ \{u,v\} \cap S \neq \emptyset \ \text{and} \ \{u,v\} \cap (V\setminus S) \neq \emptyset \},\\
   \Phi(S) &=
\begin{cases}
   0, & S = \emptyset \ \text{or} \ S = V,\\
   \frac{|\partial(S)|}{\min\{\vol(S), \vol(V \setminus S)\}}, & S \neq \emptyset \ \text{and} \ S \neq V.\\
\end{cases}   
\end{align*}

We call $\vol(S)$ the \emph{volume} of the vertex set $S$. When $S$ is interpreted as a cut $\{S, V \setminus S\}$, we call $\Phi(S)$ the  \emph{conductance} of the cut $S$.

\paragraph{Graph Conductance}
The \emph{conductance} of a graph $G$ is defined as \[\Phi(G) = \min_{S \subseteq V  \ \text{s.t.} \  S \neq \emptyset \ \text{and} \  S \neq V} \Phi(S).\] 
In other words, $\Phi(G)$ is the minimum value of $\Phi(S)$ over all non-trivial cuts $S \subseteq V$.

\paragraph{Mixing Time} A \emph{uniform lazy random walk} starting at a vertex $v \in V$ is described by the following probability distribution, where $N(u)$ denotes the set of neighbors of $u$.
 \begin{align*}
   p_0^v (u) &=
\begin{cases}
   1, & u = v,\\
   0, & u \neq v,\\
\end{cases}\\
   p_i^v (u) &= \frac{1}{2} \cdot p_{i-1}^v (u) + \frac{1}{2 \deg(u)}  \cdot \sum_{w \in N(u)} p_{i-1}^v (w), & \text{for} \; \; \; i \geq 1.
\end{align*}
 
If $G$ is connected, then the stationary distribution of a uniform lazy random walk is $\pi(u) = \deg(u) / \vol(V)$, regardless of the starting vertex $v$.
 The 
\emph{mixing time} $\mix(G)$ of $G$ is defined as the minimum number $t$ such that  $|p_t^v(u) - \pi(u)| \leq \pi(u) / |V|$ for all $u \in V$ and $v \in V$.
 The following relation~\cite{JerrumS89} between the 
mixing time $\mix(G)$ and conductance $\Phi(G)$ is well-known: 
\[
\Theta\left(\frac{1}{\Phi(G)}\right) \leq \mix(G) \leq \Theta\left(\frac{\log |V|}{\Phi(G)^2}\right).
\]

\paragraph{Expander Decompositions}
 We say that $G$ is an \emph{$\phi$-expander} if $\Phi(G) \geq \phi$. 
An $(\eps, \phi)$ \emph{expander decomposition} of a graph is a removal of at most $\eps$ fraction of the edges such that each remaining connected component has conductance at least $\phi$. 
Formally, an $(\eps,\phi)$-expander decomposition of $G$ is a partition $E=E_{1}\cup E_2 \cup \cdots\cup E_{k} \cup \Er$ of the edge set $E$ meeting the following requirements.
\begin{itemize}
    \item The set of inter-cluster edges $\Er$ satisfies $|\Er| \leq \eps |E|$.
    \item We write $V_i \subseteq V$ to denote the set of vertices incident to an edge in $E_i$. It is required that  $V = V_1 \cup V_2 \cup \cdots \cup V_k$ partitions the vertex set $V$ and $G_i=(V_i, E_i)$ has conductance $\Phi(G_i) \geq \phi$ for each $1 \leq i \leq k$.
\end{itemize}

Existentially, it is
 well known  that for any $n$-vertex graph,  an $(\epsilon,\phi)$-expander decomposition  exists for any $0 < \epsilon < 1$ and $\phi = \Omega(\epsilon/\log n)$~\cite{GoldreichR1999,KannanVV04,spielman2004nearly}, and this bound is tight. After removing any constant fraction of the edges in a hypercube, some remaining
 component must have conductance at most $O(1/\log n)$~\cite{AlevALG18}.

The following distributed algorithms for constructing expander decompositions are  due to Chang and Saranurak~\cite{ChangS20}.

\begin{theorem}\label{thm:expander-decomposition-rand}
For any $0 < \eps < 1$, an $(\eps,\phi)$-expander decomposition of a graph $G=(V,E)$ with  $\phi =  \eps^{O(1)} \log^{-O(1)} n$ can be constructed in $\eps^{-O(1)} \log^{O(1)} n$ rounds with high probability.
\end{theorem}

%\footnote{We say that an algorithm succeeds \emph{with high probability} if it succeeds with probability at least $1  - n^{-C}$, where $C$ is an arbitrarily large given constant and $n = |V|$ is the number of vertices in the underlying network $G$.}

\begin{theorem}\label{thm:expander-decomposition-det}
For any $0 < \eps < 1$, an $(\eps,\phi)$-expander decomposition of a graph $G=(V,E)$ with  $\phi =  \eps^{O(1)} 2^{-O(\sqrt{\log n \log \log n})}$ can be constructed in $\eps^{-O(1)} 2^{O(\sqrt{\log n \log \log n})}$ rounds deterministically.
\end{theorem}

\subsection{Existence of a High-degree Vertex}\label{sec:high_deg} Let $E=E_{1}\cup E_2 \cup \cdots\cup E_{k} \cup \Er$ be any $(\eps,\phi)$-expander decomposition. Let  $G_i = (V_i, E_i)$ be the subgraph of $G$ induced by $E_i$. Let $\Delta_i$ be the maximum degree of the graph $G_i$. Recall that  $V = V_1 \cup V_2 \cup \cdots \cup V_k$ partitions the vertex set $V$ and observe that $G_i$ is a subgraph of $G[V_i]$, the subgraph of $G$ induced by the vertex set $V_i$.
We will show that if $G$ is $H$-minor-free, then there must exist a vertex in each $G[V_i]$ whose degree is $ \Omega(\phi^2) |V_i|$ for any $(\eps,\phi)$-expander decomposition of $G$.

\paragraph{Edge Separators}
An \emph{edge separator} of a graph is a cut $\{S, V\setminus S\}$ such that \[\min\{|S|, |V\setminus S|\} \geq |V|/3.\]  
The \emph{size} of an edge separator is the number of cut edges $|\partial(S)|$. %It is known~\cite{diks1993separator,Miller86} that planar graphs admit an edge separator of size $O(\sqrt{\Delta |V|})$.  More generally, any graph that can be embedded on a surface of genus $g$  has a size-$O(\sqrt{g\Delta |V|})$ edge separator~\cite{SYKORA1993419}. 
In \cref{sect:edge-separator}, we show that any $H$-minor-free graph $G$ admits an edge separator of size $O(\sqrt{\Delta |V|})$, where the hidden constant in $O(\cdot)$ depends only on $H$. The following lemma is a  consequence of this result.

\begin{lemma}\label{lem:separator}
If $G$ is $H$-minor-free, then $\Delta_i = \Omega(\phi^2) |V_i|$ for each cluster $G_i = (V_i, E_i)$ of any $(\eps,\phi)$-expander decomposition of $G$. The hidden constant in $\Omega(\cdot)$ depends only on $H$.
\end{lemma}
\begin{proof}
We focus on the $H$-minor-free graph $G_i = (V_i, E_i)$ in the proof. 
Consider any $O(\sqrt{\Delta_i |V_i|})$-size edge separator $S$ of $G_i$. The fact that  $\min\{|S|, |V_i\setminus S|\} \geq |V_i|/3$ implies \[\min\{\vol(S), \vol(V_i \setminus S)\} \geq  \min\{|S|, |V_i\setminus S|\} = \Omega(|V_i|),\] and hence \[\phi \leq \Phi(G_i) \leq \Phi(S) = \frac{|\partial(S)|}{\min\{\vol(S), \vol(V \setminus S)\}} = O\left(\frac{\sqrt{\Delta_i |V_i|}}{V_i}\right) = O\left(\sqrt{\frac{\Delta_i}{|V_i|}}\right),\]
which implies $\Delta_i = \Omega(\phi^2) |V_i|$.
\end{proof}

\paragraph{Remark} There is an alternative method to force the existence of a high-degree vertex via  \emph{local treewidth} if $G$ is $H$-minor-free for an apex graph $H$. An \emph{apex graph} is a graph  that contains  a vertex whose removal makes the graph planar. Demaine and Hajiaghayi showed that if $G$ is $H$-minor-free for some apex graph $H$, then the treewidth $\tw(G)$ of $G$ is linear in the diameter of $G$~\cite{Demaine2004localtreewidth}.  
%By a well-known property of bounded-treewidth graph~\cite[Theorem~2.5]{ROBERTSON1986treewidth}, there exists $S \subseteq V$ with $|S| \leq \tw(G)+1$ such that each connected component of the subgraph of $G$ induced by $V \setminus S$ has size at most $|V|/2$.

Suppose $G$ is $H$-minor-free for some apex graph $H$.
We show the existence of a vertex $v_i^\ast \in V_i$ with degree $\Omega(\phi^{2} \log^{-1} n) |V_i|$.
Since $G_i = (V_i, E_i)$ has conductance at least $\phi$, the diameter of $G_i$ is $D = O(\phi^{-1} \log n)$, so the treewidth of $G_i$ is $O(\phi^{-1} \log n)$. By a well-known property of bounded-treewidth graphs~\cite[Theorem~2.5]{ROBERTSON1986treewidth}, there exists a subset $S \subseteq V_i$ with $|S| =  O(\phi^{-1} \log n)$ such that each connected component of the subgraph of $G_i$ induced by $V_i \setminus S$ has size at most $|V_i|/2$. Since $G_i$ is a $\phi$-expander, the number of edges incident to $S$ in $G_i$ is at least $\phi  |V_i \setminus S|$. If $|S| \leq |V_i|/2$, then $\phi |V_i \setminus S| = \Omega(\phi) |V_i|$, so at least one vertex $v_i^\ast$ in $S$ has degree in $G_i$ at least $\Omega(\phi) |V_i| / |S| = \Omega(\phi^{2} \log^{-1} n) |V_i|$. Otherwise, $|S| > |V_i|/2$ implies that $|V_i| = O(\phi^{-1} \log n)$, in which case the existence of a vertex  $v_i^\ast$ with degree $\Omega(\phi^2 \log^{-1} n) |V_i|$ is trivial.

This approach however does not generalize to the case  where $H$ is not an apex graph. If a minor-closed class of graphs $\GG$ contains all apex graphs, then the treewidth of $G \in \GG$ is not bounded by any function of the diameter of $G$~\cite{eppstein2000diameter}. This can be easily seen by considering the apex graph formed by adding an edge between a vertex $v$ to each vertex in a $k \times k$ grid. This graph has diameter $D = 1$ and treewidth $w = k+1$. 

\subsection{Routing}

Select $v_i^\ast$ as any vertex $v \in V_i$ that has the maximum degree $\Delta_i$ in $G_i$. We show that the bound given by \cref{lem:separator} implies an efficient algorithm for $v_i^\ast$ to learn the entire graph topology of $G[V_i]$. 

\paragraph{Expander Routing}
Consider a routing task where each vertex $v$ is the source and the destination of at most $L \cdot \deg(v)$ $O(\log n)$-bit messages. If $G$ is an $\phi$-expander, then such a task can be solved in $L \cdot \mix  \cdot 2^{O(\sqrt{\log n})} = L \cdot \phi^{-2} \cdot 2^{O(\sqrt{\log n})}$ rounds with high probability~\cite{GhaffariKS17,GhaffariL2018} or  $L \cdot \phi^{-O(1)} \cdot 2^{O(\sqrt{\log n \log  \log n})}$ rounds deterministically~\cite{chang2024deterministic}.

%Applying these routing algorithms to $G_i$,  $v_i^\ast$ can efficiently learn the entire graph topology of $G[V_i]$. 

%Here we will present more efficient algorithms for this task.

\paragraph{Edge Density of $H$-minor-free Graphs}
The \emph{edge density} of a graph $G$ is $|E|/|V|$. It is well-known that any $H$-minor-free graph has edge density $O(1)$, where the constant $O(1)$ depends only on $H$. Specifically, Thomason~\cite{THOMASON2001318} showed that any $K_t$-minor-free graph $G=(V,E)$ satisfies $|E| = O(t \sqrt{\log t}) \cdot |V|$. Moreover, Barenboim and Elkin~\cite{BE10} showed that given an upper bound  $d$  on the edge density of a graph $G=(V,E)$, its edge set $E$ can be oriented such that the out-degree of each vertex is at most $O(d)$ in $O(\log n)$ rounds. As a result, for any $H$-minor-free graph $G$, in $O(\log n)$ rounds we can orient its edges such that each $v \in V_i$ has out-degree $O(1)$. 

%We present the algorithm for the sake of completeness.
%Let $d = O(1)$ be an upper bound on the edge density for $H$-minor-free graphs. Then $2d$ is an upper bound on the average degree for $H$-minor-free graphs, so the number of vertices $v$ with $\deg(v) \geq 4d$ is at most $|V|/2$. 
%The following algorithm finds a required edge orientation of $G$ in $O(\log n)$ rounds. Initially $U \leftarrow V$ is the set of all vertices in $G$. In each iteration, we consider the set $S \subseteq U$ of vertices $v \in U$ whose degree in $G[U]$ is at most $4d$. For each edge $e = \{u,v\}$ with $u \in S$ and $v \in U \setminus S$, orient $e$ in the direction $u \rightarrow v$. For each edge $e = \{u,v\}$ with $u \in S$ and $v \in S$, orient $e$ in an arbitrary direction. Update $U \leftarrow U \setminus S$ and proceed to the next iteration. Since the size of $U$ is reduced by a factor of at least two after each iteration, the algorithm finishes in $O(\log n)$ rounds.

\paragraph{Information Gathering}
In view of the above discussion, the task of letting $v_i^\ast$ learn the entire graph topology of $G[V_i]$ can be reduced routing $O(1)$ messages of $O(\log n)$ bits from each $v \in V_i$ to $v_i^\ast$, as we can first spend $O(\log n)$ rounds to find an edge orientation of $G[V_i]$ with $O(1)$ out-degree, and then each vertex $v \in V_i$ only has to send information about its outgoing edges in $G[V_i]$ to $v_i^\ast$.

By \cref{lem:separator}, if $G$ is $H$-minor-free, then the degree of  $v_i^\ast$ in $G_i = (V_i, E_i)$ is $\Omega(\phi^2) |V_i|$, so the number of $O(\log n)$-bit messages sent to   $v_i^\ast$  in this routing task is $O(\phi^{-2}) \cdot \deg_{G_i}(v_i^\ast)$.
Therefore, using expander routing, this routing task can be solved in  $\phi^{-2} \cdot 2^{O(\sqrt{\log n})}$ rounds with high probability~\cite{GhaffariKS17,GhaffariL2018} or  $\phi^{-O(1)} \cdot 2^{O(\sqrt{\log n \log  \log n})}$ rounds deterministically~\cite{chang2024deterministic}.

We provide faster randomized and deterministic algorithms for this task in \cref{lem:routing-rand,lem:routing-det}. Due to \cref{lem:separator}, the conditions in \cref{lem:routing-rand,lem:routing-det} are satisfied. In these lemmas, $n$ denotes the number of vertices in the underlying network $G$, not the number of vertices in one cluster $G_i$.
The above discussion on the edge density of  $H$-minor-free graphs implies that for any $H$-minor-free graph $G$,  the degree of  $v_i^\ast$ in $G_i = (V_i, E_i)$ is $\Omega(\phi^2) |V_i| = \Omega(\phi^2) |E_i|$ by \cref{lem:separator}.

\begin{lemma}\label{lem:routing-rand}
Suppose $\deg_{G_i}(v_i^\ast)  = \Omega(\phi^2) |E_i|$. In $O(\phi^{-4} \log^3 n)$ rounds,  an $O(\log n)$-bit message from each  $v \in V_i$ can be routed via the edges $E_i$  to $v_i^\ast$ with high probability. 
\end{lemma}
\begin{proof}
The algorithm runs a lazy random walk of length $O(\phi^{-4} \log^2 n)$ from each vertex $v \in V_i$ in parallel. We claim that each step of the lazy random walk can be simulated in $O(\log n)$ rounds with probability $1 - 1/\poly(n)$, so the overall round complexity is $O(\phi^{-4} \log^3 n)$. To prove this claim, observe that for each  edge $e \in E_i$ and for each $j$, the expected number of random walks traversing $e$ in the $j$th step is $O(1)$, so a Chernoff bound implies that this number is at most $O(\log n)$ with probability $1 - 1/\poly(n)$. By a union bound over all $e \in E_i$ and $1 \leq j \leq O(\phi^{-4} \log^2 n)$, the number of $O(\log n)$-bit messages sent along each edge in each step is $O(\log n)$ with probability $1 - 1/\poly(n)$. 

For the correctness of the algorithm, we show that with probability $1 - 1/\poly(n)$ each random walk passes $v_i^\ast$, so in the end $v_i^\ast$ receives all the messages.
After $\mix = O(\phi^{-2} \log n)$ lazy random walk steps, it lands at a random vertex according to the degree distribution $\pi(u) = \deg_{G_i}(u) / 2 |E_i|$, up to a small additive error $\pm \pi(u)/n$.
In particular, it lands at $v_i^\ast$ with probability $\Omega(\deg_{G_i}(v_i^\ast) / |E_i|) = \Omega(\phi^2)$. Thus, after $s = O(\phi^{-2} \log n)$ segments of random walks of length $\mix = O(\phi^{-2} \log n)$, the probability that the walk never reach  $v_i^\ast$ is at most $(1 - \Omega(\phi^2))^s = n^{-\Omega(1)}$.
\end{proof}

\begin{lemma}\label{lem:routing-det}
Suppose $\deg_{G_i}(v_i^\ast)  = \Omega(\phi^2) |E_i|$. In $O(\phi^{-18}) \cdot 2^{O(\sqrt{\log n})}$ rounds, an $O(\log n)$-bit message from each  $v \in V_i$ can be routed via the edges $E_i$  to $v_i^\ast$ deterministically.
\end{lemma}
\begin{proof}
Although this routing task can be solved in $\phi^{-O(1)} \cdot 2^{O(\sqrt{\log n \log  \log n})}$ rounds using deterministic expander routing~\cite{chang2024deterministic}, we provide a faster and more direct algorithm via an almost maximal flow algorithm~\cite[Lemma D.10]{ChangS20}.\footnote{See the full version \href{https://arxiv.org/abs/2007.14898v1}{arXiv:2007.14898v1} of~\cite{ChangS20}.}

In order to apply~\cite[Lemma D.10]{ChangS20}, we need to do some pre-processing to the graph $G_i$.
Let $G_i'$ be the result of replacing each vertex $v$  in $G_i'$ by a $\deg_{G_i}(v)$-vertex graph $X_v$ with $\Theta(1)$ conductance and $\Theta(1)$ maximum degree in such a way that the $\deg_{G_i}(v)$ edges in $E_i$ incident to $v$ are attached to distinct $\deg_{G_i}(v)$ vertices in $X_v$.  Observe that the new graph $G_i'$ has maximum degree $O(1)$.

Define the \emph{sparsity} of a cut $S$ of a graph $G=(V,E)$ as $\Psi(S) = \frac{|\partial(S)|}{\min\{ |S|, |V \setminus S|\}}$ if $S \neq \emptyset$ and $S \neq V$. Define the   \emph{sparsity} of a graph $G=(V,E)$ as the minimum sparsity over all cuts $S \subseteq V$ with $S \neq \emptyset$ and $S \neq V$.
Then we must have $\Psi(G_i') = \Theta(\Phi(G_i)) = \Omega(\phi)$~\cite[Lemma C.2]{ChangS20}. 

Next, define $G_i''$ as the result of replacing each vertex $u \in X_{v_i^\ast}$ by  an $O(\phi^{-2})$-vertex graph $Y_u$ with $\Theta(1)$ conductance and $\Theta(1)$ maximum degree such that  $T = \bigcup_{u \in X_{v_i^\ast}} Y_u$ constitutes more than half of the vertices in $G_i''$. It is clear that the new graph $G_i''$ has sparsity $\Omega(\phi^{3})$~\cite[Lemma C.1]{ChangS20} and maximum degree $O(1)$. 

Now, let $S$ be the vertices in $G_i''$ that are not in $T$. Then the original routing problem is reduced to finding a set of paths from each $v \in S$ to an arbitrary vertex in $T$. Suppose that this set of paths satisfies that the maximum path length is $d$ and each vertex belongs to at most $c$ paths, then the routing can be done with an additional $O(cd)$ rounds.  As $|S| < |T|$,~\cite[Lemma D.10]{ChangS20} shows that such a set of paths with $c = O(\Delta \psi^{-1} \log^{3/2} n)$ and $d = O(\Delta^2 \psi^{-2}) \cdot 2^{O(\sqrt{\log n})}$ can be found in  $t = O(\Delta^6 \psi^{-6}) \cdot 2^{O(\sqrt{\log n})}$ rounds. Here $\Delta = O(1)$ is the maximum degree and $\psi = \Omega(\phi^{3})$ is the sparsity. Therefore, the overall round complexity of routing is $O(t + cd) = O(\phi^{-18}) \cdot 2^{O(\sqrt{\log n})}$.
\end{proof}

Observe that the  routing algorithms of \cref{lem:routing-rand,lem:routing-det} can also be used to deliver an $O(\log n)$-bit message from $v_i^\ast$ to each vertex $v \in V_i$ in $G_i$ by reversing the routing procedure.

\subsection{Summary}\label{sec:framework_summary}
We summarize our results as a theorem.

\begin{theorem}\label{thm:routing-main}
Given any parameter $0 < \eps < 1$, there is an algorithm for finding a partition $V = V_1 \cup V_2 \cup \cdots \cup V_k$ of the vertex set of an $H$-minor-free graph $G=(V,E)$ with the following properties.
\begin{description}
    \item[Inter-cluster Edges:] The number of inter-cluster edges % $\sum_{1\leq i \leq k}|\partial(V_i)|/2 $ 
    is at most $\eps 
    \min\{|V|, |E|\}$. %is at most $\eps$ fraction of the number of all edges $|E|$.  
    \item[Construction Time:] The round complexity for partitioning the graph is $\eps^{-O(1)} \log^{O(1)} n$ in the randomized setting and is $\eps^{-O(1)} 2^{O(\sqrt{\log n \log \log n})}$ in the deterministic setting.
    \item[Routing Time:] Each cluster $V_i$ has a leader $v_i^\ast \in V_i$ that knows the entire graph topology of $G[V_i]$. Furthermore, we can let $v_i^\ast$ exchange a distinct $O(\log n)$-bit message with each vertex $v \in V_i$ in $\eps^{-O(1)} \log^{O(1)} n$ rounds in the randomized setting and in $\eps^{-O(1)} 2^{O(\sqrt{\log n \log \log n})}$ rounds in the deterministic setting.    
\end{description}
\end{theorem}
\begin{proof}
Since $G$ is $H$-minor-free, there is a constant $t= O(1)$ depending only on $H$ such that $|E|/|V| \leq t$~\cite{THOMASON2001318}.
The partition $V = V_1 \cup V_2 \cup \cdots \cup V_k$  is constructed using the expander decomposition algorithms of \cref{thm:expander-decomposition-rand,thm:expander-decomposition-det} with parameter $\eps' = \eps / t \leq \eps$, so the requirement on the construction time is met. The upper bound on the number of inter-cluster edges $\eps' |E|  = \eps |E| /t \leq \eps |V|$ follows from the definition of an $(\eps, \phi)$-expander decomposition.

For each cluster $G_i = (V_i, E_i)$ in the expander decomposition, in $O(\phi^{-1} \log n)$  rounds the vertices in $G_i$ can select a vertex $v_i^\ast \in V_i$ that has the maximum degree in $G_i$. The algorithm for selecting  $v_i^\ast$ is as follows. In the first step, each vertex $v$ in $G_i$ broadcasts $(\ID(v), \deg_{G_i}(v))$ to its neighbors in $G_i$. After that, in each round each vertex $v$ maintains a pair $(\ID(u), \deg_{G_i}(u))$ that has the highest $\deg_{G_i}(u)$ over all pairs that $v$ has received, breaking the tie by comparing $\ID(u)$, and $v$ broadcasts this pair $(\ID(u), \deg_{G_i}(u))$ to all its neighbors in $G_i$.  The graph $G_i$ has diameter $O(\phi^{-1} \log n)$ because $G_i$ is an $\phi$-expander. Therefore, after $O(\phi^{-1} \log n)$ rounds of communication, all vertices in $G_i$ agree with the same pair $(\ID(u), \deg_{G_i}(u))$ and we may set $v_i^\ast = u$.

For  learning the graph topology of $G[V_i]$ and routing, we apply the  routing algorithms of \cref{lem:routing-rand,lem:routing-det} to $G[V_i]$, in parallel for all $1 \leq i \leq k$. In view of \cref{thm:expander-decomposition-rand,thm:expander-decomposition-det}, we use $\phi = \eps^{O(1)} \log^{-O(1)} n$  in the randomized setting and  $\phi = \eps^{O(1)} 2^{-O(\sqrt{\log n \log \log n})}$ in the deterministic setting.  Due to \cref{lem:separator}, the conditions in \cref{lem:routing-rand,lem:routing-det} are satisfied.
The requirement on the routing time is met in view of the round complexities specified in \cref{lem:routing-rand,lem:routing-det}.
\end{proof}

\paragraph{The Behavior of a Failed Execution} We briefly discuss the behavior of the algorithm of \cref{thm:routing-main} when it fails. For example, if $G$ is not $H$-minor-free, then the algorithm of \cref{thm:routing-main} might not work successfully.
Note that the choice of the parameter $t$ in the algorithm of \cref{thm:routing-main} depends only on $H$, regardless of whether the underlying graph is $H$-minor-free.

Even if $G$ is $H$-minor-free, the algorithm might fail  with a probability of $1/\poly(n)$ in the randomized setting. Understanding the behavior of a failed execution of the algorithm of \cref{thm:routing-main} is crucial to its application in property testing, which we will discuss in \cref{sec:property_testing}, as there is no guarantee that the underlying network $G$ is  $H$-minor-free.

The algorithm of \cref{thm:routing-main}  has two parts, the clustering step and the routing step.

\paragraph{Clustering Step}
In the clustering step we may assume that the algorithm always outputs a clustering  $V = V_1 \cup V_2 \cup \cdots \cup V_k$, even in a failed execution. In particular, if a vertex $v$ is not assigned to any cluster, then $v$ simply assign itself to the cluster $\{v\}$. 

In a successful execution, each  $G[V_i]$ has diameter $O(\phi^{-1} \log n)$ because $G_i$ is an $\phi$-expander and $G_i$ is the result of removing some edges from $G[V_i]$. We can also guarantee that each cluster has this property even in a failed execution, as follows. Choose  $b = O(\phi^{-1} \log n)$ be any upper bound on the cluster diameter for a successful execution of an expander decomposition algorithm. The number $b$ depends only on $\phi$ and $n$.
in $O(\phi^{-1} \log n)$ rounds, we run the following algorithm. Using $b$ rounds, each vertex $v$ computes the maximum $\ID(u)$ over all vertices $u$ within distance $b$ to $v$ in  $G[V_i]$. After that, each vertex $v$ compares its result with its neighbors in $G[V_i]$, and then $v$ marks itself $\ast$ if there is a disagreement. Finally, each vertex $v \in V_i$ checks in $2b+1$ rounds whether there is a vertex  $u \in V_i$ within distance $2b+1$ to $v$ that is marked $\ast$. If such a vertex $u$ exists, then $v$ also marks itself $\ast$. It is that there are two possible outcomes. Either all vertices in $V_i$ are marked $\ast$, or all vertices in $V_i$ are not marked $\ast$.
If the diameter of  $G[V_i]$ is at most $b$, then all vertices in $V_i$ are not marked $\ast$.
If the diameter of  $G[V_i]$ is at least $2b+1$, then all vertices in $V_i$ are marked $\ast$.
Hence if a vertex $v$ is marked $\ast$, it knows that the clustering step has failed, in which case we can let $v$ reset its cluster to be $\{v\}$.

For the number of inter-cluster edges, if $G$ is $H$-minor-free, then the upper bound $\eps \min\{|V|, |E|\}$ is always satisfied in the deterministic setting and it is satisfied with probability $1 - 1/\poly(n)$ in the randomized setting. Recall that the algorithm of \cref{thm:routing-main} is based on an expander decomposition algorithm with parameter $\eps'\leq \eps$. Since the expander decomposition algorithm does not rely on the assumption that $G$ is $H$-minor-free,   the weaker upper bound $\eps|E|$ on the number of inter-cluster edges holds regardless of whether the input graph $G$ is $H$-minor-free or not.

\paragraph{Routing Step} %{\color{red} To do.}
We distinguish between different reasons for the routing algorithms of \cref{lem:routing-rand,lem:routing-det} to fail. 
The first reason of failure is that the condition $\deg_{G_i}(v_i^\ast)  = \Omega(\phi^2) |E_i|$ for \cref{lem:routing-rand,lem:routing-det} is not satisfied. In view of \cref{lem:separator}, the only possibility that this condition is not met is when $G$ is not $H$-minor-free. In view of the above discussion, each cluster $G[V_i]$ always has diameter $O(\phi^{-1} \log n)$, so whether the condition $\deg_{G_i}(v_i^\ast)  = \Omega(\phi^2) |E_i|$ is satisfied can be checked in $O(\phi^{-1} \log n)$ rounds.

Even if the condition $\deg_{G_i}(v_i^\ast)  = \Omega(\phi^2) |E_i|$ is met, the routing algorithms of \cref{lem:routing-rand,lem:routing-det} might still fail. There are two possible reasons. One reason is that $\Phi(G_i) < \phi$ is too small due to an error in the expander decomposition algorithm in the clustering step. The other reason is because that in the randomized setting there is a small probability that the algorithm might fail. In either case, the failure occurs with probability $1 / \poly(n)$, regardless of whether $G$ is $H$-minor-free.

In an failed execution of the routing algorithms of \cref{lem:routing-rand,lem:routing-det}, only a subset of all messages are delivered. To detect a failure of delivery of a message, we can simply reverse the execution of the algorithm. Once a vertex $v \in V_i$ detects that some of its messages are not successfully delivered, it broadcasts to all vertices in $V_i$ that the routing algorithm has failed in $O(\phi^{-1} \log n)$ rounds. Hence we can assume that all vertices in a cluster $V_i$ know whether the routing algorithm is successful.

\section{Applications}\label{sec:applications}

Using \cref{thm:routing-main}, we give $\poly(1/\eps, \log n)$-round randomized algorithms and $n^{o(1)} \cdot \poly(1/\eps)$-round deterministic algorithms for various optimization, property testing, and graph decomposition problems on planar or $H$-minor-free networks. 
%Due to the space constraint, we only present a sample of applications here, namely \MaxIS{} in $H$-minor graphs and \MCM{} in planar graphs. Our results for correlation clustering (\cref{thm:clustering}), property testing (\cref{thm:testing}), low-diameter decomposition (\cref{thm:decomp}) are deferred to \cref{sec:more_app}. 
%Our full result for \MWM{} in $H$-minor free graphs %(\cref{thm:matching}) 
%is in \cref{sec:matching}.   

%computing a $(1-\eps)$-maximum matching, and $(1+\eps)$-minimum dominating set, and near-optimal solutions for other problems in planar graphs. 

\subsection{Maximum Independent Set}\label{sec:independent_set}
We design an efficient algorithm for computing a $(1 -\eps)$-approximate maximum independent set of any $H$-minor-free network by combining  \cref{thm:routing-main} with the approach of Czygrinow, Ha{\'n}{\'c}kowiak, and Wawrzyniak~\cite{czygrinow2008fast}.

Let $G$ be an $H$-minor-free graph.  Let $\alpha(G)$ denote the size of the maximum independent set of $G$. Recall that any $H$-minor-free graph has edge density $d = O(1)$, where the constant $O(1)$ depends only on $H$~\cite{THOMASON2001318}. For any $H$-minor-free graph, as $|E|/|V| \leq d$, its  minimum degree is at most $2d$.
Hence $\alpha(G) = \Theta(n)$, as an independent set $I$ of size at least $n / (2d+1)$ can be computed by repeatedly adding a minimum-degree vertex $v$ to $I$ and removing all its neighboring vertices.
For example, if $G$ is planar, then $\alpha(G) \geq n/4$ due to the four color theorem.
%the vertices can be colored by using at most 4 colors. This implies that .

Run the algorithm of \cref{thm:routing-main} on $G$ with parameter $\eps' = \eps / (2d+1)$ to partition the vertices into $V = V_1 \cup V_2 \cup \ldots \cup V_k$. For each $V_i$, we route the entire graph topology of ${G}[V_i]$ into  $v^{*}_i$ and let  $v^{*}_i$ compute the maximum independent set $I_i$ of ${G}[V_i]$ locally. Then, $v^{*}_i$ sends a message to each vertex in $V_i$ to inform if it is in $I_i$. 

Let $I = I_1 \cup \ldots \cup I_k$. For each edge $e=\{u,v\}$, if both $u$ and $v$ are in $I$, we add one of  $u$ and $v$ to $Z$. Since this can only happen if $e$ is an inter-cluster edge, we have $|Z| \leq \eps' \cdot n$.
Let $I' = I \setminus Z$. Clearly, $I'$ is an independent set. 

Using the fact that $\alpha(G) \geq n / (2d+1)$ and $\eps' = \eps / (2d+1)$, we have
\begin{align*}
    |I'| &= \left(\sum_{i=1}^{k} |I_i|\right) - |Z| \geq \alpha(G) - |Z| \geq \alpha(G) - \eps' \cdot n \geq \alpha(G) - \eps \alpha(G) = (1-\eps)\alpha(G).
\end{align*}
%\begin{align*}
%    |I'| &= \left(\sum_{i=1}^{k} |I_i|\right) - |Z| \\
%    &\geq \alpha(G) - |Z| \\
%    &\geq \alpha(G) - \eps' \cdot n \\
%    &\geq \alpha(G) - \eps \alpha(G)\\
%    &= (1-\eps)\alpha(G).
%\end{align*}
\vspace{-1.5mm}
Hence we conclude the following theorem.
\thmindependentset*

\subsection{Maximum Cardinality Matching in Planar Graphs}\label{sec:unweighted_matching}

We show that a $(1-\eps)$-approximate maximum cardinality matching of a planar network $G$ can be computed in  $\eps^{-O(1)} \log^{O(1)} n$ rounds with high probability and $\eps^{-O(1)} 2^{O(\sqrt{\log n \log \log n})}$  rounds deterministically.

Let $G$ be a planar graph.
We begin by creating a new graph $\bar{G}=(\bar{V}, \bar{E})$ from $G$ by removing some vertices from $G$ such that the sizes of the maximum matching are the same in $G$ and in $\bar{G}$. Moreover, the size of maximum matching is at least $\Omega(|\bar{V}|)$. After applying the algorithm of \cref{thm:routing-main} on $\bar{G}$, the leader $v_i^\ast$ of each cluster gathers the graph topology of the cluster and compute the maximum matching locally.

A {\it $k$-star} is a subgraph induced by the vertices $\{x, v_1, \ldots, v_k\}$ where for each $1\leq i\leq k$, $\deg(v_i) = 1$ and there is an edge  connecting $x$ and $v_i$. A {\it $k$-double-star} is a subgraph  induced by the vertices $\{x,y, v_1, \ldots, v_k\}$ where for each $1 \leq i \leq k$, $\deg(v_i) = 2$ and there are two edges $\{x, v_i\}$ and $\{y, v_i\}$. The graph $\bar{G}$ is obtained from $G$ by eliminating all 2-stars and 3-double-stars through the deletion of some vertices, as described below.

To eliminate 2-stars, every vertex $u$ with degree 1 sends a token $(u)$ to its neighbor. Then every vertex who has received more than one tokens bounces all the tokens, \emph{except one of them}, back to their originators. Vertices of degree 1 whose token was bounced back are removed from $G$. To eliminate 3-double stars, every vertex $u$ with exactly two neighbors $u_1, u_2$ sends a token $(u, (u_1, u_2))$ to their neighbors. Every vertex then aggregate the tokens based on the second coordinate, the 2-tuple $(u_1, u_2)$. If there are more than 2 tokens with the same second coordinate, all \emph{but two of them} are bounced back to their originators. Vertices whose tokens were bounced back are removed from $G$. 

Note that the eliminations of $2$-stars and $3$-double-stars do not change the size of the maximum matching. Moreover, we have the following property:

\begin{lemma}[{\cite[Lemma 6]{czygrinow2006}}]\label{lem:linear_matching} Let $G = (V,E)$ be a planar graph with $n = |V|$ and no isolated vertices. If $G$ contains no 2-stars and 3-double-stars then the size of the maximum matching of $G$ is $\Omega(n)$. \end{lemma}

By \cref{lem:linear_matching}, the size of the maximum matching of $|\bar{V}|$ is at least $c \cdot |\bar{V}|$ for some constant $c > 0$. Czygrinow, Ha{\'n}{\'c}kowiak, and Wawrzyniak~\cite{czygrinow2008fast} used \cref{lem:linear_matching} to design an  $O(\log^\ast n)$-round deterministic algorithm for computing a $(1-\eps)$-approximate maximum cardinality matching in the \local model, for any constant $\eps > 0$. 
We show that \cref{thm:routing-main} allows us to obtain efficient matching algorithms in the \congest model as well.

Now we run the algorithm of \cref{thm:routing-main} on $\bar{G}$ with parameter $\eps' = c \cdot \eps$ to partition the vertices into $\bar{V} = V_1 \cup V_2 \cup \ldots \cup V_k$. For each $V_i$, we route the entire graph topology of $\bar{G}[V_i]$ into  $v^{*}_i$ and let it compute the maximum matching $M_i$ of $\bar{G}[V_i]$ locally. We claim that the union of the matching $M = M_1 \cup   M_2 \cup \ldots \cup M_k$ is an $(1-\eps)$-approximate maximum matching.

Let $M^{*}$ be a maximum matching. Let $M^{*}_i = M^{*} \cap (V_i \times V_i)$ be $M^{*}$ restricted to $V_i$. We have
\begin{allowdisplaybreaks}
\begin{align*}
|M| &= \sum_{i=1}^{k} |M_i| \geq \sum_{i=1}^{k} |M^{*}_i| = |M^{*}| - (\mbox{\# inter-cluster $M^{*}$-edges}) \geq |M^{*}| - \eps' \cdot |\bar{V}| \\
&= |M^{*}| - \eps \cdot  c \cdot |\bar{V}| \geq |M^{*}| - \eps |M^{*}| = (1-\eps) |M^{*}|.
\end{align*}
\end{allowdisplaybreaks}
Hence we conclude the following theorem.
\begin{theorem}
A $(1-\eps)$-approximate maximum matching of a planar network $G$ can be computed in  $\eps^{-O(1)} \log^{O(1)} n$ rounds with high probability and $\eps^{-O(1)} 2^{O(\sqrt{\log n \log \log n})}$  rounds deterministically in the \congest model.
\end{theorem}

In \cref{sec:matching} we will generalize this result to the more difficult maximum weighted matching problem  for an arbitrary $H$-minor-free graph to prove \cref{thm:matching}.

\subsection{Correlation Clustering}\label{sec:correlation_clustering}

In the agreement maximization correlation clustering problem, the edge set is partitioned into $E = E^+ \cup E^-$ two parts, and the goal is to compute a clustering of the vertices $V = V_1 \cup V_2 \cup \cdots \cup V_k$ maximizing $\sum_{i=1}^k |E^+ \cap (V_i \times V_i)| + \sum_{1 \leq i < j \leq k} |E^+ \cap (V_i \times V_j)|$, which is the number of intra-cluster $E^+$-edges plus the number of inter-cluster $E^-$-edges.

Given a partition $E = E^+ \cup E^-$ of the edges in $G$, let $\gamma(G)$ denote the optimal value for the agreement maximization correlation clustering problem. Note that $\gamma(G) \geq |E|/2$ if $G$ is connected. This is because if  $|E^+|  \geq |E|/2$,  we can put each vertex as a standalone cluster and get a score of at least $|E|/2$. Otherwise, putting every vertex in the same cluster  yields a score of at least $|E|/2$.

%We assume $n = |V| \geq 2$, since otherwise the problem is trivial. Hence we have $\gamma(G) \geq |E|/2 \geq |V|/4$.

Apply the algorithm of  \cref{thm:routing-main} on $G$ with parameter $\epsilon' = \epsilon / 2$ to partition the vertices into $V_1\ldots V_k$. For each $V_i$, route the entire graph topology of ${G}[V_i]$ into $v^{*}_i$ and  let $v^{*}_i$ compute an optimal correlation clustering $\mathcal{C}_i$ of $G[V_i]$. Let $\mathcal{C}$ be the union of $\mathcal{C}_1, \mathcal{C}_2, \ldots, \mathcal{C}_k$. Let $\mathcal{C}^{*}$ be an optimal clustering and $\mathcal{C}^{*}_i$ be the restriction of $\mathcal{C}^{*}$ to $V_i$. Formally, $\mathcal{C}^{*}_i$ is constructed by adding $C \cap V_i$ to $\mathcal{C}^{*}_i$ for each cluster $C \in \mathcal{C}^{*}$ such that $C \cap V_i \neq \emptyset$.

Using the fact that $\gamma(G) \geq |E|/2$ and $\epsilon' = \epsilon / 2$, we have
\begin{align*}
\score(\mathcal{C}) &\geq \sum_{i=1}^{k}  \score(\mathcal{C}_i) \geq \sum_{i=1}^{k} \score(\mathcal{C}^{*}_i)\geq \score(\mathcal{C}^{*}) - \epsilon' |E| \geq \gamma(G) - \epsilon \gamma(G) = (1-\epsilon)\gamma(G).
\end{align*}

Hence we conclude the following theorem. Note that the requirement that $G$ is $H$-minor-free is only used in applying \cref{thm:routing-main}.

\thmclustering*

\subsection{Property Testing}\label{sec:property_testing}

We design an efficient algorithm for testing an arbitrary minor-closed property  $\mathcal{P}$ that is closed under taking disjoint union. This covers many natural graph classes, including planar graphs, outerplanar graphs, graphs with treewidth at most $w$, and $H$-minor-free graphs for a fixed {connected} graph $H$.

We pick $s$ to be the smallest positive integer such that $K_s \notin \mathcal{P}$, i.e.~the $s$-vertex clique does not have property $\mathcal{P}$. If such a number $s$ does not exist, then $\mathcal{P}$ contains the set of all cliques. Since $\mathcal{P}$ is minor-closed and any finite graph is a minor of some clique, $\mathcal{P}$ must be the trivial property that contains all graphs, in which case we have a trivial property tester that works by letting each vertex output {\sf Accept}.

%Since $\mathcal{P}$ is minor-closed, there exists a {finite} set of {forbidden minors} $\HH$ such that $G \notin \mathcal{P}$ if and only if $H \preceq G$ for some $H \in \HH$~\cite{ROBERTSON2004325}.
%For the trivial case of $\HH = \emptyset$, we can simply let each vertex output {\sf Accept}, as  $\mathcal{P}$ is the set of all graphs.

%In the subsequent discussion, we assume  $\HH \neq \emptyset$, then the problem can be reduced to testing $H$-minor-freeness for each $H \in \HH$, as we can let a vertex $v$ outputs  {\sf Accept} if its output is  {\sf Accept} for the algorithm of testing $H$-minor-freeness for each $H \in \HH$, otherwise $v$ outputs  {\sf Reject}.

\paragraph{Algorithm}
From now on we assume that $s$ exists, and we let $H = K_s$ be the $s$-clique.
%In what follows, we consider the task of property testing $H$-minor-freeness of a fixed $H$.
Our property testing algorithm applies \cref{thm:routing-main} under the assumption that the underlying graph is $H$-minor-free, and then each vertex $v$ makes its decision as follows.
\begin{itemize}
    \item Suppose that the routing algorithm of \cref{thm:routing-main} works successfully for a cluster $V_i$, then $v_i^\ast$ knows the graph topology of $G[V_i]$.
We let $v_i^\ast$ locally check whether $G[V_i]$ has property $\mathcal{P}$ and broadcast the results to all vertices in $V_i$. If $G[V_i]$ does not have property $\mathcal{P}$, then  all vertices in $V_i$  outputs {\sf Reject}, otherwise all vertices in $V_i$  outputs {\sf Accept}.
\item Suppose that the routing algorithm of \cref{thm:routing-main} does not work successfully for a cluster $V_i$. If it fails because the condition $\deg_{G_i}(v_i^\ast)  = \Omega(\phi^2) |E_i|$ is not met, then all vertices in $V_i$  outputs {\sf Reject}, otherwise all vertices in $V_i$  outputs {\sf Accept}.
\end{itemize}

Here we recall from the discussion in  \cref{sec:framework_summary} that each cluster $V_i$ is able to check whether the routing algorithm works successfully and whether the condition $\deg_{G_i}(v_i^\ast)  = \Omega(\phi^2) |E_i|$ is met. Therefore, each vertex $v$ is able to decide whether to output  {\sf Accept} or {\sf Reject} in the above algorithm.

\paragraph{Analysis}
Suppose $G$ has property $\mathcal{P}$.
Because $\mathcal{P}$ is minor-closed, $G[V_i]$ also has property $\mathcal{P}$.
Moreover, $G \in \mathcal{P}$ implies that $G$ is $H$-minor-free.
Recall the discussion in \cref{sec:framework_summary} that the condition $\deg_{G_i}(v_i^\ast)  = \Omega(\phi^2) |E_i|$ is not met only when $G$ is not $H$-minor-free. Therefore, from the description of our algorithm, all vertices will output {\sf Accept} with probability one.

Suppose $G$ is $\eps$-far from having property $\mathcal{P}$. There are two cases. The first case is that the algorithm of \cref{thm:routing-main} does not fail. Recall that the algorithm of \cref{thm:routing-main} is based on an expander decomposition algorithm with parameter $\eps'\leq \eps$. Therefore, as long as the execution of the expander decomposition algorithm is successful, the number of inter-cluster edges is at most $\eps|E|$, regardless of whether the input graph $G$ is $H$-minor-free or not. Since $G$ is $\eps$-far from having property $\mathcal{P}$, the graph $G'$ resulting from removing all inter-cluster edges also does not have property $\mathcal{P}$.
Since $G'$ is the disjoint union of all clusters  $G[V_i]$ and $\mathcal{P}$ is closed under taking disjoint union, there must be at least one cluster $G[V_i]$ that does not have property $\mathcal{P}$, so all vertices $V_i$ in this cluster will output {\sf Reject}, as required.

The second case is that the algorithm of \cref{thm:routing-main} fails.  If it fails because the condition $\deg_{G_i}(v_i^\ast)  = \Omega(\phi^2) |E_i|$ is not met, then all vertices in $V_i$  outputs {\sf Reject}, as required. As discussed in  \cref{sec:framework_summary}, the probability that  algorithm of \cref{thm:routing-main} fails due to other reasons is at most $1/ \poly(n)$, so the probability that all vertices in the graph output {\sf Accept} is at most  $1/ \poly(n)$. 
In particular, a failure in the expander decomposition algorithm might cause the number of inter-cluster edges to be significantly higher than $\eps|E|$, potentially causing all $G[V_i]$ to have property $\mathcal{P}$. Although such a failure might not be detected, it occurs with probability at most  $1/ \poly(n)$.

Hence we conclude the following theorem. Note that in the randomized setting our algorithm  has one-sided error in that all vertices output {\sf Accept}  with probability one if $G$ has property $\mathcal{P}$.

\thmtesting*

\paragraph{Lower Bound}  We give a concrete example of a minor-closed property  $\mathcal{P}$ that is \emph{not} closed under taking disjoint union and requires $\Omega(n)$ rounds to test even for constant $\eps > 0$ and in the \local model. Therefore, the  requirement in \cref{thm:testing} that the graph property is closed under taking disjoint union is, in a sense, necessary.  

Let  $\mathcal{P}$ be any graph property that is
 not closed under taking disjoint union. Then there exist two graphs $G_1 \in \mathcal{P}$ and $G_2 \in \mathcal{P}$ such that their disjoint union $G' = G_1 \cup G_2$ does not have property  $\mathcal{P}$. Therefore, if we are in the setting where we allow the underlying network $G$ to be disconnected and that the number of vertices $n=|V|$ is not a global knowledge, then it is impossible to decide whether the underlying graph has property  $\mathcal{P}$, because a vertex in $G_1$ can never know whether the underlying network $G$ is $G_1$ or $G_1 \cup G_2$.

For the case $G$ is guaranteed to be connected and $n=|V|$ is a global knowledge, we give a concrete example of a minor-closed property  $\mathcal{P}$ that is not closed under taking disjoint union and requires $\Omega(n)$ rounds to test even for constant $\eps > 0$ and in the \local model. 

Let $H$ be a $5$-clique on the vertex set $\{v_1, v_2, v_3, v_4, v_5\}$. Let $H_1$ be the result of removing the two edges $\{v_1, v_2\}$ and $\{v_1, v_3\}$ from $H$.  Let $H_2$ be the result of removing the two edges $\{v_1, v_2\}$ and $\{v_2, v_3\}$ from $H$. Let $\tilde{H} = H_1 \cup H_2$ be the disjoint union of $H_1$ and $H_2$. Let $\mathcal{P}$ be the set of all $\tilde{H}$-minor-free graphs. It is clear that  $\mathcal{P}$ is minor-closed and it is not closed under taking disjoint union. For example, we have $H_1 \in \mathcal{P}$ and $H_2 \in \mathcal{P}$, but their union $\tilde{H} = H_1 \cup H_2$  does not have property $\mathcal{P}$.
Let $t$ be any positive integer and let $s = 2t+1$. For any $i \in \{1,2\}$ and $j\in\{1,2\}$, define the graph $G_{i,j}^t$ as follows. 
\begin{itemize}
    \item Start with an $3s$-vertex path $P=(u_1, u_2, \ldots, u_{3s})$. 
    \item For each $1 \leq x \leq s$, attach a copy of $H_i$ to $u_x$ by adding an edge connecting $u_x$ and the vertex $v_1$ of $H_i$. 
    \item For each $2s+1 \leq x \leq 3s$, attach a copy of $H_j$ to $u_x$ by adding an edge connecting $u_x$ and the vertex $v_1$ of  $H_j$.
\end{itemize}

It is clear that $G_{1,1}^t$, $G_{2,2}^t$, and $G_{1,2}^t = G_{2,1}^t$ have the same number of vertices $n = 23s$ and the same number of edges $m = 27s-1$. Both $G_{1,1}^t$ and  $G_{2,2}^t$ are $\tilde{H}$-minor-free. For $G_{1,2}^t = G_{2,1}^t$, at least $s$ edges need to be removed to turn it into a $\tilde{H}$-minor-free graph, so  $G_{1,2}^t = G_{2,1}^t$ is $\eps$-far from being $\tilde{H}$-minor-free with $\eps = 1/27$.

Suppose we have a deterministic distributed property testing algorithm $\mathcal{A}$ for $\mathcal{P}$ with $\eps = 1/27$. We claim that $\mathcal{A}$ needs more than $t$ rounds on graphs with $n= 23s$ vertices. The reason is that the output of at least one vertex $v$ in $G_{1,2}^t = G_{2,1}^t$ must be {\sf Reject}. However, our construction of $G_{i,j}^t$ guarantees that the radius-$t$ neighborhood of $v$ in $G_{1,2}^t = G_{2,1}^t$ must be isomorphic to the radius-$t$ neighborhood of a vertex $u$ in  $G_{1,1}^t$ or $G_{2,2}^t$. Therefore, if the number of rounds of $\mathcal{A}$ is at most $t$, then  {\sf Reject} is also a possible output when we run $\mathcal{A}$ on $G_{1,1}^t$ or $G_{2,2}^t$, so $\mathcal{A}$ cannot be correct. This lower bound also extends to the randomized setting. We omit this extension because it is straightforward and tedious.

\subsection{Low-diameter Decompositions}\label{sec:low_diam_decomposition}
 
Using \cref{thm:routing-main}, we design an efficient algorithm that finds a partition of the vertex set $V = V_1 \cup V_2 \cup \cdots \cup V_k$ such that the number of inter-cluster edges $\sum_{1\leq i \leq k}|\partial(V_i)|/2$ is at most $\eps |E|$ and the diameter of the induced subgraph $G[V_i]$ is at most $D = O(\eps^{-1})$ for each $1 \leq i \leq k$.

We first run \cref{thm:routing-main} with parameter $\tilde{\eps} = \eps/2$ to obtain a clustering $V = V_1 \cup V_2 \cup \cdots \cup V_k$ such that the number of inter-cluster edges is at most $\tilde{\eps}  |E| \leq \eps|E|/2$. We then refine the cluster $V_i$ by letting $v_i^\ast$ compute a low-diameter decomposition of $G[V_i]$ with $\tilde{\eps}  = \eps/2$ and $\tilde{D}=  O(\tilde{\eps}^{-1})$ using any known sequential algorithm~\cite{klein1993excluded,Fakcharoenphol2003improved,Ittai2019padded} for this task. Hence each cluster in the final clustering has diameter $O(\eps^{-1})$. This step introduces at most $\eps|E|/2$ inter-cluster edges, so the total number of inter-cluster edges is at most $\eps|E|/2 + \eps|E|/2 = \eps|E|$, as required.

\thmdecomp*

%\clearpage
\section{Maximum Weighted Matching}\label{sec:matching}

In this section, we show how to apply our framework to get a $\poly(\log n, 1/\epsilon)$-round algorithm in the \congest model for $(1-\epsilon)$-\MWM{} in $H$-minor free graphs. 

\subsection{Preliminaries}

%Duan and Pettie \cite{DuanP-approxMWM} gave a linear-time algorithm for computing $(1-\epsilon)$-\MWM{} in the centralized setting. We give a $\poly(\log n, 1/\epsilon)$-round algorithm in the \congest model for $(1-\epsilon)$-\MWM{} in H-minor free graphs by combining their approach with the expander decomposition method.

The input graph is a weighted $H$-minor free graph $G = (V,E,\hat{w})$, where $\hat{w}:E \to \{1 ,\ldots, W\}$ and $\max_{X \subseteq V}|E(X)|/|V(X)| \leq C_H$ for some constant $C_H \geq 1$. Note that $W$ is an globally-known upper bound on the weight of the edges. Since we are looking for a $(1-\epsilon)$-approximation, we may assume without loss of generality that $W = \poly(n)$ (See \cite[Section 2]{DuanP-approxMWM}). Let $\epsilon' = \Theta(\epsilon)$ be a parameter that we will choose later. Assume without loss of generality that $W$ and $\epsilon'$ are powers of two. We can make the assumption because if $W$ is not a power of two, we can simply set $W$ to be the next smallest power of two. Similarly for $\epsilon'$, but for the next largest power of two. Let $L = \lg W$, $\delta_0 = \epsilon' W$ and $\delta_i = \delta_0 / 2^{i}$ for $1 \leq i \leq L$.

Similar to \cite{DuanP-approxMWM}, we maintain the following variables throughout the algorithm, with the addition of $\Delta w(\cdot)$ -- the weight modifier:

%\begin{table}[h]
{\centering
\begin{tabular}{lll}
$M$:& The set of matched edges. &\\
$y(u)$: & The dual variable defined on each vertex $u \in V$. &  \\
$z(B)$: & The dual variable defined on each $B \in \mathcal{V}_{odd}$, where $\mathcal{V}_{odd} = \{ V' \subseteq V \mid \mbox{$|V'|$ is odd} \}$. & \\
$\Omega$:& $\Omega \subseteq \mathcal{V}_{odd}$ is the set of active blossoms.\\
$\Delta w (e):$ & The weight modifier defined on each edge $e \in E$. &
\end{tabular}
}%\end{table}
%In addition, we maintain a variable denoting the weight modifier $$\Delta w : E \to R$$. 

\paragraph{Matchings and Augmenting Paths} Given a matching $M$, a vertex is {\it free} if it is not incident to any edge in $M$. An {\it alternating path} is a path whose edges alternate between $M$ and $E \setminus M$. An {\it augmenting path} $P$ is an alternating path that begins and ends with free vertices. Given an augmenting $P$, let $M\oplus P = (M \setminus P) \cup (P \setminus M)$ denote the resulting matching after we augment along $P$. Note that we must have $|M \oplus P| = |M| + 1$.

\paragraph{Dual Variables} The variables $y$ and $z$ are called the {\it dual variables}, because they correspond to the variables in the dual linear program of Edmonds' formulation \cite{Edmonds65}. For convenience, given an edge $e = uv$, we define $$yz(e)=y(u)+y(v)+\sum_{B \in \mathcal{V}_{odd} : e \in B} z(B)$$

\paragraph{Blossoms} A blossom is specified with a vertex set $B$ and an edge set $E_{B}$.  A trivial blossom is when $B = \{v \}$ for some $v \in V$ and $E_{B}=\emptyset$. A non-trivial is defined recursively: If there are odd number of blossoms $B_0 \ldots B_{\ell}$ connected as an odd cycle by $e_i \in B_{i} \times B_{[(i+1)\mod (\ell+1)]}$ for $0 \leq i \leq \ell$, then $B = \bigcup_{i=0}^{\ell} B_i$ is a blossom with $E_B = \{e_0 \ldots, e_{\ell} \} \cup \bigcup_{i=0}^{\ell} E_{B_i} $. It can be shown inductively that $|B|$ is odd and so $B \in \mathcal{V}_{odd}$. A blossom is {\it full} if $|M \cap E_{B}| = (|B| -1 )/2$. The only vertex that is not adjacent to the matched edges in a full blossom is called the {\it base} of $B$. Note that $E(B) = \{(u,v) \mid u,v \in B\}$ may contain edges not in $E_{B}$. 

A blossom is \emph{active} whenever $z(B) > 0$. Only full blossoms can become active. Thoughout the execution, $\Omega$ is the set of active blossoms, which forms a laminar (nested) family and can be represented by a set of rooted trees. The leaves of the trees are the trivial blossoms. If a blossom $B$ is defined to be the cycles formed by $B_0,\ldots, B_{\ell}$, then $B$ is the parent of $B_0,\ldots, B_{\ell}$. The blossoms that are represented by the roots are called the {\it root blossoms}. 

Given $\Omega$, let $G / \Omega$ denote the unweighted simple graph obtained by contracting all the root blossoms in $\Omega$. The following are some basic properties about the augmenting paths on the contracted graph:

\begin{lemma}(Folklore, summarized by \cite[Lemma 2.1]{DuanP-approxMWM}) Let $\Omega$ be a set of full blossoms with respect to a matching $M$.
\begin{enumerate}[leftmargin=*,itemsep=-1ex, topsep = 0pt,partopsep=1ex,parsep=1ex]
\item If $M$ is a matching in $G$, then $M / \Omega$ is a matching in $G/ \Omega$.
\item Every augmenting path $P'$ relative to $M/\Omega$ in $G /\Omega$ extends to an augmenting path $P$ relative to $M$ in $G$.
\item Let $P'$ and $P$ be mentioned as in (2). Then $\Omega$ remains a valid set of full blossoms.
\end{enumerate}
\end{lemma}

\begin{definition}[Free vertex types]\label{dfn:freevertex}
Let $\hat{F}_{s}$ and $\hat{F}_{b}$ denote the set of free vertices in $V(G)$ that not are contained in any blossoms (singleton) and contained in some blossom in $\Omega$ respectively. Let $\hat{F} = \hat{F}_s \cup \hat{F}_b$. By default, we use the notations to denote such sets of free vertices with respect to the current time of reference. 
\end{definition}

%\begin{definition}
%A {\it $k$-free star} is a subgraph $x, f_1, \ldots f_k$ where for each $1 \leq i \leq k$, $f_i \in \hat{F_s}$, $|E_{elig}(\{f_i\}, V \setminus \hat{F_s})| = 1$, and $(x,f_i) \in E_{elig}$ and $x \in V \setminus \hat{F_s}$.  A {\it $k$-double free star} is a subgraph $x, y, f_1, \ldots, f_k$ where for each $1 \leq i \leq k$, $f_i \in \hat{F_s}$, $|E_{elig}(\{f_i\}, V \setminus \hat{F}_s)| = 2$, and $(x, f_i), (y, f_i) \in E_{elig}$ and $x,y \in V \setminus \hat{F_s}$.  
%\end{definition}
\begin{definition}[Free and regular blossoms]\label{dfn:freeblossom}
For each $B \in \Omega$, we say $B$ is  {\it free} if it contains a free vertex. Otherwise, $B$ is {\it regular}.
\end{definition}

\begin{definition}[Partition of contracted graphs]\label{dfn:contracted}
Let $\Omega$ be the set of active blossoms. Given a vertex set $\hat{V} \subseteq G$, define $$\hat{V} / \Omega = \{v \in G/\Omega \mid \mbox{the set of vertices in $G$ represented by $v$ are fully contained in $\hat{V}$}\}.$$
Given a partition $(\hat{V}_1, \ldots ,\hat{V}_{k})$ of $G$. Let $(\hat{V}_1/ \Omega, \ldots ,\hat{V}_{k}/\Omega)$ be the corresponding set of disjoint vertices in $G/\Omega$. Note that it may be the case that some vertices in $G/\Omega$ are not contained in any $\hat{V}_{i}/\Omega$.

\end{definition}

\begin{figure}
\centering
\includegraphics[scale=0.15]{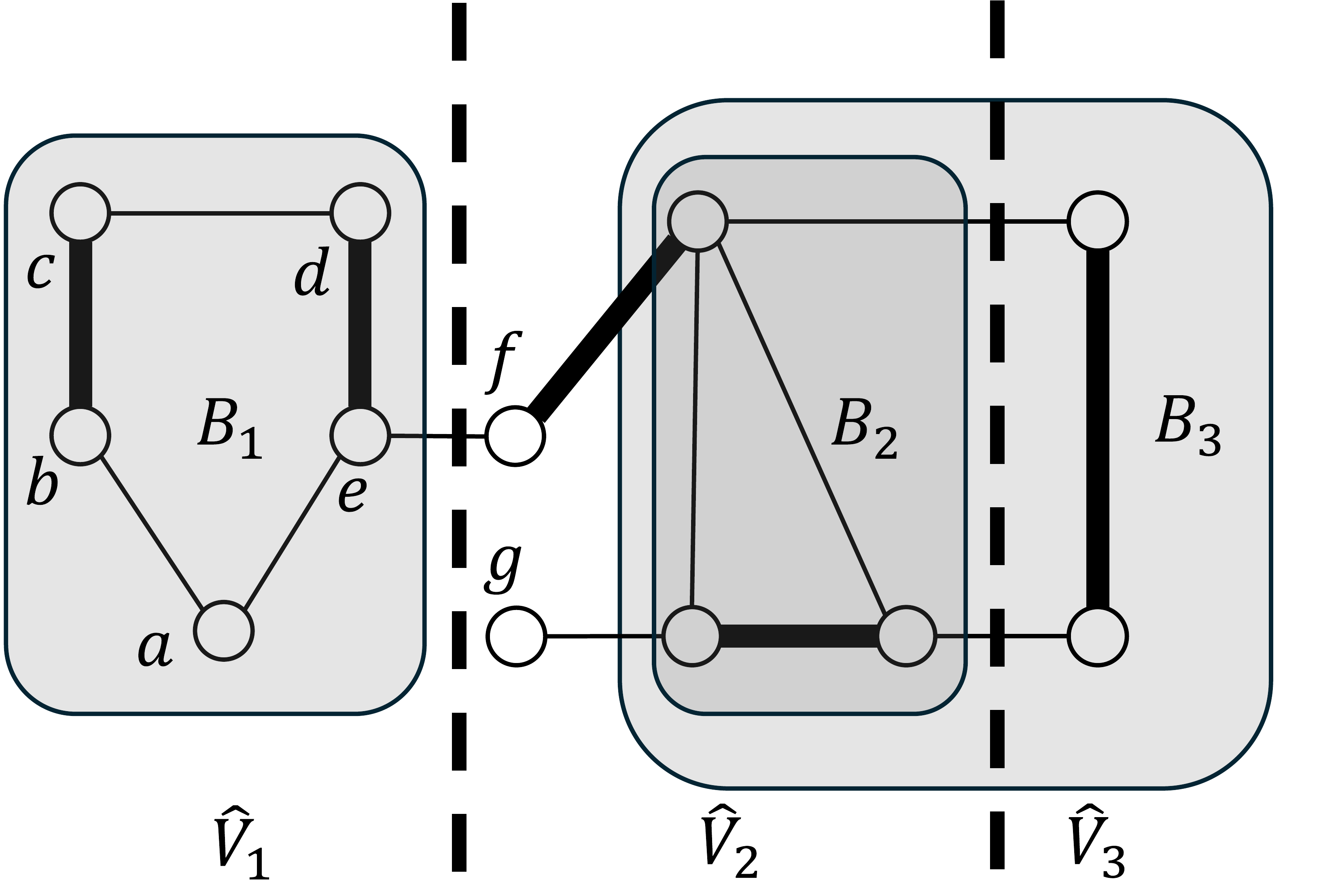}
\caption{\small An example illustrating \Cref{{dfn:freevertex}}, \Cref{dfn:freeblossom}, and \Cref{dfn:contracted}. In this example, $\Omega = \{B_1, B_2, B_3\}$ where $B_2$ is nested inside $B_3$. For the free vertex types, $\hat{F} = \{a,g\}$, $\hat{F}_s = \{g \}$ and $\hat{F}_{b} = \{a\}$. For the blossom types, $B_1$ is a free blossom, whereas $B_2$ and $B_3$ are regular blossoms. Lastly, $\hat{V}_1 / \Omega = \{a,b,c,d,e\}$, $\hat{V}_2 / \Omega = \{f,g\}$, and $\hat{V}_3 /\Omega = \emptyset$.}
\end{figure}

\paragraph{The Weights} Recall that we denote the original weights of the graph by $\hat{w}$. The weight modifier, $\Delta w$, is a new variable we introduce to effectively decouple the components in an expander decomposition. Initially $\Delta w(e)$ is set to 0 for every edge $e \in E$. The effective weight is defined to be $w(e) = \hat{w}(e) + \Delta w (e)$. The weight at scale $i$ is defined to be $w_i(e) = 2^{i} \cdot \lfloor w(e)/ 2^i \rfloor$.

\paragraph{The Invariants} The algorithm consists of $1 + \log_{2} W$ scales, where each scale consists of $O(1/\epsilon)$ iterations. We maintain the variables $M, y, z, \Omega,$ and $\Delta w$ so that they satisfy the relaxed complimentary slackness condition modified from \cite{DuanP-approxMWM} at the end of each iteration of each scale $i$:

\begin{property} (Relaxed Complementary Slackness)\label{prop:RCS}
\begin{enumerate}[topsep=0.5ex,itemsep=-.2ex]
	\item \label{RCS:1} {\bf Granularity.} $z(B)$ and $w_i(e)$ is a non-negative multiple of $\delta_i$ for all $B\in \mathcal{V}_{odd}, e \in E$ and $y(u)$ is a non-negative multiple of $\delta_i/2$ for all $u \in V$.
	\item {\bf Active Blossoms.}  \label{RCS:2} $|M\cap E_B| = \floor{|B|/2}$ for all $B\in \Omega$.  If $B \in \Omega$ is a root blossom then $z(B)>0$; if $B \notin \Omega$ then $z(B) = 0$.  Non-root blossoms may have zero $z$-values.
	\item {\bf Near Domination.} \label{RCS:3} $yz(e)\geq w_i(e)-1.5\delta_i$ for each edge $e=(u,v)\in E$. 
	\item \label{RCS:4}{\bf Near Tightness.} Call a matched edge or blossom edge type $j$ if it was last made a matched edge or a blossom edge in scale $j \leq i$. If e is such a type $j$ edge, then $yz(e) \leq w_i(e) + 3(\delta_j - \delta_i)$
	\item \label{RCS:5}{\bf Free Vertex Duals.} The sum of the $y$-value of the free vertices is at most $\tau \cdot |F| + \epsilon' \cdot \hat{w}(M^{*})$, where $M^{*}$ is an MWM w.r.t.~$\hat{w}$, $F$ is the set of free vertices, $\tau$ is a variable managed by the algorithm (not a variable in the LP) such that $y(v) \geq \tau$ for every $v$. $\tau$ will decrease to 0 when the algorithm ends. 
	 \item \label{item:boundedw} {\bf Bounded Weight Change.} The sum of $|\Delta w(e)|$ is at most $\epsilon' \cdot \hat{w}(M^{*})$.
\end{enumerate}
\end{property}
The main modifications from \cite{DuanP-approxMWM} are the following: 
\begin{itemize}[leftmargin=*]
\item We added \cref{item:boundedw} to impose an upper bound on the weight modification caused by the modifier.
\item We modified \cref{RCS:5} so that the $y$-values of the free vertices is no longer required to be equal. This is because we may freeze a small fraction free vertices to prevent their $y$-values being changed during an iteration. As a result, they are no longer required to be zero in the end. However, the sum of the $y$-values will be upper bounded in the end.
\item We loosened the constants in \cref{RCS:3} and \cref{RCS:4} because of the possible parity difference on $y$-values of the free vertices.
\end{itemize}
\paragraph{The Eligible Graph} The eligible graph in each scale $i$ is the subgraph consists of the ``tight'' edges, which are either blossoms edges or the ones that nearly violate the complementary slackness condition. Such edges are defined as follow:
  
\begin{property} At scale $i$, an edge is {\it eligible} if at least one of the following holds.
\begin{enumerate}[topsep=0.5ex,itemsep=-.5ex]
    \item $e \in E_B$ for some $B \in \Omega$.
    \item $e \notin M$ and $yz(e) \leq w_i(e) - \delta_i$
    \item \label{elig:3} $e \in M$ is a type $j$ edge and $yz(e) \geq w_i(e) + 3(\delta_j - \delta_i) - 0.5\delta_i$ 
\end{enumerate}
\end{property}

Let $E_{elig}$ be the set of eligible edges. We define the eligible graph as $G_{elig} = (V, E_{elig})$. We say $u$ is an {\it eligible neighbor} of $v$ if $(u,v) \in E_{elig}$. Note that we made a modification on \cref{elig:3} from the original definition in \cite{DuanP-approxMWM}. This is because we need to make sure that when we increase $\delta_i$ to an eligible matched edge, it will become ineligible afterwards.  

\begin{observation}\label{obs:aug_gone}
If we augment along an augmenting path $P$ in $G_{elig}/\Omega$, all the edges of $P$ will become ineligible.
\end{observation}

\begin{definition}[Inner, outer, and reachable vertices]
Let $F$ denote the set of free vertices in $G_{elig} / \Omega$. Let $V_{in}$ and $V_{out}$ denote the set of vertices in $G_{elig}/ \Omega$ that are reachable from $F$ with odd-length augmenting paths and even-length augmenting paths respectively. Let $R = V_{in} \cup V_{out}$. Let $\hat{V}_{in}, \hat{V}_{out},$ and $\hat{R}$ denote the set of original vertices in $G$ that are represented by $V_{in}, V_{out},$ and $R$ respectively. 
\end{definition}

%\clearpage

\subsection{The Algorithm}\label{sec:alg} The Duan-Pettie algorithm consists of $O(\log W)$ scales. In each scale, the goal is to reduce the $y$-value of the free vertices roughly by half, while maintaining \cref{prop:RCS} (the Relaxed Complementary Slackness condition) with respect to the scale. This is done by executing $O(1/\epsilon)$ iterations of the augmentation step, the blossom shrinking step, the dual adjustment step, and the blossom dissolution step. Each step is done on the eligible graph $G_{elig}$. However, it is unclear whether the steps can be efficiently implemented in the \congest model. The basic idea is to run expander decomposition on $G_{elig}$. Then, make the inter-component edges ineligible by adding $\delta_i$ to or subtracting $\delta_i$ from the weight modifier $\Delta w(e)$ for each inter-component edge $e$. Now each component of the eligible graph is fully contained in some component of the expander decomposition so it would be possible employ \cref{thm:routing-main} to perform those steps locally on some node in each component. 

There are a few challenges to overcome in order to implement such a approach. First, the Duan-Pettie algorithm maintains a set of active blossoms in $\Omega$. The expander decomposition may cut through the active blossoms. The trick of modifying the weights to make the edges ineligible no longer works because blossoms edges are always eligible.  To preserve the eligibility of blossom edges, we shift the cut that goes through active blossoms by the procedure $\textsc{Cut}(\cdot)$ (\Cref{alg:cut}) . In the procedure, if a cut goes through a regular blossom, the cut will be transferred to the matched edge incident to the base of the blossom. If a cut goes through a free blossom, the free vertex inside the blossom will be temporarily frozen for this iteration. This is done by adding a dummy vertex with a dummy matched edge between the dummy vertex and the free vertex. Note that the frozen free vertices are also the reason why we have to modify \cref{RCS:5}, and subsequently \cref{RCS:3} and \cref{RCS:4} in \cref{prop:RCS}. Overall, as we will show, such shifting will not increase the error. Moreover, it will effectively cut inter-component augmenting paths and prevent inter-component blossoms to be formed in the eligible graph.

\begin{figure}[ht]
\centering
\framebox{
\begin{minipage}{0.97\textwidth}
\small
\noindent $\textsc{Cut}(\hat{V}_1, \ldots, \hat{V}_k)$:\\

Let $C$ denote the inter-component edges of the partition $\{\hat{V}_i \}_{i=1}^{k}$.
    \begin{itemize}%[topsep=0.5ex,itemsep=-.2ex]
         \item  For each $e\in C$, let $B_e \in \Omega$ denote the root blossom (if any) containing $e$. \\ Define $h(e) = \begin{cases}e & \mbox{if $e$ is not contained in any active blossoms. } \\e' & \mbox{if $B_e$ is regular and $e'$ is the matched edge leaving $B_e$} \\ f & \mbox{if $B_e$ is free and $f$ is the free vertex contained in $B_e$} \end{cases} $
            \item  For each $e \in h(C) \cap E_{elig}$, set $\Delta w(e) \leftarrow \begin{cases} \Delta w(e) - \delta_i &\mbox{if $e \notin M$}\\ \Delta w(e) + \delta_i  & \mbox{if $e \in M$}\end{cases}$
            \item  For each $f \in h(C) \cap \hat{F}_{b}$, create a temporary dummy vertex $f'$ and add a temporary matched edge with weight 0 between $f$ and $f'$. 
%            \item Update $R$ and $\hat{R}$.
        \end{itemize}
\end{minipage}
}
\caption{The procedure $\textsc{Cut}(\hat{V}_1, \ldots, \hat{V}_k)$.}\label{alg:cut}
\end{figure}

Second, when running the expander decomposition with parameter $\epsilon''$ on $G_{elig}$, \Cref{thm:routing-main} only guarantees that the number of inter-component edges does not exceed $\epsilon'' |E_{elig}|$. This does not guarantee, however, the amount of adjustment to the weight modifier is small relatively to the weight of the optimal matching. Roughly speaking, we will need the number of inter-component edges to be upper bounded by $O(\epsilon'' |M|)$ in order to have such a guarantee, where $M$ is the size of the current matching. In general, it can be the case that $ |E_{elig}| \gg |M|$ as it is possible some vertex is connected to many free vertices. 

Therefore, we will have to process the graph before running the expander decomposition by removing some of the free vertices. Note that we run two expander decompositions in an iteration, one in the augmentation step, and another one in the blossom shrinking step. In the augmentation step, the idea is to remove those high-degree free vertices (more precisely, the free vertices that have many eligible incident edges) before running the expander decomposition. Using the fact that the graph is sparse, we can show the number of low-degree free vertices is linear in the number of matched vertices. Therefore, the expander decomposition will only cut $O(\epsilon'' |M|)$ edges. To incorporate the low-degree free vertices during the augmentation step, each such free vertex will choose a random neighbor and join its component. If we repeat it for enough times, they will ``cover'' all the augmentation paths because these free vertices are of low-degrees. Also, since each low-degree free vertex $f$ is connected with only one vertex $u_f$ in each repetition, $u_f$ does not have worry about having to contend with other vertex for $f$. As a result, $u_f$ only needs to remember one free vertex connected to it. This allows us to route the necessary information to a single node in each component to compute a maximal set of augmentation paths. 

In the blossom shrinking step, instead of removing low-degree free vertices, we remove the free vertices that cannot be included in any new free blossoms, before running the expander decomposition. Observe that if every neighbor of a free vertex $f$ is connected to some other free vertices in $G_{elig}$ after the augmentation step, then it is impossible for $f$ to be included in any new blossoms. This is because, for otherwise, an augmenting path would have been formed. We will show this formally in \cref{lem:no_augmenting_path}. And as we will show in \cref{lem:inter_edges2}, the number of remaining free vertices will be linear in the number of matched vertices. Since those removed free vertices cannot be included in the new blossoms, each component would be able to perform the blossom shrinking step independently.

For the dual adjustment step and the blossom dissolution and clean up step, we will continue to use the same expander decomposition computed in the blossom shrinking step. These steps can be implemented easily if one can perform information gathering and local computation inside each component of the expander decomposition {\it and} inside each root blossom. We already know that the former can done by \Cref{thm:routing-main}. For the latter, observe that when a blossom is formed, it must have been fully contained in a component of some expander decomposition in some iteration. We can then use the component to simulate the information gather inside the blossom. We will describe the implementation in more details in \Cref{sec:imple}. Our full algorithm is described in \Cref{fig:edmondssearch}.

\begin{figure}[htp]
\centering
\framebox{
\begin{minipage}{0.97\textwidth}
\small
$M \leftarrow \emptyset$, $\Omega \leftarrow \emptyset$, $\delta_0 \leftarrow \epsilon' W$, $\tau = W/2 - \delta_0/2$\\ $y(u) \leftarrow \tau$ for all $u \in V$, $z(B) \leftarrow 0$ for all $B \in \mathcal{V}_{odd}$, $\Delta w(e) \leftarrow 0$ for $e \in E$. \\
Execute scales $i = 0, 1, \ldots, L=\log_{2} W$ and return the matching $M$. \\

\underline{Scale $i$:}
\begin{itemize}[leftmargin=*]
\item[--] Repeat the following until $\tau = W/2^{i+2} - \delta_i / 2$ if $i \in [0,L)$, or until it reaches $0$ if $i = L$.

%for $(1 / \epsilon) + \chi_i$ iterations, where $\chi_i = \begin{cases} 1, & \mbox{if $i=0$ or $i=L$} \\ 2, & \mbox{if $1 \leq i \leq L-1$}\end{cases} $:

\begin{enumerate}[itemsep=0ex, leftmargin=*]
%\item {\bf Expander Cut:} \label{step:1}
%    \begin{enumerate}
        
        %For every $x,y \in V \setminus \hat{F}_s$, eliminate all 2-free stars centered at $x$ and 3-double free stars centered at $(x,y)$. Let $\hat{F}_d$ denote the deleted free vertices.
%Let $E' = \{ (x,y) \in E_{elig} \cap( \hat{F}_{\leq C_H} \times V) \mid \mbox{there exists $z \in \hat{F}_{s}$ where $z \neq x$ such that $(z,y) \in E_{elig}$} \}$. %Let $\hat{F}_{bad}$ be the isolated free vertices in $G \setminus E'$. 
%Let $G' = G \setminus \hat{F}_{\leq C_H}$.
    
\item {\bf Augmentation:} \label{step:2} \label{step:start}
    %Find a maximal set $\Psi$ of vertex-disjoint augmenting paths in $G_{elig}/ \Omega$ and set $M \leftarrow M \oplus (\bigcup_{P \in \Psi_i} P)$. 
\begin{enumerate}[leftmargin=*]
\item \label{step_exp:1} Find a maximal matching on $G_{elig}[\hat{F}_{s}]$ and then update $G_{elig}$ and the free vertices ($\hat{F}$ and $\hat{F_s}$).

\item Let $\hat{F}_{\leq C_H} \subseteq \hat{F}_s$ denote the set of singleton free vertices with at most $C_H$ eligible incident edges. 

\item \label{step_exp:3} Let $(\hat{U}_1, \ldots, \hat{U}_k)$ be an expander decomposition on $G \setminus \hat{F}_{\leq C_H}$ with parameter
$\epsilon'' = \epsilon/ (48 C^2_H \log W)$. Invoke $\textsc{Cut}(\hat{U}_1, \ldots, \hat{U}_k).$%Let $(V_1, \ldots V_k)$ be a partition on a subset of vertices in $G_{elig} / \Omega$ defined as follows. Given $v \in G_{elig} / \Omega$, let $\hat{v}$ denote the vertices represented by $v$. If $\hat{v}$ is fully contained in $\hat{V}_i$ for some $i$, then $v$ is assigned to $V_i$. Otherwise, $v$ is not assigned to any partitions.

\item Repeat the following for $K \cdot C^2_H \log n$ iterations for some constant $K>0$. At iteration $t$:
\begin{itemize}[leftmargin=*]
\item For every vertex $f \in \hat{F}_{\leq C_{H}}$, select a random neighbor $u_f$.

Set $\hat{U}^{(t)}_i \leftarrow \hat{U}_i \cup \{f \in \hat{F}_{\leq C_{H}} \mid u_f \in \hat{U}_i  \}$. 

Set $G_{elig}^{(t)} \leftarrow G_{elig} \setminus \{(f, v) \in E_{elig} \mid f \in \hat{F}_{\leq C_{H}} \} \cup \{(f, u_f) \mid f \in \hat{F}_{\leq C_{H}} \}$.

%$f$ joins component $\hat{V}_{i_f}$ where $\hat{V}_{i_f}$ is the component containing $u_f$.  Let $(\hat{V}^{(t)}_1, \ldots \hat{V}^{(t)}_k)$ be the partition after the vertices in $\hat{F}_{bad}$ joined. 
\item Let $(U^{(t)}_1, \ldots U^{(t)}_k) \leftarrow (\hat{U}^{(t)}_1/\Omega, \ldots, \hat{U}^{(t)}_k /\Omega)$. 

For each $1 \leq i \leq k$, find a maximal set $\Psi_i$ of vertex-disjoint augmenting paths in $(G^{(t)}_{elig}/ \Omega)[U^{(t)}_i]$. 

Set $M \leftarrow M \oplus (\bigcup_{i} \bigcup_{P \in \Psi_i} P)$. 

Update $G_{elig}$ and the free vertices ($\hat{F}$ and $\hat{F_s}$).
\end{itemize}
 
 \end{enumerate}

%For each $i$, find a maximal set $\Psi_i$ of vertex-disjoint augmenting paths in $(G_{elig}/ \Omega) [V_i \cup \hat{F}_s]$ and set $M \leftarrow M \oplus (\bigcup_{i} \bigcup_{P \in \Psi_i} P)$. Update $G_{elig}$.
\item {\bf Blossom Shrinking:}
%Repeat the following for $C_H$ iterations. At the $t$'th iteration, for every vertex $f \in \hat{F}_{\leq C_{H}}$, joins the component of its $t$'th neighbor. All the edges between $f$ and the vertices in the component will be included. Let $(V^{(t)}_1, V^{(t)}_2, \ldots, V^{(t)}_k)$ be the components after the free vertices (and the whole blossom they are in, if any) joined. For each $i$, let $\Omega_i$ be a maximal set of (nested) blossom in $(G_{elig}/ \Omega) [V^{(t)}_i]$. Set $z(B) \leftarrow 0$ for every $B \in \Omega_i$. Update $\Omega \leftarrow  \Omega \cup (\bigcup_{i} \Omega_i)$.
\begin{enumerate}[leftmargin=*]
\item Let $\hat{F}' \subseteq \hat{F}_s$ be the set of free vertices whose every neighbor is connected to at least two free vertices (including itself) in $\hat{F}_s$.

\item \label{blossom_step_2} Let $(\hat{V}_1, \ldots, \hat{V}_k)$ be an expander decomposition on $G \setminus \hat{F}'$ with parameter $\epsilon'' = \epsilon/ (48 C^2_H \log W)$. Invoke $\textsc{Cut}(\hat{V}_1, \ldots, \hat{V}_k).$

\item  Let $(V_1, \ldots V_k) \leftarrow (\hat{V}_1/\Omega, \ldots, \hat{V}_k /\Omega)$. For each $i$, let $\Omega_i$ be a maximal set of (nested) blossom in $(G_{elig}/ \Omega) [V_i \cup \hat{F}']$. Set $z(B) \leftarrow 0$ for every $B \in \Omega_i$. Update $\Omega \leftarrow  \Omega \cup (\bigcup_{1\leq i\leq k}\Omega_i)$. Note that $\{ \Omega_i \}_{1\leq i \leq k}$ are disjoint (see \cref{lem:no_bad_free_vertex}).

\end{enumerate}

%Note that if $B \in \Omega_{i}$ and $B' \in \Omega_{j}$ with $i\neq j$, then $B \cap B'$ must be a singleton set containing a vertex in $\hat{F}_{\leq C_{H}}$. $\mathtt{merge}(\cdot)$ is a process that convert potential overlapping blossoms into a nested structure by repeatedly finding such a pair of blossoms and putting $B$ as the child of $B'$.

%Note that $\{ \Omega_i \}_{1\leq i \leq k}$ are disjoint blossoms (see \cref{lem:no_bad_free_vertex}).
\item {\bf Dual Adjustment:} 
\begin{itemize}
	\item $\tau \leftarrow \tau - \delta_i / 2$
    \item $y(u) \leftarrow y(u) - \delta_i / 2$, if $u \in \hat{V}_{out}$
    \item $y(u) \leftarrow y(u) + \delta_i / 2$, if $u \in \hat{V}_{in}$
    \item $z(B) \leftarrow z(B) + \delta_i$, if $B \in \Omega$ is a root blossom with $B \subseteq \hat{V}_{out}$
    \item $z(B) \leftarrow z(B) - \delta_i$, if $B \in \Omega$ is a root blossom with $B \subseteq \hat{V}_{in}$
\end{itemize}
\item {\bf Blossom Dissolution and Clean Up:} \label{step:end}
\begin{itemize}
    \item Some root blossoms might have zero $z$-values after the dual adjustment step. Dissolve them by removing them from $\Omega$. Update $G_{elig}$ with the new $\Omega$.
    \item Remove all dummy vertices and edges created in $\textsc{Cut}(\cdot)$.
\end{itemize}
\end{enumerate}
\item[--] If $i \in [0,L)$, set $\delta_{i+1} \leftarrow \delta_{i} / 2$, $\tau \leftarrow \tau + 1.5\delta_{i+1}$ and $y(u) \leftarrow y(u)+ 1.5\delta_{i+1}$ for every $u \in V$.
\end{itemize}
\end{minipage}
}
\caption{The modified Duan-Pettie algorithm with distributed expander decompositions}%In the expander cut step, some active blossoms in $\Omega$ may be split into different components in the decomposition. If an active blossom contains a free vertex, the free vertex will be considered as temporarily frozen as we created a dummy vertex and a matched edge between them. 
\label{fig:edmondssearch}
\end{figure}

\subsection{The Analysis}
Throughout the analysis, we assume without loss of generality that we choose $\epsilon' = \min(\Theta(\epsilon), 1/6)$. 

\begin{lemma}\label{lem:dif}
Let $G_1 = (V,E,\hat{w})$ and $G_2 = (V,E,w)$. Let $M_1^{*}$ and $M_2^{*}$ be optimal matchings in $G_1$ and in $G_2$. Suppose that $\sum_{e \in E}|\hat{w}(e) - w(e)| \leq \gamma \hat{w}(M_1^{*})$. We have $\hat{w}(M_2^{*}) \geq (1-2\gamma)\cdot \hat{w}(M_1^{*})$.
\end{lemma}
\begin{proof}
We have
\begin{align*}
    \hat{w}(M_2^{*}) &= \sum_{e \in M_2^{*}} \hat{w}(e) \\
    &= \sum_{e\in M_2^{*}} (w(e) - (w(e) - \hat{w}(e) )) \\
    &\geq \sum_{e\in M_2^{*}} (w(e) - |w(e) - \hat{w}(e) | )\\
    &= w(M_2^{*}) - \gamma \hat{w}(M^{*}_1) \\
    &\geq w(M^{*}_1) - \gamma \hat{w}(M^{*}_1)  & \mbox{$M^{*}_2$ optimal w.r.t.~$w$}\\ 
    &\geq \hat{w}(M^{*}_1) - 2 \gamma \hat{w}(M^{*}_1) = (1-2\gamma)\cdot \hat{w}(M^{*}_1)  &&&& \qedhere
\end{align*}
\end{proof}

\begin{lemma}\label{lem:edge_weight_lower_bound} At any point, if $e$ is a matched edge or a blossom edge and satisfies the near tightness condition then $yz(e) \leq (1+6\epsilon')w(e)$. If $e$ satisfies the near domination condition at the end of scale $L$, then $yz(e) \geq  w(e) - 1.5\epsilon'$.
\end{lemma}
\begin{proof}
First, we show that if $e$ is of type $j$, it must be the case that $w(e) \geq W/ 2^{j+1} + \delta_j$. If $e$ is of type $j$, it must have become eligible at scale $j$ while unmatched. An unmatched edge can only become eligible in scale $j$ if $w_j(e) - \delta_j \geq yz(e)$. The minimum $y$-value over the vertices at scale $j$ is at least $W/2^{j+2}$. Therefore, $w(e)\geq w_{j}(e) \geq yz(e) + \delta_j \geq 2\cdot (W/2^{j+2}) + \delta_j = W/2^{j+1} + \delta_j$. 
 
Let $e$ be a matched edge or a blossom edge and satisfies the near tightness condition. Suppose that the current scale is $i$, the edge $e$ must be an edge of type $j$ for some $j \leq i$. Therefore, $yz(e) \leq w_i(e) + 3\delta_i \leq w(e) + 3\delta_j \leq w(e) + 3\epsilon' W / 2^{j} \leq w(e) + 6\epsilon' w(e) = (1+6\epsilon')w(e)$. 
 
 If $e$ satisfies the near domination condition at the end of scale $L$, then $yz(e) \geq w_L(e) - 1.5\delta_L = w(e) - 1.5\epsilon'$.
\end{proof}

\begin{lemma}
Suppose that $y, z, M, \Delta w$, and $\Omega$ satisfy the relaxed complementary slackness condition at the end of scale $L$. Then $\hat{w}(M) \geq (1-\epsilon) \cdot \hat{w}(M^{*})$.
\end{lemma}
\begin{proof}
Let $M_1^{*}$ and $M_2^{*}$ be optimal matching with respect to $\hat{w}$ and $w$. First, we have
\begin{allowdisplaybreaks}
\begin{align*}
    w(M) &= \sum_{e \in M} w(e) \\
    &\geq \sum_{e \in M}(1+6\epsilon)^{-1} \cdot yz(e)  \hspace{57mm} \mbox{near tightness \& \cref{lem:edge_weight_lower_bound}} \\
    &= (1+6\epsilon')^{-1} \cdot \left( \sum_{u \in V} y(u) + \sum_{B \in \Omega} \frac{|B| - 1}{2} \cdot z(B) - \sum_{u \in \hat{F}}y(u)\right) \\
    &\geq (1+6\epsilon')^{-1} \left(\sum_{u \in V(M_2^{*})} y(u) + \sum_{B \in \Omega} (|M_2^{*} \cap E(B)|) \cdot z(B) \right) - (1+6\epsilon')^{-1}(\epsilon' \cdot \hat{w}(M_1^{*})) \\
    &\geq (1+6\epsilon')^{-1} \cdot \left( \sum_{e\in M_2^{*}} yz(e) \right) - 2\epsilon' \cdot \hat{w}(M_1^{*}) \\
    &\geq (1+6\epsilon')^{-1} \cdot  \left(\sum_{e \in M_{2}^{*}} w(e)  - 1.5 \epsilon' \right) - 2\epsilon' \cdot \hat{w}(M_1^{*}) \\
        &= (1+6\epsilon')^{-1} \cdot  \left(\sum_{e \in M_{2}^{*}}( \hat{w}(e) - 1.5 \epsilon' + (w(e) - \hat{w}(e)))   \right) - 2\epsilon' \cdot \hat{w}(M_1^{*})\\
        &\geq (1+6\epsilon')^{-1} \cdot  \left(\sum_{e \in M_{2}^{*}} (1-1.5\epsilon')\cdot \hat{w}(e) - \sum_{e \in M_{2}^{*}} |(w(e) - \hat{w}(e))| \right) - 2\epsilon' \cdot \hat{w}(M_1^{*}) \\
        &\hspace{130.5mm} \hat{w}(e) \geq 1\\
        &\geq (1+6\epsilon')^{-1} \cdot  \left(\sum_{e \in M_{2}^{*}} ((1-1.5\epsilon')\cdot \hat{w}(e)) - \epsilon'\hat{w}(M^{*}_1) \right) - 2\epsilon' \cdot \hat{w}(M_1^{*})\\
     &\geq (1+6\epsilon')^{-1}\cdot(1-1.5\epsilon') \cdot (1-2\epsilon')  \hat{w}(M^{*}_1) - 4\epsilon' \cdot \hat{w}(M_1^{*}) \hspace{25mm} \mbox{By \cref{lem:dif}} \\
     &= (1-\Theta(\epsilon')) \hat{w}(M^{*}_1) = (1-\epsilon) \hat{w}(M^{*}_1)\qedhere
\end{align*}
\end{allowdisplaybreaks}
\end{proof}

%In the expander cut step, we first preprocess the graph as follows. Let $F$ be the set of free vertices that are not contained in any active blossoms. First, find a maximal matching on $G_{elig}[F]$ using Luby's algorithm. Then, make them matched so they become ineligible. %Let $E_{elig](A,B)$ denotes the number of neighbors of $u$  

%For every $x \in V \setminus F$, if a 2-star centered at $x$ exists we delete all the leaves until there is only one remaining. If every $x,y \in V \setminus$ where $x \neq y$, if a 3-double star centered at $(x,y)$ exists, we delete all the leaves until there are only two remaining. Let $G'_{elig}$ be the resulting graph that reachable from the remaining free vertices via alternating paths in $G_{elig} / \Omega$. 

%In the expander cut step, we perform expander decomposition on $G'_{elig}$ with parameter $\epsilon'' = \Theta(\epsilon'^2 / \log n)$. 

%Let $F$ be the set of free vertices in $G$. Let $M_{elig} = M \cap G_{elig}$ and $V_{M_{elig}}$ be the set of vertices that are adjacent to at least one eligible matched edge. Let $G'_{elig} = G_{elig}[V_{M_{elig}}]$.

%Let $G'_{elig} = G_{elig} / \Omega$ be the eligible graph obtained by contracting all the active blossoms. Let $V'_{M_{elig}} \subseteq V(G'_{elig})$ be the subset of vertices that are incident to a eligible matched edge. Let $V_{M_{elig}}$ be the vertices in $G$ obtained by expanding $V'_{M_{elig}}$. Let $G''_{elig} = G_{elig}[V_{M_{elig}}]$.

\begin{lemma}\label{lem:inter_edges}
The number of inter-component edges in the expander decomposition, $(\hat{U}_1, \ldots, \hat{U}_k)$, of the augmentation step is at most $3\epsilon''  C^2_H |M|$, where $M$ denotes the current matching.
\end{lemma}
\begin{proof}
Recall that $\hat{F}_s \subseteq \hat{F}$ denote the set of free vertices that are not contained in any active blossoms.  We will show that $\hat{F}_s \setminus \hat{F}_{\leq C_{H}} = C_H \cdot (|V \setminus \hat{F}_{s}|)$. Recall that $C_H \geq 1$ is the constant such that for all subgraphs $X \subseteq G$, $|E(X)| / |V(X)| \leq C_H$.

Let $\hat{F}_{>C_H}$ be $\hat{F}_s \setminus \hat{F}_{\leq C_H}$. First, we show that $|\hat{F}_{>C_H}| \leq C_H \cdot |V \setminus \hat{F}_s|$. Since  every vertex in $F_{>C_H}$ is incident to at least $C_H+1$ vertices in $V \setminus \hat{F}_s$ in $G_{elig}$, there are at least $(C_H+1)\cdot |\hat{F}_{>C_H}|$ edges in $G_{elig}[\hat{F}_{>C_H} \cup V \setminus \hat{F}_s]$. However, since $G$ is $H$-minor free, there are at most $C_H \cdot |\hat{F}_{>C_H} \cup V \setminus \hat{F}_s|$ edges in $G_{elig}[\hat{F}_{>C_H} \cup V \setminus \hat{F}_s]$. This implies 
\begin{align*} (C_H+1)\cdot |\hat{F}_{>C_H}| &\leq C_H \cdot |\hat{F}_{>C_H} \cup V \setminus \hat{F}_s|\\ 
 (C_H+1)\cdot |\hat{F}_{>C_H}| &\leq C_H \cdot (|\hat{F}_{>C_H}| + |V \setminus \hat{F}_s|)\\
|\hat{F}_{>C_H}| &\leq C_H \cdot  |V \setminus \hat{F}_s|.
\end{align*}

Note that since every vertex in $V \setminus \hat{F}$ is matched, we have $|V \setminus \hat{F}| = |M|/2$. For a free vertex in $\hat{F} \setminus \hat{F}_s$, it must be contained in an active blossom. The smallest active blossom contains at least one matched edges so $|\hat{F} \setminus \hat{F}_s| \leq |M|$. Therefore, 
\begin{align*}
    |V \setminus \hat{F}_{\leq C_{H}}| &= |V \setminus \hat{F}_s| + |\hat{F}_s \setminus \hat{F}_{\leq C_{H}}| \\
    & \leq (C_H + 1)\cdot |V \setminus \hat{F}_s| \\
    &= (C_H + 1) \cdot ( |V \setminus \hat{F}| + |\hat{F}\setminus \hat{F}_s|) \\
    &\leq (C_H + 1) \cdot (|M| / 2 + |M| ) \leq 3 C_H \cdot |M|
\end{align*}

The expander decomposition cuts at most $\epsilon''|E(G_{elig}[V \setminus \hat{F}_{\leq C_{H}}])|$ edges. This is at most 
\[\epsilon''|E(G_{elig}[V \setminus \hat{F}_{\leq C_{H}}])| \leq \epsilon'' \cdot C_H |V \setminus \hat{F}_{\leq C_{H}}| \leq  3\epsilon''  C^2_H |M|.  \qedhere\]
\end{proof}

\begin{lemma}\label{lem:inter_edges2}
The number of inter-component edges in the expander decomposition, $(\hat{V}_1, \ldots, \hat{V}_k)$, of the blossom shrinking step is at most $3\epsilon''  C_H |M|$, where $M$ denotes the current matching (excluding the dummy matched edges added during the augmentation step).
\end{lemma}
\begin{proof}
By definition, a vertex $f \in \hat{F}_s \setminus \hat{F}'$ must be adjacent to some vertex in $V \setminus \hat{F}_s$ such that $f$ is the only free vertex in $F_s$ adjacent to it. Thus, $|\hat{F}_s \setminus \hat{F}'| \leq |V \setminus \hat{F}_{s}|$. 

Since every vertex in $V \setminus \hat{F}$ is matched, we have $|V \setminus \hat{F}| = |M|/2$. For a free vertex in $\hat{F} \setminus \hat{F}_s$, it must be contained in an active blossom. The smallest active blossom contains at least one matched edges so $|\hat{F} \setminus \hat{F}_s| \leq |M|$. Therefore, 
\begin{align*}
    |V \setminus \hat{F}'| &= |V \setminus \hat{F}_s| + |\hat{F}_s \setminus \hat{F}'| \\
    & \leq 2\cdot |V \setminus \hat{F}_s| \\
    &= 2 \cdot ( |V \setminus \hat{F}| + |\hat{F}\setminus \hat{F}_s|) \\
    &\leq 2 \cdot (|M| + |M| / 2) \leq 3 \cdot |M|
\end{align*}

The expander decomposition cuts at most $\epsilon''|E(G_{elig}[V \setminus \hat{F}'])|$ edges. This is at most 
\[\epsilon''|E(G_{elig}[V \setminus \hat{F}'])| \leq \epsilon'' \cdot C_H |V \setminus \hat{F}'| \leq  2\epsilon''  C_H |M|. \qedhere\]
\end{proof}

\begin{lemma} Let $M^{*}$ be an optimal matching in $G$ w.r.t. $\hat{w}$. At any point of the algorithm, $\hat{w}(M^{*}) \geq |M| \cdot (W / 2^{i+2})$. \end{lemma}
\begin{proof}
We have:
\begin{align*}
\hat{w}(M^{*}) &\geq \sum_{e \in M} \hat{w}(e)	 \\
&\geq \left(\sum_{e \in M} w(e)\right) -  \epsilon' \hat{w}(M^{*}) && \mbox{By \cref{prop:RCS}(\ref{item:boundedw})} \\
&\geq |M| \cdot (W / 2^{i+1}) - \epsilon' \hat{w}(M^{*}) && \mbox{By \cref{lem:edge_weight_lower_bound}}
\end{align*}
Therefore, $\hat{w}(M^{*}) \geq |M| \cdot (W/2^{i+1}) / (1+\epsilon') \geq W/2^{i+2}$.
\end{proof}

\begin{lemma}\label{lem:boundedchange} In each iteration of each scale, $\sum_{e \in E}| \Delta w(e)|$ change by at most $ 24\epsilon'' \epsilon' C^2_H \hat{w}(M^{*})$. Moreover, the sum of the $y$-values of the free vertices decreases by at least $|\hat{F}|\cdot (\delta_i / 2) - 24\epsilon'' \epsilon' C^2_H \hat{w}(M^{*})$.

\end{lemma}
\begin{proof}
By \cref{lem:inter_edges} and \cref{lem:inter_edges2}, since the number of inter-component edges in both expander decompositions is at most $3\epsilon'' C^2_H |M| + 3\epsilon'' C_H |M| \leq 6\epsilon'' C^2_H |M|$, the total change on $\sum_{e \in E}| \Delta w(e)|$ is at most 
\begin{align*}
6 \epsilon'' C^2_H |M| \cdot \delta_i &\leq 6 \epsilon'' C^2_H (2^{i+2}/W)\cdot \hat{w}(M^{*}) \cdot \delta_i \\
&= 6 \epsilon'' C^2_H (2^{i+2}/W)\cdot \hat{w}(M^{*}) \cdot (\epsilon' W / 2^{i}) \\
&= 24 C^2_H \epsilon'' \epsilon' \hat{w}(M^{*}).
\end{align*}

Similarly, since the number of inter-component edges is at most $6 \epsilon'' C^2_H |M|$, the number of free vertices that are frozen is also at most $6\epsilon'' C^2_H |M|$. Each free vertex in $\hat{F}$ besides those who are frozen is supposed to have its $y$-value decrease by $\delta_i/2$. Also note that a frozen free vertex could have its $y$-value increase by $\delta_i/2$, if it becomes inner. Therefore, $\sum_{f \in \hat{F}} y(f)$ decreases by at least \begin{align*} |\hat{F}| \cdot (\delta_i / 2) - 6\epsilon'' C^2_H |M| \delta_i &\geq  |\hat{F}| \cdot (\delta_i / 2) - 6\epsilon'' C^2_H (2^{i+2}/W)\cdot \hat{w}(M^{*}) \cdot (\epsilon' W / 2^{i}) \\
&\geq  |\hat{F}| \cdot (\delta_i / 2) - 24\epsilon'' \epsilon' C^2_H \cdot \hat{w}(M^{*}) \qedhere \end{align*}
\end{proof}

\begin{corollary}\cref{prop:RCS}(\ref{RCS:5})(\ref{item:boundedw}) holds throughout the algorithm. \end{corollary}

\begin{proof}
For \cref{prop:RCS}(\ref{RCS:5}), note that there are at most $(L+1)\cdot (1/\epsilon + 2) \leq 2 (\log W)/\epsilon$ iterations in total. By \cref{lem:boundedchange}, $\sum_{e \in E}|\Delta w(e)|$ changes by at most $24\epsilon'' \epsilon' C^2_H \hat{w}(M^{*})$. By triangle inequality, we know the total change of  $\sum_{e \in E}|\Delta w(e)|$ throughout the algorithm is at most $24\epsilon'' \epsilon' C^2_H \hat{w}(M^{*}) \cdot 2 (\log W)/\epsilon \leq \epsilon' \hat{w}(M^{*})$ as $\epsilon'' = \epsilon / (48C^2_H \log W )$.

For \cref{prop:RCS}(\ref{item:boundedw}), let $\hat{F}$ denote the set of free vertices at this moment. Note that if $f \in \hat{F}$ is never frozen, we must have $y(f) = \tau$. By \cref{lem:boundedchange}, the frozen effects cause $\sum_{f \in \hat{F}} y(f)$ to increase at most $24\epsilon'' \epsilon' C^2_H \hat{w}(W^{*})$ per iteration. Since there are at most $2 (\log W)/ \epsilon$ iterations in total, we must have $\sum_{f \in \hat{F}} y(f) \leq \tau \cdot |\hat{F}| +  24\epsilon'' \epsilon' C^2_H \hat{w}(W^{*}) \cdot 2 (\log W)/ \epsilon \leq \tau \cdot |\hat{F}|+  \epsilon' \hat{w}(M^{*})$.
\end{proof}

\begin{lemma}Suppose that  \cref{prop:RCS}(\ref{RCS:1})--(\ref{RCS:4}) holds at the end of scale $i$, they must also hold at the beginning of scale $i+1$. \end{lemma}
\begin{proof}
It is easy to see that \cref{prop:RCS}(\ref{RCS:1})(\ref{RCS:2}) are automatically satisfied as $\delta_{i+1} = \delta_{i} / 2$.  

For \cref{prop:RCS}(\ref{RCS:3}) (near domination), let $yz(e)$ and $yz'(e)$ denote the $yz$-value of $e$ prior to and after the update respectively in the last line of Scale $i$. We have:
\begin{align*}
yz'(e) &= yz(e) + 3\delta_{i+1} && \mbox{the $y$-value of both endpoints increase by $1.5\delta_{i+1}$}\\
&\geq w_i(e) - 1.5\delta_{i} + 3\delta_{i+1} && \mbox{near domination at the end of scale $i$} \\
&\geq w_{i+1}(e) - \delta_{i+1} - 1.5\delta_{i} + 3\delta_{i+1}  \\
&\geq w_{i+1}(e) - 1.5\delta_{i+1} 
\end{align*}

For \cref{prop:RCS}(\ref{RCS:4}) (near tightness), let $yz(e)$ and $yz'(e)$ denote the $yz$-value of $e$ prior to and after the update respectively in the last line of Scale $i$. Suppose that $e$ is of type $j$. We have:
\begin{align*}
yz'(e) &= yz(e) + 3\delta_{i+1} && \mbox{the $y$-value of both endpoints increase by $1.5\delta_{i+1}$}\\
&\leq w_i(e) + 3(\delta_j - \delta_i) + 3\delta_{i+1} && \mbox{near tightness at the end of scale $i$} \\
&\leq w_{i+1}(e) + 3(\delta_j - \delta_i + \delta_{i+1})  \\
&\leq w_{i+1}(e) + 3(\delta_j - \delta_{i+1}) &&&& \qedhere
\end{align*}
\end{proof}

%Let $C$ denote the inter-component edges of the expander decomposition. We define a mapping function $h: C \to E(G) \cup F$ as follows: If $e \notin E(B)$ for any $B \in \Omega$ then $h(e) = e$. If $e \in E(B)$ for some $B \in \Omega$. Let $B$ be the root blossom where $e \in E(B)$. If $B$ contains a free vertex $f$ then $h(e) = f$. Otherwise, $h(e)$ is defined to be the matched edge incident to the base of $B$.

%Let $C'$ be $h(C) \cap E(G)$ and let $F''$ be $h(C) \cap F$. For each $e \in C'$, if $e$ is matched we subtract $\delta_i$ to $\Delta w(e)$. If $e$ is unmatched we add $\delta_i$ to $\Delta w(e)$.  For each vertex $f \in F''$, create a dummy vertex $u_f$ and add a matched edge with weight $0$ between $u_f$ and $f$. 

%Next we will show that Step \ref{step:start} to Step \ref{step:end} can be done in each component of the expander decomposition locally. 

\begin{lemma}\label{lem:aug0} With high probability, for every augmenting path $P$ in $G_{elig}/\Omega$ there exists an iteration $t$ of the augmentation step where $P$ is contained in $G^{(t)}_{elig}/\Omega$.  \end{lemma}

\begin{proof}
Let $P_{wxyz}$ denote the set of augmenting paths of $G_{elig}/\Omega$ where $(w,x)$ is the first edge and $(y,z)$ is the last edge. Note that it may be possible that $x = y$. Note that if both $(w,x)$ and $(y,z)$ are in  $G^{(t)}_{elig}/\Omega$ then we know all paths in $P_{wxyz}$ are contained in $G^{(t)}_{elig}/\Omega$.

Given an iteration $t$, we have:
\begin{align*}
&\Pr((w,x), (y,z) \in G^{(t)}_{elig}/\Omega) \\
&\geq \Pr((w,x)\in G^{(t)}_{elig}/\Omega) \cdot \Pr((y,z)\in G^{(t)}_{elig}/\Omega) \\
&\geq \frac{1}{C_H} \cdot \frac{1}{C_H} = \frac{1}{C^2_H}
\end{align*}
Therefore, 
$$\Pr( \{(w,x), (y,z)\} \not\subseteq G^{(t)}_{elig}/\Omega \mbox{ for all $t$}) \leq \left(1-\frac{1}{C^2_H}\right)^{K C^2_H \log n} = 1/\poly(n)$$
By taking an union over all $n^4$ possible combinations of $w,x,y,z$, we conclude w.h.p.~every augmenting path is contained in $G^{(t)}_{elig}/\Omega$ for some $t$.
\end{proof}

\begin{lemma}\label{lem:no_augmenting_path} After the augmentation step, w.h.p.~there are no augmenting paths in $G_{elig} / \Omega$.  \end{lemma}
\begin{proof}
If there exists an augmenting path $P \in G_{elig} / \Omega$ after the augmentation step, $P$ must be also be an augmenting path in $G_{elig} / \Omega$ before the augmentation step due to  \cref{obs:aug_gone}. By \cref{lem:aug0}, w.h.p.~there exists $t$ where $P$ is contained in $G^{(t)}_{elig}/\Omega$. The only reason that $P$ is not found is because that $P$ is not contained entirely in $G[U^{(t)}_i]$ for some $i$. This can happen when some vertex $v \in P$ is not in any $U^{(t)}_i$ or $P$ goes across two components $U^{(t)}_i$ and $U^{(t)}_j$ where $i\neq j$. 

If there exists a vertex $v \in P$ such that $v$ is not in any $U^{(t)}_i$, it must be the case that $v$ is a blossom that is cut internally by the cut of the expander decomposition. Let $e$ an edge in the blossom that is cut. If $v$ is an internal vertex of $P$, $h(e)$ must be the matched edge that connects $v$ and $u$, where $u$ is also in $P$. Since we have added $\delta_i$ to $\Delta w(h(e))$, $h(e)$ is no longer eligible and $P$ must not be an eligible augmenting path. If $v$ is not an internal node of $P$, then $v$ must be a contracted free blossom. Let $f$ be the free vertex in $v$.  Since we have temporarily frozen $f$ by putting a dummy matched edge $f'f$, $P$ cannot be an augmenting path.

Otherwise, $P$ goes across two component of the expander decomposition. It must be the case that there exists two consecutive vertices $x,y \in P$ where $x \in U^{(t)}_i$ and $y \in U^{(t)}_j$ for some $i \neq j$. If $xy$ is matched, let $\widehat{xy}$ denote the matched edge in $G$ represented by $xy$. In this case, we have added $\delta_i$ to the $\Delta w$-value of $\widehat{xy}$. If $xy$ is unmatched, let $\widehat{xy}$ denote the set of all unmatched edges in $G$ that are represented by $xy$. We have subtracted $\delta_i$ from the $\Delta w$-value of all the edges in $\widehat{xy}$. Therefore, in any case, edges between $x$ and $y$ are no longer eligible so either $P$ cannot be an alternating path in $G_{elig}/ \Omega$. 
\end{proof}

\begin{lemma} \label{lem:no_bad_free_vertex}
In the blossom shrinking step, for each $1 \leq i \leq k$ and for every $f \in \hat{F}'$, $f \notin B$ for all $B \in \Omega_i$. Therefore, $\{\Omega_i\}_{1\leq i \leq k}$ are disjoint.
\end{lemma}
\begin{proof}
Suppose to the contrary that $B \in \Omega_i$ is a blossom formed during the blossom shrinking step with $f \in B$ for some $f \in \hat{F}'$. 

Given a vertex $v$, we define the {\it free degree} to be the number of eligible neighbors in $\hat{F}_{s}$ . Let $x \in B$ be a neighbor of $f$. Since $f \in \hat{F}'$, the free degree of $x$ must be at least 2. Let $f' \in \hat{F}_s$ be the other free vertex that is an eligible neighbor of $x$. If blossoms $B$ is formed, then there exists an augmenting path from $f$ to $f'$ in $G /\Omega_{old}$ where $\Omega_{old}$ is the set of active blossoms before the blossom shrinking step. By \cref{lem:no_augmenting_path}, this is a contraction. 
\end{proof}

\begin{lemma}\label{lem:outer_outer}
There are no eligible edges between vertices in $V_{out}$ of $G_{elig}/\Omega$ after the Blossom Shrinking step.
\end{lemma}
\begin{proof}
Suppose to the contrary that there exists $u,v \in V_{out} \subseteq V(G_{elig}/\Omega)$ such that $(u,v) \in E(G_{elig}/ \Omega)$. Since $u,v \in V_{out}$, $(u,v)$ must be an eligible unmatched edge.  

First we claim that we may assume w.l.o.g.~there exists alternating paths $P_u$ and $P_v$  in $G_{elig}/ \Omega$ from the same free vertex to $u$ and $v$. By the definition of $V_{out}$, there exists two alternating paths $P'_u$ and $P'_v$ from some free vertices to $u$ and $v$. It must be the case that $P'_u$ and $P'_v$ shares at least one vertex. Otherwise, an augmenting path will be formed by concatenating $P'_u$, $(u,v)$ and $P'^{r}_v$, where $P'^{r}_v$ denotes the reverse of $P'_v$, which is impossible after the augmentation step by \cref{lem:no_augmenting_path}.  Let $x$ denote the first vertex on $P'_u$ such that $x \in P'_u \cap P'_v$. Let $P'_u(x)$ and $P'_v(x)$ denote the prefixes of $P'_u$ and $P'_v$ from the beginning of the paths to the first occurrence of $x$. Note that $|P'_u(x)|$ and $|P'_v(x)|$ must have the same parity; otherwise an augmenting path will be formed by concatenating $P'_u(x)$ and $P'^{r}_u(x)$.  Since 
$|P'_u(x)|$ and $|P'_v(x)|$ have the same parity, we may replace $P'_v(x)$ in $P'_v$ by $P'_u(x)$. Let $P_v$ be the resulting path and let $P_u = P'_u$. $P_u$ and $P_v$ must be alternating paths starting from the same free vertex that reaches $u$ and $v$.

Let $\Omega_{old}$ denote the set of blossoms before the Blossom Shrinking step. Let $\bar{P}_u$ and $\bar{P}_v$ denote some mapping of the alternating paths $P_u$ and $P_v$ from $G / \Omega$ to $G/ \Omega_{old}$.

Let $V(\bar{P}_u \cup \bar{P}_v)$ denote the vertices that on $\bar{P}_u$ and $\bar{P}_v$. Let $\hat{V}(\bar{P}_u \cup \bar{P}_v)$ denote the original vertices in $G$ represented by those in $V(\bar{P}_u \cup \bar{P}_v)$. If $V(\bar{P}_u \cup \bar{P}_v)$ is fully contained in some $V_i \cup \hat{F}'$ for some $i$, then $\Omega_i$ would not be a maximal set of nested blossoms in $V_i\cup \hat{F}'$. Therefore, it must be the case that either there exists $s \in V(\bar{P}_u \cup \bar{P}_v)$ such that $s$ does not belong to $V_i \cup \hat{F}'$ for all $i$ or there exists an edge $xy \in \bar{P}_u$ or $xy \in \bar{P}_v$ such that $x$ and $y$ are in different components of the expander decomposition.

Suppose there exists $s \in V(\bar{P}_u \cup \bar{P}_v)$ such that $s \notin V_i \cup \hat{F}'$ for any $i$, $s$ must be a contracted blossom that was cut internally by the cut of the expander decomposition. Let $\hat{s}$ denotes the set of original vertices that $s$ is representing. It must be the case that some edge $e \in E(\hat{s})$ is cut in the expander decomposition. If $s$ is the starting vertex of $\bar{P}_u$ and $\bar{P}_v$, then $s$ must be a contracted free blossom. Let $f$ be the free vertex in $\hat{s}$. Since we have added a dummy vertex $f'$ that is matched to $f$, $\bar{P}_u$ and $\bar{P}_v$ are not alternating paths that start from a free vertex. This contradicts with our assumption.  If $s$ is not the starting vertex, let $s' s$ be the matched edge incident to $s$. Note that $s's$ must be in $\bar{P}_u$ or $\bar{P}_v$. Let $\hat{s}' \hat{s}$ denote the edge in the original graph represented by $s's$. Since we have added $\delta_i$ to $\Delta w(\hat{s}'\hat{s})$, $s's$ must have become ineligible. This implies either $\bar{P}_u$ or $\bar{P}_v$ is not an alternating path in $G / \Omega_{old}$.

Otherwise, suppose there exists an edge $xy \in \bar{P}_u$ or $xy \in \bar{P}_v$ such that $x$ and $y$ are in different components of the expander decomposition.  If $xy$ is matched, let $\widehat{xy}$ denote the matched edge in $G$ represented by $xy$. In this case, we have added $\delta_i$ to the $\Delta w$-value of $\widehat{xy}$. If $xy$ is unmatched, let $\widehat{xy}$ denote the set of all unmatched edges in $G$ that are represented by $xy$. We have subtracted $\delta_i$ from the $\Delta w$-value of all the edges in $\widehat{xy}$.  Therefore, in any case, edges between $x$ and $y$ are no longer eligible so either $\bar{P}_u$ or $\bar{P}_v$ cannot be an alternating path in $G_{elig}/ \Omega_{old}$.  
\end{proof}

\begin{lemma}Suppose that \cref{prop:RCS}(\ref{RCS:1})--(\ref{RCS:4}) holds in the beginning of an iteration, they must also hold at the end of the iteration. \end{lemma}
\begin{proof}
\cref{prop:RCS}(\ref{RCS:1}) is satisfied because throughout the iteration, $z(B)$ and $\Delta w_i(e)$ changes only by multiples of $\delta_i$ and $y(u)$ changes by multiples of $\delta_i/2$. 

Then we claim \cref{prop:RCS}(\ref{RCS:2}) is satisfied. First we argue that each $B \in \Omega$ is {\it full} (i.e .$|M \cap E_{B}| = \lfloor |B| / 2\rfloor$). This is because the augmentation step does not affect the fullness of the blossoms and all the new blossoms created during the Blossom Shrinking step must be full.  Now observe that all the blossoms at the beginning of the iteration must have non-zero $z$-values. Moreover, all the zero $z$-value root blossoms formed during the Blossom Shrinking step must be outer. It implies that all the $z$-values of the root blossoms can only be non-negative after the Dual Adjustment. Then all the root blossoms with zero $z$-values are dissolved after the Blossom Dissolution step. Also note that the dummy vertices never join any blossoms during the iteration.

Now we show that \cref{prop:RCS}(\ref{RCS:3},\ref{RCS:4}) are satisfied. In the augmentation step, when we switch the status of an eligible matched edge to unmatched or an eligible unmatched edge to matched, \cref{prop:RCS}(\ref{RCS:3}) and \cref{prop:RCS}(\ref{RCS:4}) are still guaranteed to hold. This is because changing a matched edge to unmatched does not impose additional constraints. On the other hand, an eligible unmatched edge $e$ at scale $i$ must satisfy $yz(e) \leq w_i(e) - \delta_i$. This means that if we change the status of $e$ to become matched, it also satisfy \cref{prop:RCS}(\ref{RCS:4}) (near tightness), which is required for matched edges.

Moreover, in the invocations of $\textsc{Cut}(\cdot)$, when we adjust $\Delta w$, we only add $\delta_i$ to the eligible matched edges and subtract $\delta_i$ from eligible unmatched edges. If $e$ is an eligible matched edge of type $j$, it satisfy $w_i(e) + 3(\delta_j - \delta_i ) - 0.5 \delta_i \leq yz(e) \leq w_i(e) + 3(\delta_j - \delta_i)$. After we add $\delta_i$ to $\Delta w(e)$, we must have $w_i(e) + 3(\delta_j - \delta_i) - 1.5\delta_i \leq yz(e) \leq w_i(e) + 3(\delta_i - \delta_j) - \delta_i$. This means $yz(e)$ still satisfies both near domination and near tightness. If $e$ is an eligible unmatched edge, subtracting $\delta_i$ from $\Delta w(e)$ would only make near domination easier to satisfy.

Now we show that the dual adjustment step also maintains the near tightness property and the near domination property. Consider an edge $e=uv$. If both $u$ and $v$ are not in $\hat{V}_{in} \cup \hat{V}_{out}$ or both $u$ and $v$ are in the same root blossom in $\Omega$, then $yz(e)$ is unchanged, which implies near tightness and near domination remain to be satisfied. The remaining cases are as follows:

\begin{enumerate}
\item $e \notin M$ and at least one endpoint is in $\hat{V}_{in} \cup \hat{V}_{out}$. If $e$ is ineligible, then $yz(e) > w_i(e) - \delta_i$. Since both $yz(e)$ and $w_i(e)$ are multiples of $\delta_i/2$, it must be the case that $yz(e) \geq w_i(e) - 0.5\delta_i$ before the adjustment. Since $yz(e)$ can decrease at most $\delta_i$ during the adjustment, we have $yz(e) \geq w_i(e) - 1.5\delta_i$ after the adjustment. If $e$ is eligible, then at least one of $u, v$ is in $\hat{V}_{in}$ due to \cref{lem:outer_outer}. Therefore, $yz(e)$ cannot be reduced, which preserves \cref{prop:RCS}(\ref{RCS:3}).

\item $e \in M$ and at least one endpoint is in $\hat{V}_{in} \cup \hat{V}_{out}$. If $e$ is ineligible, we have $yz(e) < w_i(e) + 3(\delta_j - \delta_i) - 0.5\delta_i$. Since both $yz(e)$ and $w_i(e)$ are multiples of $\delta_i / 2$, it must be the case that $yz(e) \leq w_i(e) + 3(\delta_i-\delta_j) - \delta_i$ before the adjustment. Since every vertex in $\hat{V}_{out}$ is either free or incident to an eligible matched edge, we have $u,v \notin \hat{V}_{out}$. Therefore, $yz(e)$ is increased by at most $\delta_i$, and it does not decrease during the adjustment. This implies $yz(e) \leq w_i(e) + 3(\delta_i - \delta_j)$ after the adjustment, satisfying \cref{prop:RCS}(\ref{RCS:3},\ref{RCS:4}) (near domination and near tightness). If $e$ is eligible then it must be the case that one endpoint is in $\hat{V}_{in}$ and the other is in $\hat{V}_{out}$. $yz(e)$ value would remain unchanged and so \cref{prop:RCS}(\ref{RCS:3},\ref{RCS:4}) (near domination and near tightness) are satisfied.\qedhere
%\item $e \notin M$ and only $v \notin \hat{V}_{in} \cup \hat{V}_{out}$.  
\end{enumerate}
\end{proof}

\newcommand{\exchange}{\operatorname{\texttt{Exchange}}}

\subsection{Implementations and Running Times}\label{sec:imple}
In this section, we describe how to implement the algorithm in the \congest model with the proposed running time.

\paragraph{Storing the Variables and the Blossom Structures with $O(\log n)$ Bits per Vertex}
We begin by describing how we store the variables involved in the algorithm. For variables $M$, $y(u)$, $\Delta w$, and the type of the edge, they can be kept by the vertex itself or the incident vertices. Given a set of blossoms $\Omega$ and a set of vertices $X$, we define $\Omega[X]$ to be $\{B \in \Omega \mid B \subseteq X \}$, the set of blossoms that are fully contained in $X$. We will need the following property: Each node is only storing $O(\log n)$-bits so that given any $X\subseteq V$, $\Omega[X]$ and the $z$-values of the blossoms in $\Omega[X]$ can be reconstructed by using the information stored in every vertex of $X$. 

For each blossom $B$, we name the blossom by an edge $e_{B}=(u,v)$ where $u,v \in B$ and $u$ and $v$ were outer at the time right before $B$ is contracted. Note that no two blossoms can have the same name. Now for every edge $e$, we record whether it is a blossom edge. If $e$ is a blossom edge, we label $e$ by $e_B$ where $B$ is the immediate blossom it belongs to. Also, to preserve the structure of the blossom tree, for each $B$, we store the name of the parent blossom of $B$ and the $z$-value of $B$ at $e_B$.  Now given $X$, $\Omega[X]$ can be reconstructed by gathering the name of the parent of $e_B$ for every $B \in \Omega[X]$ and the label of every edge in $G[X]$. Likewise, the $z$-values of the blossoms in $\Omega[X]$ can be recovered by gathering the $z$-value stored at $e_B$ for each $B$.

%For each blossom tree, we can contract it into a root blossom of $X$ by traversing the trees in the postorder: Whenever leaving a blossom $B$, we collect the edges labeled $e_B$ and contract them into $B$. 

Each edge is only keeping $O(\log n)$-bits information. Since the degeneracy of the graph is $O(1)$, we can assign the edges to their incident edges in a way where every vertex is responsible for at most $O(1)$ edges. Therefore, each vertex is only keeping $O(\log n)$-bits information.

%To store the information about the blossoms structure in $\Omega$, each node $v$ stores the ID of the base node of the immediate blossom it is in as $base(v)$ and that of the root blossom it is in as $root\_base(v)$. Initially, $base(v) = root\_base(v) = v$. 

\paragraph{Routing Messages in Root Blossoms} Given a set of vertices $X$, let $\exchange(X)$ be the procedure that a vertex in $X$ collects a message of size of $O(\log n)$ from every vertex in $X$ and then the vertex sends a message of size $O(\log n)$ back to every vertex in $X$. Let $\mathcal{C}$ be a collection of sets, $\exchange(\mathcal{C})$ denotes the procedure where $\exchange(X)$ is executed for every $X \in \mathcal{C}$. Given an expander decomposition $(\hat{V}_1, \ldots, \hat{V}_k)$ with parameter $\epsilon$, by \Cref{thm:routing-main}, $\exchange(\{\hat{V}_i \}_{i=1}^{k})$ can be implemented in $\poly(\log n, 1/\epsilon)$ rounds.

Now we claim that at any point in the algorithm, if $\Omega_{root} \subseteq \Omega$ is the current set of root blossoms, $\exchange(\Omega_{root})$ can be implemented in $\poly(\log n, 1/\epsilon)$ rounds.  Consider each root blossom $B \in \Omega_{root}$, when $B$ is formed it must belong entirely to some component in the expander decomposition of the iteration.

$\exchange(\Omega_{root})$ is done by iterating through all the expander decompositions of the blossom shrinking step for each scale $i$ and each iteration $j$.  Let $\Omega_{ij} \subseteq \Omega_{root}$ denote the set of root blossoms formed during scale $i$ and iteration $j$. Let $\mathcal{V}_{ij}$ denote the exapnder decomposition of scale $i$ and iteration $j$.  For each scale $i$ and iteration $j$, we can perform $\exchange(\mathcal{V}_{ij})$ to simulate $\exchange(\{\Omega_{ij}[\hat{V}_t] \}_{t=1}^{k_{ij}})$, where $k_{ij} = |\mathcal{V}_{ij}|$. The total rounds needed is:
$$O( \#(scales) \times (\#iterations) \times T_{exp\_routing}) = poly(\log n, 1/\epsilon) \mbox{ rounds.}$$

%Each vertex stores which $V_i$ of the expander decomposition it belonged to in each iteration of each scale. 

%  Each component $V_t$ in the decomposition will first aggregate the graph topology as well as the blossom structure at $v^{*}_t$. Then it will aggregate the messages generated by the blossoms in $\Omega_{ij}[V_t]$. Finally, it sends the corresponding messages to the corresponding nodes in each blossom  in $\Omega_{ij} [V_t]$. 

%

   %Each node also records the adjacent blossom edges. To route the blossom structure information to node $v^{*}_i$, each node $v$ sends $base(v)$ as well as the adjacent blossom edges to $v^{*}_i$. To achieve this, each node only needs to send $O(\log n)$ bits of information.  

We start by discussing the implementation of the subroutine $\textsc{Cut}(\hat{V}_1, \ldots, \hat{V}_k)$. Suppose edges in $C$ are marked, we may mark $h(C)$ as follows. Perform $\exchange(\Omega_{root})$. For each $B \in \Omega_{root}$, for every internal edge $e \in C \cap B$, unmark $e$ and then (i) mark the matched incident edge to $B$ if $B$ is a regular blossom, or (ii) mark the free vertex in $B$ if $B$ is a free blossom. The marked edges and vertices form $h(C)$. Therefore, $\textsc{Cut}(\cdot)$ can be implemented in $\poly(\log n, 1/\epsilon)$ rounds. Next, we show how each step of the algorithm can be implemented:

\begin{enumerate}[leftmargin=*]

\item {\bf Augmentation:}  The maximal matching in step $(\ref{step_exp:1})$ can be done in $O(\log n)$ by Luby's algorithm.  The expander decomposition in Step $(\ref{step_exp:3})$ can be computed in $\poly(\log n, 1/\epsilon)$ rounds. At iteration $t$ of the augmentation, we need to find a maximal set of augmenting path in $(G_{elig}/ \Omega)[U^{(t)}_i]$ for each $1 \leq i \leq k$. Note that the only vertices in $(G^{(t)}_{elig}/ \Omega)[U^{(t)}_i]$ that is not in $(G_{elig}/ \Omega)[U_i]$ are vertices in $\hat{F}_{\leq C_{H}}$. These free vertices must have exactly one neighbor in $(G^{(t)} /\Omega)[U^{(t)}_i]$. We let each such vertex sends a token containing its ID to its neighbor.  It suffices for every vertex who received tokens to keep an arbitrary token, because a vertex can be in at most one augmenting path and the augmenting path can use whichever neighboring free vertex as the endpoint (as those free vertices have degree 1).

Then we perform $\exchange(\{\hat{U}_i\}_{i=1}^{k})$ to aggregate the graph topology along with the dangling free vertices, dual variables, and $\Omega[\hat{U}_i]$ at a node of $v_i \in \hat{U}_i$. The node then performs the augmentation step on $(G^{(t)}_{elig}/\Omega)[U_i]$ locally. Then, it broadcast the updated information back to every vertex in $\hat{U}_i$.

\item {\bf Blossom Shrinking:} %In the blossom shrinking step, we need to find $\Omega_{i}$, a maximal set of (nested) blossom in $(G_{elig}/ \Omega) [V_i \cup \hat{F}_{\leq C_{H}}]$, for each $1 \leq i \leq k$. %By \cref{lem:no_bad_free_vertex}, no vertices in $\hat{F}_{bad}$ can be included in any new blossoms. 
%As a free vertex in $\hat{F}_{\leq C_{H}}$ may become the base of some blossoms, we need to be careful about the implementation. 
%First, each free vertex in $\hat{F}_{\leq C_{H}}$ sends a token containing its ID to all its neighbors. If a vertex receive more than one tokens, it is fine for the vertex to forget the IDs of the tokens.  This is because a vertex that receives more than one tokens cannot be included in any new blossoms; otherwise, an augmenting path would be formed, which contradicts with \cref{lem:no_augmenting_path}. Therefore, each vertex only needs to store whether it is adjacent to a free vertex and at most one ID of an free vertex. Then we can perform $\exchange(\{\hat{V}_i\}_{i=1}^{k})$ to aggregate the necessary information in each $V_i$ to compute $\Omega_i$ for each $1\leq i \leq k$. When broadcasting $\Omega_{i}$ back to 
In the blossom shrinking step, we need to find a maximal set of (nested) blossom in $(G_{elig}/ \Omega) [V_i \cup \hat{F}']$ for each $1 \leq i \leq k$. By \cref{lem:no_bad_free_vertex}, no vertices in $\hat{F}'$ can be included in any new blossoms. Therefore, for each $1 \leq i \leq k$, it suffices for every vertex in $\hat{V}_i$ to record whether it has a eligible neighbor in $\hat{F}'$.  

Note that although nodes in $\hat{F}'$ cannot be in any blossoms,  they may cause some blossoms to be formed. 
Thus, we first let every vertex in $\hat{F}'$ inform its neighbor about its existence. Then we perform $\exchange(\{\hat{V}_i\}_{i=1}^{k})$ to aggregate the graph topology, whether each vertex has an eligible neighbor in $\hat{F}'$, dual variables, and $\Omega[\hat{V}_i]$ to some node $v_i$ in $\hat{V}_i$. The node then performs the blossom shrinking step on $(G_{elig}/\Omega)[V_i \cup \hat{F}']$ locally. Then, it broadcast the updated information back to every vertex in $\hat{V}_i$.

\item {\bf Dual Adjustment:} For the adjustment on the $y$-values, once every vertex determines whether it is in $\hat{V}_{in}$ and $\hat{V}_{out}$, it will be able to adjust its $y$-value properly.  The following is a procedure to determine whether each vertex is in $\hat{V}_{in}$ and $\hat{V}_{out}$.

\begin{enumerate}
\item Every vertex in $\hat{F}'$ mark itself in $\hat{V}_{out}$. Then it sends a message through each incident eligible unmatched edge to inform the neighbor to mark themselves in $\hat{V}_{in}$

\item Perform $\exchange(\{\hat{V}_i\}_{i=1}^{k})$ and have each $V_i$ determine locally each node in $V_i$ whether it is in $\hat{V}_{in}$ and $\hat{V}_{out}$ and then send information back to each node in $V_i$. For every vertex in $\hat{V}_{out}$, sends a message through each incident eligible unmatched edge to inform the neighbor to mark itself in $\hat{V}_{in}$.

\item Perform $\exchange(\Omega_{root})$ to do the following: If a vertex $v$ is contained in a root blossom $B$ and $v \in \hat{V}_{in}$, label every $v' \in B$ to be in $\hat{V}_{in}$. (Note that we only need to do this for vertices in $\hat{V}_{in}$ because if the base of a root blossom is in $\hat{V}_{out}$ then all vertices in the blossom will end up in the same component as the base).
\end{enumerate}

Now every vertex has indentified whether itself is in $\hat{V}_{in}$ and $\hat{V}_{out}$. The $y$-values of the vertcies can now be adjusted accordingly.

For the $z$-values, note that only the root blossoms need to have their $z$-values adjusted. Therefore, we perform $\exchange(\Omega_{root})$ to check for each blossom $B \in \Omega_{root}$, whether all of its members are in $\hat{V}_{in}$ or $\hat{V}_{out}$. If it is the former, then we set $z(B) \leftarrow z(B) - \delta_i$. If it is the latter, then we set $z(B) \leftarrow z(B) + \delta_i$. Recall that $z_B$ is stored on the edge $e_{B}$.

\item {\bf Blossom Dissolution and Clean Up:}
We may perform a $\exchange(\Omega_{root})$ to dissolve all the blossoms with zero $z$-values. The removal of the dummy matched edges can be done locally at the free vertices.

%Let $v^{*}_i$ be the vertex with the maximum degree in $V_i$. For each $i$, we route the entire graph topology of $G_{elig}[V_i]$ to $v^{*}_i$ in $poly(\log n, 1/\epsilon)$ rounds. We also route the information of the blossom structure $\Omega$ to $v^{*}_i$. 
\end{enumerate}

Each of the step takes $\poly(\log n, 1/\epsilon)$ rounds. Since there are at most $O( (\log W) / \epsilon )$ iterations, the total running time is $\poly(\log n, 1/\epsilon)$ rounds as $W = \poly(n)$. 
%We summarize this section with the following theorem:
We conclude the proof of the following theorem.

\thmmatching*

% \paragraph{Routing the Blossom Structure.} Each node $v$ stores the ID of the base node of the immediate blossom it is in as $base(v)$.  Each node also records the adjacent blossom edges. To route the blossom structure information to node $v^{*}_i$, each node $v$ sends $base(v)$ as well as the adjacent blossom edges to $v^{*}_i$. To achieve this, each node only needs to send $O(\log n)$ bits of information.  

%Once $v^{*}_i$ received the information, it can reconstruct the blossom structure $\Omega$ and $G'_{elig}[V_i] / \Omega$. Then each $v^{*}_i$ may perform augmentation on $G'_{elig}[V_i] /\Omega$ locally. After the augmentation step, $v^{*}_i$ sends the status of the edges back to every node in the graph. Note that the status of each edge is the only quantity that changes during this step.

%\clearpage
\section{Edge Separators}\label{sect:edge-separator}

In this section, we show that any $H$-minor-free graph $G$ admits an edge separator of size $O(\sqrt{\Delta n})$, which is needed in the proof of \cref{lem:separator}.

\thmedgeseparator*

For the rest of this section, $G=(V,E)$ is  $H$-minor-free. For any $R \subseteq V$, a connected component of $G[V \setminus R]$ is called an \emph{$R$-flap}. The following lemma reduces \cref{thm:edge-separator} to finding a subset $R \subseteq V$ with $\vol(R) = O(\sqrt{\Delta n})$ such that each $R$-flap has at most $(2/3)n$ vertices.

\begin{lemma}\label{lem:reduction1}
If there exists $R \subseteq V$ with $\vol(R) = O(\sqrt{\Delta n})$ such that each $R$-flap has at most $(2/3)n$ vertices, then there is an edge separator of size $O(\sqrt{\Delta n})$.
\end{lemma}
\begin{proof}
Let $V_1, V_2, \ldots, V_k$ be the connected components after removing all edges incident to $R$. Each $V_i$ is either an $R$-flap or a single vertex in $R$. We rank these sets according to their size $|V_1| \geq |V_2| \geq \cdots \geq |V_k|$. We pick $S = \bigcup_{1 \leq i \leq i^\ast} V_i$, where $i^\ast$ is the smallest index such that $\left| \bigcup_{1 \leq i \leq i^\ast} V_i \right| \geq n/3$. It is clear that $(1/3)n \leq |S| \leq (2/3)n$ and the number of edges crossing $S$ and $V \setminus S$ is $O(\sqrt{\Delta n})$.
\end{proof}

\subsection{Preliminaries}
We start with some graph terminology needed in the proof of \cref{thm:edge-separator}. 
%In particular, we state the \emph{graph structure theorem} of Robertson and Seymour~\cite{ROBERTSON2003nonplanar}, which establishes a useful structural property of any $H$-minor-free graph. 

\paragraph{Tree Decompositions} A \emph{tree decomposition} of a graph $G=(V,E)$ is a pair $(T,X)$, where $T$ is a tree on the vertex set $X =\{X_1, X_2, \ldots, X_k\}$, which is a family of subsets of $V$, satisfying the following conditions. 
\begin{itemize}
    \item The union of all $X_i$ equals $V$, that is, $X_1 \cup X_2 \cup \cdots \cup X_k = V$.
    \item For each edge $e=\{u,v\} \in E$, there is an element $X_i \in X$ such that $\{u,v\} \subseteq X_i$.
    \item For any $X_i \in X$, $X_j \in X$, and $X_l \in X$ such that $X_l$ is in the unique path connecting $X_i$ and $X_j$ in $T$, we have $X_i \cap X_j \subseteq X_l$.
\end{itemize}

The \emph{width} of a tree decomposition $(T,X)$ is $\max_{X_i \in X} |X_i| - 1$. The \emph{treewidth} $\tw(G)$ of  $G$ is the minimum width among all tree decompositions of $G$. 
%If $T$ is a path, then $(T,X)$ is a \emph{path decomposition}. Similarly, the \emph{pathwidth} of  $G$ is the minimum width among all path decompositions of $G$. 
For any edge $e = \{u,v\}$ in a tree $T$, we write $S_{u \rightarrow v}$ to denote the connected component of $T - e$ that contains $v$. \cref{lem:treewidth-aux} is a well-known property of bounded-treewidth graphs~\cite[Theorem 2.5]{ROBERTSON1986treewidth}.

%\begin{lemma}\label{lem:tree-aux}
%For any tree $T=(V,E)$ with an weight function $w(v) \geq 0$ for each $v \in V$, there is a vertex $u \in V$ such that $\sum_{r \in S_{u \rightarrow v}} w(r) \leq (1/2) \sum_{r \in V} w(r)$. 
%\end{lemma}
%\begin{proof}
%We assign $u \in V$ to an edge $e=\{u,v\}$ if   $\sum_{r \in S_{u \rightarrow v}} w(r) > (1/2) \sum_{r \in V} w(r)$. If some vertex $u$ is unassigned, then we are done. Otherwise, at least one edge $e=\{u,v\}$ is assigned to both $u$ and $v$, as $|E| = |V| - 1$. However, 
%\begin{align*}
%\sum_{r \in V} w(r) &= \sum_{r \in S_{u \rightarrow v}} w(r) + \sum_{r \in S_{v \rightarrow u}} w(r) \\
%&> (1/2) \sum_{r \in V} w(r) + (1/2) \sum_{r \in V} w(r) \\
%&= \sum_{r \in V} w(r),
%\end{align*}
%which is a contradiction.
%\end{proof}

\begin{lemma}\label{lem:treewidth-aux}
Let $G=(V,E)$ be a graph associated with an weight function $w(v) \geq 0$ for each $v \in V$.
For any tree decomposition $(T,X)$ of $G$, there exists $W  \in X$  such that each $W$-flap $S$ satisfies $\sum_{v \in S} w(v) \leq (1/2) \sum_{v \in V} w(v)$.
\end{lemma}
\begin{proof}
%Take a tree decomposition $(T,X)$ of width $\tw(G)$.
For any edge $\{A,B\}$ in $T$, we write $Z_{A \rightarrow B} = \bigcup_{C \in S_{A \rightarrow B}} C \setminus A$. Note that $Z_{A \rightarrow B}$ is a subset of $V$.
We assign an edge $e=\{A,B\}$ to  $A \in X$  if   $\sum_{v \in Z_{A \rightarrow B}} w(v) > (1/2) \sum_{v \in V} w(v)$. If some  $A$ is unassigned, then we are done by setting $W = A$, as any $A$-flap $S$ must be  $S = Z_{A \rightarrow B}$ for some neighboring  $B$.
Otherwise, at least one edge $e=\{A,B\}$ is assigned to both $A$ and $B$, as the number of edges in $T$ is the number of vertices in $T$ minus one. However,
\begin{align*}
\sum_{v \in V} w(v) &= \sum_{v \in A \cap B} w(v) + \sum_{v \in Z_{A \rightarrow B}} w(v) + \sum_{v \in Z_{B \rightarrow A}} w(v) \\
&> (1/2) \sum_{v \in V} w(v) + (1/2) \sum_{v \in V} w(v) \\
&= \sum_{v \in V} w(v),
\end{align*}
which is a contradiction.
\end{proof}

The following lemma is an immediate consequence of \cref{lem:treewidth-aux}, by taking a tree decomposition $(T,X)$ of width $\tw(G)$.

\begin{lemma}\label{lem:treewidth-aux2}
For any graph $G=(V,E)$ with an weight function $w(v) \geq 0$ for each $v \in V$, there exists a subset $W \subseteq V$ of size $|W| \leq \tw(G)+1$ such that each $W$-flap $S$ satisfies $\sum_{v \in S} w(v) \leq (1/2) \sum_{v \in V} w(v)$.
\end{lemma}

\paragraph{Clique Sums} We say that $G$ is a \emph{$k$-clique-sum} of two graphs $H_1$ and $H_2$ if $G$ can be constructed as follows.
\begin{itemize}
    \item Pick a number $1 \leq s \leq k$. 
    \item Pick an $s$-clique $C_1=\{u_1, u_2, \ldots u_s\}$ in $H_1$ and an  $s$-clique $C_2=\{v_1, v_2, \ldots v_s\}$ in $H_2$. 
    \item Combine the two graphs $H_1$ and $H_2$ by identifying $u_i = v_i$ for each $1 \leq i \leq s$. 
    \item Remove some edges in the clique $C_1 = C_2$.
\end{itemize}

We also write $G = H_1 \oplus_k H_2$. Note that $k$-clique-sum $\oplus_k$ is not a well-defined operator, as there are multiple options for selecting the cliques $C_1$ and $C_2$  and selecting the deleted edges.

More generally, we say that $G$ is a $k$-clique-sum of a list of graphs $H_1, H_2, \ldots H_r$ if there exists a sequence of graphs $G_1, G_2, \ldots G_r$ with $G = G_r$ such that $G_1 = H_1$ and $G_i = G_{i-1} \oplus_k H_{i}$ for each $2 \leq i \leq r$. This construction of $G$  gives rise to the following tree decomposition $(T,X)$.
\begin{itemize}
    \item The family of subsets is $X = \{V_1, V_2, \ldots V_r\}$, where $V_i$ is the vertex set for $H_i=(V_i, E_i)$. Here $V_i$ is also viewed as a subset of $V$.
    \item For each $2 \leq i \leq r$, there is an index $1 \leq j \leq i-1$ such that the clique of $G_i$ used in the $k$-clique-sum operation $G_i = G_{i-1} \oplus_k H_{i}$ belongs to $H_j$. The edge $\{V_i, V_j\}$ is added to the tree $T$. 
\end{itemize}

This tree decomposition $(T,X)$ described above is not unique in general, as there might be multiple possible choices of $j$ for each $i$ in the construction of the edge set of $T$. Applying  \cref{lem:treewidth-aux} to this tree decomposition $(T,X)$, we obtain the following lemma.

%\begin{lemma}
%Let $G$ be a $k$-clique-sum of $H_1, H_2, \ldots H_r$.  Consider any $H_i=(V_i, E_i)$. Let $S$ be any $V_i$-flap in $G$.
%Let $Z$ be the set of vertices in $V_i$ adjacent to a $S$ in $G$. Then $Z$ induces a clique of size at most $k$ in $H_i$. 
%\end{lemma}

\begin{lemma}\label{lem:cliquesum-aux}
Let $G=(V,E)$ be  a $k$-clique-sum of a list of graphs $H_1, H_2, \ldots H_r$. Let $V_i$ be the vertex set of $H_i=(V_i, E_i)$. There exists  an index $1 \leq i \leq r$ such that each $V_i$-flap $S$ has size $|S| \leq |V|/2$.
\end{lemma}
\begin{proof}
Set $w(v) = 1$ for each $v \in V$ and apply  \cref{lem:treewidth-aux} to the tree decomposition $(T,X)$ associated with the construction of $G$ from $H_1, H_2, \ldots H_r$ via $k$-clique-sums.
%The proof is analogous to the proof of \cref{lem:treewidth-aux}. Note that here $V_i$ is interpreted as a subset of $V$. Take the tree $T$ associated with the construction of $G$ from $H_1, H_2, \ldots H_r$ via $k$-clique-sums.
%For any edge $\{H_i,H_j\}$ in $T$, we write $Z_{H_i \rightarrow H_j} = \bigcup_{V_k \in S_{H_i \rightarrow B}} V_k \setminus V_i$.
%We assign $H_i$ to an edge $e=\{H_i,H_j\}$ if   $|Z_{H_i \rightarrow H_j}| > (1/2) |V|$. If some  $H_i$ is unassigned, then we are done by setting $W = V_i$, as any $V_i$-flap $S$ must be  $S = Z_{H_i \rightarrow H_j}$ for some neighbor $H_j$.
%Otherwise, at least one edge $e=\{H_i, H_j\}$ is assigned to both $H_i$ and $H_j$, as the number of edges in $T$ is the number of vertices in $T$ minus one. However,
%\[
%|V| = |V_i \cap V_j| + |Z_{H_i \rightarrow H_j}| + |Z_{H_j \rightarrow H_i}| > (1/2)|V| + (1/2)|V| = |V|,\]
%which is a contradiction.
\end{proof}

\paragraph{Almost Embedding} For a given closed surface $\Sigma$ and a number $h$, we say that $G$ is \emph{$h$-almost embeddable} on $\Sigma$ if $G$ can be constructed as follows.
\begin{itemize}
    \item Start with a graph $G_0$ that can be embedded on a surface $\Sigma$ and its embedding.
    \item Add at most $h$ \emph{vortices} of depth at most $h$ to the graph. The precise definition of vortices is omitted as it is not needed in our proof. See~\cite{kawarabayashi2007survey,ROBERTSON2003nonplanar} for details.
    \item Add at most $h$ \emph{apex} vertices to the graph, and add an arbitrary number of edges with  at least one endpoint being an apex. 
\end{itemize}

We write $A$ to denote the set of apex vertices. We have $|A| \leq h$. If $A = \emptyset$, then $G$ is said to be  \emph{apex-free}  $h$-almost embeddable on $\Sigma$. In particular, $G[V \setminus A]$ is apex-free $h$-almost embeddable on $\Sigma$.  The only thing we need to know about a graph $G=(V,E)$ that is $h$-almost embeddable on $\Sigma$ is that there exists a subset $A \subseteq V$ of size at most $h$ such that the subgraph  $G[V \setminus A]$ satisfies the following property.
\begin{itemize}
    \item For any minor $G' \preceq G[V \setminus A]$, its treewidth $\tw(G') = O(D)$ is linear in its diameter $D = \diam(G')$, as guaranteed by the following lemma by Grohe~\cite[Proposition 11]{grohe2003local}.
\end{itemize}

\begin{lemma}[{{\cite[Proposition 11]{grohe2003local}}}]\label{lem:treewidth-diameter-rel}
Let $G$ be a minor of a graph that is apex-free $h$-almost embeddable on $\Sigma$. Then its treewidth $\tw(G) = O(D)$ is linear in its diameter $D = \diam(G)$, where the hidden constant in $O(\cdot)$ depends only on $h$ and $\Sigma$.
\end{lemma}

Now we are ready to state the graph structure theorem of Robertson and Seymour~\cite{ROBERTSON2003nonplanar}.

\begin{theorem}[{{\cite[Theorem~1.3]{ROBERTSON2003nonplanar}}}]\label{thm:structure}
For any graph $H$, there is a number $h$ such that
any  $H$-minor-free graph $G$ can be constructed by taking $h$-clique-sum of a list of graphs that are $h$-almost embeddable on some surfaces on which $H$ cannot be embedded.
\end{theorem}

We note that the list of surfaces $\Sigma_1, \Sigma_2, \ldots, \Sigma_k$ on which $H$ cannot be embedded is \emph{finite} for any graph $H$. If $H$ is planar, then the list is empty. In that case, $G$ is the result of taking  $h$-clique-sum of a list of graphs of at most $h$ vertices. That is, each graph in the list only includes the set of apex vertices $A$. 

\subsection{The Existence of a Small Edge Separator}

For the rest of the section, we prove \cref{thm:edge-separator}. 
Let $G=(V,E)$ be the $H$-minor-free graph under consideration. Then  \cref{thm:structure} implies that $G$ can be described as an $h$-clique-sum of a list of graphs $H_1, H_2, \ldots H_r$. Moreover, for each $H_i=(V_i, E_i)$, there is a subset $A_i \subseteq V_i$ of size $|A_i| \leq h$ such that the subgraph of $H_i$ induced by $V_i \setminus A_i$ is apex-free $h$-almost embeddable on some surface $\Sigma$ on which $H$ cannot be embedded.

By \cref{lem:cliquesum-aux}, there is an index $i^\ast$ such that each $V_{i^\ast}$-flap has at most $n/2$ vertices. In view of \cref{lem:reduction1}, to prove \cref{thm:edge-separator} is suffices to find a set $R$ with $A_{i^\ast} \subseteq R \subseteq V_{i^\ast}$ with $\vol(R) = O(\sqrt{\Delta n})$ such that each $R$-flap has at most $(2/3)n$ vertices.

\paragraph{High-level Ideas} Before presenting the proof, we briefly discuss its high-level ideas. Let us first focus on the special case where the graph $G$ consists of only the vertices $V_{i^\ast}$. In this case, \cref{lem:treewidth-diameter-rel}  already allows us to find a desired set $R$ via the breadth-first layering approach~\cite{Lipton79separator}, which we briefly explain as follows.
The apex vertices $A_{i^\ast}$ can be ignored since the number of edges incident to  $A_{i^\ast}$ is at most $O(\Delta)= O(\sqrt{\Delta n})$, as  $|A_{i^\ast}| \leq h = O(1)$.
Given any breadth-first search tree on the subgraph induced by $V_{i^\ast} \setminus A_{i^\ast}$, there are two cases. If a desired set $R$ can be 
obtained by setting $R$ as one of the layers of the  breadth-first layering, then we are done. Otherwise, most of the vertices in the graph are confined to $O(\sqrt{n / \Delta})$ consecutive layers $i, \ldots, j$. To deal with the  vertices in layers $i, \ldots, j$, we consider the graph resulting from  contracting all vertices in layers $0, \ldots, i-1$ and removing all vertices in layers $j+1, \ldots$. The resulting graph has diameter  $O(\sqrt{n / \Delta})$. By \cref{lem:treewidth-diameter-rel}, this graph has treewidth $O(\sqrt{n / \Delta})$, so a desired set $R$ can be computed using \cref{lem:treewidth-aux2} with $w(v) = 1$, as we note that $|S| = O(\sqrt{n / \Delta})$ implies $\vol(S) = O(\sqrt{n \Delta})$ for any vertex set $S$.

To extend this approach to the general case, our key idea is to consider a weight function $f(v)$ for all $v \in V_{i^\ast} \setminus A_{i^\ast}$ that aims to take care of the vertices outside of $V_{i^\ast}$. Indeed, \cref{lem:treewidth-aux2} applies to the weighted case.  For each $V_{i^\ast}$-flap $S$, let $C_S \subseteq V_{i^\ast} \setminus A_{i^\ast}$ be the set of vertices in $V_{i^\ast} \setminus A_{i^\ast}$ adjacent to some vertex in $S$. Then we will distribute $|S|$ amount of weight to the weight functions $f(v)$ for $v \in C_S$. 
As we will later see, by properly designing a weight function $f(v)$ the approach sketched above can be extended to the general case.

\paragraph{Weight Assignment}  We  consider a weight assignment function $f(v)$ for all $v \in V_{i^\ast} \setminus A_{i^\ast}$ satisfying the following properties.
\begin{itemize}
    \item For each $v \in V_{i^\ast} \setminus A_{i^\ast}$, $f(v) = 1 + \sum_{u \in V \setminus  V_{i^\ast}}{f_{u \rightarrow v}}$.
    \item For each $V_{i^\ast}$-flap $S$, let $C_S \subseteq V_{i^\ast} \setminus A_{i^\ast}$ be the set of vertices in $V_{i^\ast} \setminus A_{i^\ast}$ adjacent to some vertex in $S$. 
    If $C_S = \emptyset$, then $f_{u \rightarrow v} = 0$ for all $u \in S$ and $v \in V_{i^\ast} \setminus A_{i^\ast}$.
    Otherwise, for each $u \in S$, the function $f_{u \rightarrow v}$ satisfies the following requirements.
    \begin{itemize}
        \item $\sum_{v \in V_{i^\ast} \setminus A_{i^\ast}} f_{u \rightarrow v} = 1$.
        \item $f_{u \rightarrow v} = 0$ if $v \notin C_S$.
        \item $f_{u \rightarrow v} \geq 0$ if $v \in C_S$.
    \end{itemize}
\end{itemize}
It is clear that $0 \leq |C_S| \leq h$ for each $V_{i^\ast}$-flap $S$, and we have
\begin{align*}
0 \leq \sum_{v\in V_{i^\ast} \setminus A_{i^\ast}} f(v)
&= \sum_{v\in V_{i^\ast} \setminus A_{i^\ast}} \left( 1 + \sum_{u \in V \setminus  V_{i^\ast}}{f_{u \rightarrow v}} \right)\\
&= |V_{i^\ast} \setminus A_{i^\ast}| +   \sum_{u \in V \setminus  V_{i^\ast}}
\sum_{v\in V_{i^\ast} \setminus A_{i^\ast}}
f_{u \rightarrow v}\\
&\leq |V_{i^\ast} \setminus A_{i^\ast}| +   |V \setminus V_{i^\ast}|\\
&= |V \setminus A_{i^\ast}| \\
&\leq n.
\end{align*}
 We select the function $f$ to minimize $\sum_{v\in V_{i^\ast} \setminus A_{i^\ast}} f^2(v)$ among all functions satisfying the requirements.
%\begin{itemize}
%    \item We restrict our consideration to those functions $f$ minimizing $\max_{u\in V_{i^\ast} \setminus A_{i^\ast}} f(u)$.
%    \item Then we choose a function $f$ minimizing $\left|\left\{v \in V_{i^\ast} \setminus A_{i^\ast} \ | \ f(v) = \max_{u\in V_{i^\ast} \setminus A_{i^\ast}} f(u)\right\}\right|$.
%\end{itemize}
Such a function $f$ has the following property.

\begin{lemma}\label{lem:function-aux}
Consider any $u \in V\setminus V_{i^\ast}$ and $v \in V_{i^\ast}$ with  $f_{u \rightarrow v} > 0$.
Let $S$ be the $V_{i^\ast}$-flap containing $u$. For each $v' \in C_S$, either one of the following holds.
\begin{itemize}
    \item $f(v') > f(v)$ and $f_{u \rightarrow v'} = 0$.
    \item $f(v') = f(v)$ and $f_{u \rightarrow v'} \geq 0$.
\end{itemize}
%we have $f(v') = f(v)$ for all $v' \in C_S$, where $S$ is the $V_{i^\ast}$-flap containing $u$.
\end{lemma}
\begin{proof}
It suffices to find a weight assignment function $\tilde{f}$ such that 
\[ \sum_{v\in V_{i^\ast} \setminus A_{i^\ast}} \tilde{f}^2(v) <  \sum_{v\in V_{i^\ast} \setminus A_{i^\ast}} f^2(v)\]
for the following two cases.
\begin{itemize}
    \item $f(v') > f(v)$ and $f_{u \rightarrow v'} > 0$.
    \item $f(v') < f(v)$.
\end{itemize}

We first consider the case $f(v') < f(v)$. %as the other case is similar by switching the roles of $v$ and $v'$.
Let $\delta = \min\left\{f_{u \rightarrow v},  \frac{f(v) - f(v')}{2}\right\} > 0$. Consider the new function $\tilde{f}$ resulting from decreasing $f_{u \rightarrow v}$ by $\delta$ and increasing $f_{u \rightarrow v'}$ by $\delta$. 
By the convexity of the square function, we have $\tilde{f}^2(v)+\tilde{f}^2(v') < f^2(v) + f^2(v')$,  so  $\sum_{v\in V_{i^\ast} \setminus A_{i^\ast}} \tilde{f}^2(v) <  \sum_{v\in V_{i^\ast} \setminus A_{i^\ast}} f^2(v)$.

The case $f(v') > f(v)$ and $f_{u \rightarrow v'} > 0$ can be handled by switching the role of $v$ and $v'$ in the above argument. That is, we consider  $\delta = \min\left\{f_{u \rightarrow v'},  \frac{f(v') - f(v)}{2}\right\} > 0$, and  the new function $\tilde{f}$ is the result of decreasing $f_{u \rightarrow v'}$ by $\delta$ and increasing $f_{u \rightarrow v}$ by $\delta$. 
\end{proof}

\paragraph{High-weight Vertices} We consider a threshold $\tau = \eps \sqrt{\Delta n}$, where $\eps$ is some small positive constant number to be determined. Consider the set of high-weight vertices 
\[R_1 = \left\{v \in V_{i^\ast} \setminus A_{i^\ast} \ \middle| \ f(v) \geq \eps \sqrt{\Delta n}\right\},\]
%which will be included in $R$. 
Observe that $|R_1| \leq \sum_{v\in V_{i^\ast} \setminus A_{i^\ast}} f(v) / \tau = O\left(\sqrt{n / \Delta}\right)$, so we have\[\vol(R_1) \leq \Delta |R| = O(\sqrt{\Delta n}).\]

We make the following crucial observation.  
\begin{lemma}\label{lem:heavy}
For each $V_{i^\ast}$-flap $S$ with $C_S \neq \emptyset$, exactly one of the following holds.
\begin{itemize}
    \item We have $C_S \subseteq R_1$. In this case, $S$ is also a $(A_{i^\ast} \cup R_1)$-flap.
    \item We have $f_{u \rightarrow v} = 0$ for all  $u \in S$ and $v \in R_1$.
\end{itemize}
\end{lemma}
\begin{proof}
Suppose both $R_1 \cap C_S$ and $C_S \setminus R_1$ are non-empty. We show that $f_{u \rightarrow v} = 0$ for all  $u \in S$ and $v \in R_1$. Note that  we already have $f_{u \rightarrow v} = 0$ for each $v \notin C_S$ by the definition of $f$.
Pick any  $v \in R_1 \cap C_S$ and $v' \in C_S \setminus R_1$. Then $f(v) \geq \tau > f(v')$.

Suppose  $f_{u \rightarrow v} > 0$ for some $u \in S$. Then 
\cref{lem:function-aux} implies that $f(v') \geq f(v)$, which is a contradiction. Hence
 $f_{u \rightarrow v} = 0$ for each $u \in S$.
\end{proof}

\paragraph{Breadth-first Search}  We perform a breadth-first search  in each connected component of  $H_{i^\ast}[V_{i^\ast} \setminus A_{i^\ast}]$.
 Note that $G[V_{i^\ast} \setminus A_{i^\ast}]$ is a subgraph of $H_{i^\ast}[V_{i^\ast} \setminus A_{i^\ast}]$ because the clique-sum operation might remove edges.
 \begin{itemize}
     \item Let $L_0$ be any set of vertices in $V_{i^\ast} \setminus A_{i^\ast}$ that includes exactly one vertex for each connected component of $H_{i^\ast}[V_{i^\ast} \setminus A_{i^\ast}]$. Intuitively, $L_0$ is the set of roots of breadth-first search.
     \item For $i \geq 1$, let $L_i$ be the set of vertices  $u$ in $V_{i^\ast} \setminus A_{i^\ast}$ with $\dist(u,L_0)=i$ in $H_{i^\ast}[V_{i^\ast} \setminus A_{i^\ast}]$. Intuitively, $L_i$ is the set of layer-$i$ vertices of breadth-first search.
     \item Let $d = \max_{v \in V_{i^\ast} \setminus A_{i^\ast}} \dist(v, L_0)$ be the largest index such that $L_d \neq \emptyset$. 
 \end{itemize}
 Note that $L_i = \emptyset$ for each $i \notin \{0, 1, \ldots, d\}$.

\begin{lemma}\label{lem:flap-intersection}
For each $V_{i^\ast}$-flap $S$, there is an index $i$ such that the vertices in $V_{i^\ast} \setminus A_{i^\ast}$ adjacent to $S$ are confined to $L_i \cup L_{i+1}$.
\end{lemma}
\begin{proof}
Recall that $C_S \subseteq V_{i^\ast} \setminus A_{i^\ast}$ is the set of vertices in $V_{i^\ast} \setminus A_{i^\ast}$ adjacent to some vertex in $S$.
The vertices $C_S$ form a clique of size at most $h$ in $H_{i^\ast}=(V_{i^\ast}, E_{i^\ast})$, since $G$ is an $h$-clique-sum of $H_1, H_2, \ldots H_r$. 
If $C_S = \emptyset$, then we can pick $i$ to be any index. Otherwise, we pick $u \in C_S$ to minimize $\dist(u,L_0)$ in $H_{i^\ast}[V_{i^\ast} \setminus A_{i^\ast}]$. Let $i = \dist(u,L_0)$. Because $C_S$ is a clique in $H_{i^\ast}[V_{i^\ast} \setminus A_{i^\ast}]$, we have $C_S \subseteq L_i \cup L_{i+1}$.
\end{proof}

For simplicity, we write $f(S) = \sum_{u \in S}f(u)$ for each subset $S \subseteq V_{i^\ast} \setminus A_{i^\ast}$. Consider the following definitions.
\begin{itemize}
    \item $i_a$ is the smallest index $i$ with $f(L_0 \cup L_1 \cup \cdots \cup L_i) \geq f(V_{i^\ast} \setminus A_{i^\ast})/3$.
     \item $i_b$ is the largest index $i$ with $f(L_i \cup L_{i+1} \cup \cdots \cup L_d) \geq f(V_{i^\ast} \setminus A_{i^\ast})/3$.
     %\item $I$ is the set of integers $i$ with $ i_a < i < i_b$.
\end{itemize}
 It is clear that $0 \leq i_a \leq i_b \leq d$. 
 
\paragraph{Heavy and Light Layers} We say that $L_i$ is \emph{heavy} if \[\vol(L_i) \geq \sqrt{\Delta n} \ \  \text{ or } \ \ \sum_{v \in L_i} f(v) \geq (1/18)n,\] otherwise  $L_i$ is \emph{light}, where the volume is measure with respect to the original graph $G$. 

\begin{lemma}\label{lem:easy-case-aux}
Let $0 \leq i \leq d$ be an index. Let $S$ be an $(A_{i^\ast} \cup L_{i})$-flap. Then the following holds.
\begin{itemize}
    \item If $S \cap L_j \neq \emptyset$ for some $0 \leq j < i$, then $|S| \leq  f(L_{1} \cup L_{2} \cup \cdots \cup L_{i-1})$.
    \item If $S \cap L_j \neq \emptyset$ for some $i < j \leq d$, then $|S| \leq  f(L_{i+1} \cup L_{i+2} \cup \cdots \cup L_{d})$.    
\end{itemize}
\end{lemma}
\begin{proof}
We only consider the case $S \cap L_j \neq \emptyset$ for some $0 \leq j < i$, as the other case is similar.
 Let $U$ be the union of all vertices in the following sets. It is clear that $S \subseteq U$. 
\begin{itemize}
    \item $L_{1} \cup L_{2} \cup \cdots \cup L_{i-1}$.
    \item All $V_{i^\ast}$-flap $S'$ adjacent to $L_{1} \cup L_{2} \cup \cdots \cup L_{i-1}$.
\end{itemize}

We have 
\[|U| = |L_{1} \cup L_{2} \cup \cdots \cup L_{i-1}|
  + \sum_{\text{$S'$ is a $V_{i^\ast}$-flap adjacent to $L_{1} \cup L_{2} \cup \cdots \cup L_{i-1}$}} |S'|.\]
  
 By the definition of $f$, we have $f_{u \rightarrow v} > 0$ only if $v \in C_{S'}$, where $S'$ is the $V_{i^\ast}$-flap that contains $u$. Furthermore, if $S'$ is a $V_{i^\ast}$-flap adjacent to $L_{1} \cup L_{2} \cup \cdots \cup L_{i-1}$, then $C_{S'} \subseteq L_{1} \cup L_{2} \cup \cdots \cup L_{i}$ due to \cref{lem:flap-intersection}. Hence we have \[|S'| = \sum_{u \in S', \ v \in C_{S'}} f_{u \rightarrow v} = \sum_{u \in S', \ v \in L_{1} \cup L_{2} \cup \cdots \cup L_{i}} f_{u \rightarrow v}.\] 

Since $f(v) = 1 + \sum_{u \in V \setminus  V_{i^\ast}}{f_{u \rightarrow v}}$, we have 
$|S| \leq |U| \leq f(L_{1} \cup L_{2} \cup \cdots \cup L_{i})$, as we can write 
\begin{align*}
f(L_{1} \cup \cdots \cup L_{i}) &= |L_{1} \cup \cdots \cup L_{i}| + \sum_{u \in V \setminus V_{i^\ast}, \ v \in L_{1} \cup \cdots \cup L_{i}} f_{u \rightarrow v}\\
&\geq |L_{1} \cup \cdots \cup L_{i}| + \sum_{\text{$S'$ is a $V_{i^\ast}$-flap adjacent to $L_{1} \cup \cdots \cup L_{i-1}$}} \; \sum_{u \in S', \ v \in C_{S'}} f_{u \rightarrow v}\\
&= |L_{1} \cup  \cdots \cup L_{i}| + \sum_{\text{$S'$ is a $V_{i^\ast}$-flap adjacent to $L_{1} \cup  \cdots \cup L_{i-1}$}}  |S'|\\
&= |U|.\qedhere
\end{align*}
\end{proof}

\begin{lemma}\label{lem:easy-case}
Suppose that there exists an index $i_a < i' < i_b$ such that $L_{i'}$ is light. Then $R = A_{i^\ast} \cup L_{i'}$ satisfies that $\vol(R) = O(\sqrt{\Delta n})$ and each $R$-flap has at most $(2/3)n$ vertices.
\end{lemma}   
\begin{proof}
We already have $\vol(L_{i'}) = O(\sqrt{\Delta n})$ by the definition of a light layer. 
Since $|A_{i^\ast}| \leq h = O(1)$, we also have $\vol(A_{i^\ast}) = O(\Delta) = O(\sqrt{\Delta n})$. Therefore, $\vol(R) = O(\sqrt{\Delta n})$.

For the rest of the proof, we verify that each $R$-flap has at most $(2/3)n$ vertices. Let $S$ be any $R$-flap. If $S \cap V_{i^\ast} = \emptyset$, then we have $S \subseteq S'$ for some $V_{i^\ast}$-flap $S'$. However, our choice of $V_{i^\ast}$ guarantees that $|S'| \leq (1/2)n$, due to  \cref{lem:cliquesum-aux}, so  $|S'| \leq |S| \leq (1/2)n < (2/3)n$. 

Suppose that $S \cap V_{i^\ast} \neq \emptyset$. Let $u \in S \cap V_{i^\ast}$. Then $u \in L_i$ for some $i \neq i'$. We assume that $i < i'$, as the case of $i > i'$ is similar. We have 
\[|S| \leq f(L_{1} \cup L_{2} \cup \cdots \cup L_{i'}) \leq (2/3)f(V_{i^\ast} \setminus A_{i^\ast}) \leq (2/3)n,\]  where the first inequality is due to \cref{lem:easy-case-aux} and the second inequality is due to $i' < i_b$.  
\end{proof}

\paragraph{Diameter Reduction} In view of \cref{lem:easy-case}, from now on, we focus on the case where $L_{i'}$ is heavy for all $i_a < i' < i_b$.  We define two indices $i_a'$ and $i_b'$ as follows.
  
\begin{itemize}
    \item  The index $0 \leq i_a' \leq i_a$ is defined as follows.
    \begin{itemize}
        \item If $i_a = 0$ or $L_{i_a - 1}$ is light, then $i_a' = i_a$.
        \item Otherwise,  $i_a'$ is the smallest index $0 \leq i \leq i_a-1$ such that  $L_i, L_{i+1}, \ldots, L_{i_a - 1}$ are heavy.
    \end{itemize}
    \item  The index $i_b \leq i_b' \leq d$ is defined as follows.
    \begin{itemize}
        \item If $i_b = d$ or $L_{i_b + 1}$ is light, then $i_b' = i_b$.
        \item Otherwise,  $i_b'$ is the largest index $i_b+1 \leq i \leq d$ such that  $L_{i_b + 1}, L_{i_b+2}, \ldots, L_{i}$ are heavy.
    \end{itemize}
\end{itemize}

\begin{lemma}\label{lem:diam-reduction}
Given that $L_{i'}$ is heavy for all $i_a < i' < i_b$, we have $i_b' - i_a' =  O\left(\sqrt{n / \Delta}\right)$.
\end{lemma}
\begin{proof}
By the assumption of the lemma and the definition of $i_a'$ and $i_b'$, all of $L_{i_a'}, L_{i_a' + 1}, \ldots, L_{i_b'}$ must be heavy,  except that $L_{i_a}$ and $L_{i_b}$ might be light. Among all heavy  $L_i$, all of them are heavy due to the reason $\vol(L_i) \geq \sqrt{\Delta n}$, except that at most $18$ of them are heavy only because $\sum_{v \in L_i} f(v) \geq n/18$, as $\sum_{v \in V_{i^\ast} \setminus A_{i^\ast}} f(v) \leq n$.

Hence  we have $\vol(L_{i_a'} \cup L_{i_a' + 1} \cup \cdots \cup L_{i_b'}) \geq (i_b' - i_a' - 2 - 18) \sqrt{ \Delta n}$. Since the original graph $G=(V,E)$ is $H$-minor free, we have $|E| = O(n)$. Therefore, $(i_b' - i_a' - 20) \sqrt{\Delta n} = O(n)$, so $i_b' - i_a' =  O\left(\sqrt{n / \Delta}\right)$.
\end{proof}

We define the set $R_2$  as follows.
\[R_2 = \begin{cases}
L_{i_a' - 1} \cup L_{i_b' + 1}, & \text{if $i_a' > 0$ and $i_b' < d$.}\\
L_{i_a' - 1}, & \text{if $i_a' > 0$ and $i_b' = d$.}\\
L_{i_b' + 1}, & \text{if $i_a' = 0$ and $i_b' < d$.}\\
\emptyset, & \text{if $i_a' = 0$ and $i_b' = d$.}
\end{cases}\]
Since both $L_{i_a' - 1}$ and $L_{i_b' + 1}$ are light by our choice of $i_a'$ and $i_b'$, we have $\vol(R_2) = O(\sqrt{\Delta n})$.

 We define the graph $G_s$ as follows.
\begin{itemize}
    \item Start with $H_{i^\ast}[V_{i^\ast} \setminus A_{i^\ast}]$. 
    \item For each connected component of $H_{i^\ast}[V_{i^\ast} \setminus A_{i^\ast}]$, contract all vertices in $L_0, L_1, \ldots, L_{i_a'-1}$ into a vertex $r^\ast$.
    \item Remove all vertices in $L_{i_b' + 1}, L_{i_b' + 2}, \ldots, L_d$.
\end{itemize}

By \cref{lem:diam-reduction}, the diameter of $G_s$ is $O\left(\sqrt{n / \Delta}\right)$. Furthermore,  $G_s$  is a minor of the graph $H_{i^\ast}[V_{i^\ast} \setminus A_{i^\ast}]$ which is apex-free $h$-almost-embeddable on surface $\Sigma$. By \cref{lem:treewidth-diameter-rel}, the treewidth of $G_s$ is also $\tw(G_s) = O\left(\sqrt{n / \Delta}\right)$. Apply \cref{lem:treewidth-aux2} with the following weight function.
\begin{itemize}
    \item $w(r^\ast) = 0$ for each connected component of $H_{i^\ast}[V_{i^\ast} \setminus A_{i^\ast}]$.
    \item $w(v) = f(v)$ for each $v \in L_{i_a'} \cup L_{i_a' + 1} \cup \cdots \cup L_{i_b'}$.
\end{itemize}

Observe that $w(v)  = f(v) \geq 1$ for each $v \in L_{i_a'} \cup L_{i_a' + 1} \cup \cdots \cup L_{i_b'}$ and the summation of $w(v)$ over all vertices in $G_s$ is $\sum_{v \in L_{i_a'} \cup L_{i_a' + 1} \cup \cdots \cup L_{i_b'}} f(v) \leq \sum_{v \in V_{i^\ast}} f(v) \leq n$.
From \cref{lem:treewidth-aux2} we obtain a set $R_3 \subseteq L_{i_a'} \cup L_{i_a' + 1} \cup \cdots \cup L_{i_b'}$ meeting the following conditions.
\begin{itemize}
    \item $|R_3| = O\left(\sqrt{n / \Delta}\right)$, so $\vol(R_3) = O(\sqrt{\Delta n})$ in $G$.
    \item In $H_{i^\ast}[V_{i^\ast} \setminus A_{i^\ast}]$, each connected component of $L_{i_a'} \cup L_{i_a' + 1} \cup \cdots \cup L_{i_b'} \setminus R_3$ has at most $(1/2)n$ vertices.
    %This is true not only when the underlying graph is $H_{i^\ast}[V_{i^\ast} \setminus A_{i^\ast}]$, but it is also true when the underlying graph is $H_{i^\ast}$ or $G$.
\end{itemize}

\begin{lemma}\label{lem:main-flap-size}
Suppose that $L_{i'}$ is heavy for all $i_a < i' < i_b$.
Then $R = A_{i^\ast} \cup R_1 \cup R_2 \cup R_3$ satisfies have $\vol(R) = O(\sqrt{\Delta n})$  and each $R$-flap has at most $(2/3)n$ vertices.
\end{lemma}
\begin{proof} The fact that $\vol(R) = O(\sqrt{\Delta n})$ is clear. Let $S$ be any $R$-flap. If $S \cap V_{i^\ast} = \emptyset$, then we have $S \subseteq S'$ for some $V_{i^\ast}$-flap $S'$, and we know that $|S'| \leq |S| \leq (1/2)n < (2/3)n$ by   \cref{lem:cliquesum-aux}.

Suppose that $S \cap V_{i^\ast} \neq \emptyset$. Let $u_0 \in S \cap V_{i^\ast}$. Then $u_0 \in L_i$ for some $i$. % \notin \{i_a'-1, i_b' +1\}$.
By our choice of $R_2$, there are three cases: $1 \leq i \leq i_a' - 2$, $i_a' \leq i \leq i_b'$, and $i_b'+2 \leq i \leq d$.

For the case $1 \leq i \leq i_a' - 2$,  \cref{lem:easy-case-aux} shows that any $(A_{i^\ast} \cup L_{i_a'-1})$-flap $S'$ that contains $u_0$ satisfies 
\[|S'| \leq f(L_{1} \cup L_{2} \cup \cdots \cup L_{i_a'-1}) \leq
f(L_{1} \cup L_{2} \cup \cdots \cup L_{i_a-1}) \leq
(1/3)f(V_{i^\ast} \setminus A_{i^\ast}) \leq (1/3)n < (2/3)n.\] 
Since $A_{i^\ast} \cup L_{i_a'-1} \subseteq R$, any $R$-flap $S$ must be a subset of some $(A_{i^\ast} \cup L_{i_a'-1})$-flap $S'$. Hence we also have $|S| \leq (2/3)n$. For the case $i_b'+2 \leq i \leq d$, a similar analysis also shows that $|S| \leq (2/3)n$.

For the rest of the proof, we assume that $i_a' \leq i \leq i_b'$. Then $u_0 \in W$ for some connected component $W$ of the subgraph of $H_{i^\ast}[V_{i^\ast} \setminus A_{i^\ast}]$ induced by $L_{i_a'} \cup L_{i_a' + 1} \cup \cdots \cup L_{i_b'} \setminus R_3$. Here $W$ can be seen as a vertex subset of both $H_{i^\ast}[V_{i^\ast} \setminus A_{i^\ast}]$ and $G$.
In the graph $G$, let $U$ be the union of $W$ and all $V_{i^\ast}$-flap $S'$ that is adjacent to $W$. 
It is clear that $S \subseteq U$. % note: the proof is of this fact is tedious to write down formally.

For any $V_{i^\ast}$-flap $S'$, the set $C_{S'}$ of vertices in $V_{i^\ast} \setminus A_{i^\ast}$ adjacent to $S'$ is a clique of size at most $h$. Given that $S'$  is adjacent to $W$, we have $C_{S'} \subseteq W \cup R_3 \cup L_{i_a'-1} \cup L_{i_b'+1}$. Therefore, for each $u \in S'$, we have $f_{u \rightarrow v} > 0$ only if $v \in W \cup R_3 \cup L_{i_a'-1} \cup L_{i_b'+1} \setminus R_1$. 
The reason that $R_1$ can be excluded is that $C_{S'}$ is not a subset of $R_1$, so \cref{lem:heavy} ensures that $f_{u \rightarrow v} = 0$ for all $v \in R_1$.

Similar to the proof of \cref{lem:easy-case-aux}, we can upper bound $|U|$ by
\[|U| \leq f(W) + f(R_3 \setminus R_1) + f(L_{i_a'-1}) + f(L_{i_b'+1}).\] % the proof for this part is tedious to write down formally.

By our choice of $R_3$, we have $f(W) \leq (1/2)n$. To upper bound $f(R_3 \setminus R_1)$,  recall that $|R_3 \setminus R_1| \leq |R_3| =  O\left(\sqrt{n / \Delta}\right)$ and each $v  \notin R_1$ has $f(v) < \eps \sqrt{\Delta n}$, so  $f(R_3 \setminus R_1)  = O(\eps n) \leq (1/18)n$ by selecting $\eps$ to be a sufficiently small constant. We also have $f(L_{i_a'-1}) \leq (1/18)n$  and $f(L_{i_b'+1}) \leq (1/18)n$ as $L_{i_a'-1}$ and $L_{i_b'+1}$ are light. To sum up,  $|S| \leq |U| \leq (1/2)n + 3 \cdot (1/18)n = (2/3)n$.
\end{proof}

Combining \cref{lem:reduction1,lem:easy-case,lem:main-flap-size} we conclude \cref{thm:edge-separator}.

%\section*{Declaration} \begin{description} \item[Competing interest:] The authors declare that they have no known competing financial interests or personal relationships that could have appeared to influence this research work. \item[Funding:] Yi-Jun Chang was supported by the grant (NUS ODPRT Grant No.~R-252-000-C04-133) of the NUS Presidential Young Professorship. Hsin-Hao Su was supported by NSF Grant No.~CCF-2008422. \item[Data Availability:] This research work is purely theoretical and it does not involve any data. \end{description}

%\newpage
{\small
\bibliographystyle{alpha}
\bibliography{references}
}
%\appendix 
%\clearpage
%\newpage

\end{document}